\documentclass[11pt,oneside]{article}

\usepackage{epsfig,lscape}
\usepackage{amssymb,amsfonts,amsmath}
\usepackage{rotating}

\usepackage{amsthm}

\usepackage{color}
\usepackage[dvipsnames]{xcolor}
\usepackage{multirow}

\usepackage{float}
\usepackage[hypertexnames=false]{hyperref}

\usepackage{ulem}

\def\me{\mathrm e}

\def\dif{\mathrm d}

\def\N{\mathrm{N}}

\def\T{ {\mathrm{\scriptscriptstyle T}} }

\newcommand\mytext[1]{\text{\tiny{#1}}}

\def\over1{\overline{1}}

\newtheorem{pro}{Proposition}

\theoremstyle{definition}

\theoremstyle{definition}

\usepackage[medium,compact]{titlesec}

\titlespacing*{\section} {0pt}{1.5ex}{1ex}
\titlespacing*{\subsection} {0pt}{1.5ex}{1ex}
\titlespacing*{\subsubsection} {0pt}{1ex}{1ex}

\allowdisplaybreaks 

\begin{document}

\setlength{\abovedisplayskip}{3.3pt}
\setlength{\belowdisplayskip}{3.3pt}


\renewcommand*{\thefootnote}{\fnsymbol{footnote}}

\newcommand{\wx}{\textcolor{orange}}

\newcommand{\zt}{\textcolor{blue}}

\begin{titlepage}

\begin{center}
{\Large Re-examining and calibrating weighted survival analysis for causal inference}

\vspace{.1in} 
Wenfu Xu\footnotemark[1], Yi Zhang\footnotemark[2], Tobias Gerhard\footnotemark[3], Zhiqiang Tan\footnotemark[2]\\
\footnotetext[1]{College of Economics, China Jiliang University}
\footnotetext[2]{Department of Statistics, Rutgers University}
\footnotetext[3]{Institute for Health, Health Care Policy, and Aging Research, Rutgers University}

\vspace{.1in}
\today
\end{center}

\paragraph{Abstract.}

Causal inference with time-to-event outcomes is fundamental in various scientific studies.
In a static setup with fitted propensity scores, weighted Kaplan–Meier estimation for survival probabilities and weighted Breslow–Peto estimation
for hazard ratios have been widely used, but their statistical properties have been overlooked or studied only to a limited extent.
We re-examine the weighted Kaplan–Meier method
by formally linking it with the general framework of augmented inverse probability weighted estimation including both point and variance estimation.
Furthermore, to address limitations of existing weighted methods for survival analysis, we develop new methods and associated theory
through calibrated estimation in both low-dimensional and high-dimensional settings.
We present a simulation study and an empirical application on the effectiveness of adjunctive psychotropic treatments for patients with schizophrenia.
The calibrated methods yield coverage proportions closer to target ones in the simulation study,
and produce shorter confidence intervals in both simulation and empirical studies.


\end{titlepage}

\renewcommand*{\thefootnote}{\arabic{footnote}}

\section{Introduction} \label{sec:intro}

Causal inference for time-to-event outcomes is of interest in various scientific fields including medical and social sciences.
In observational studies, a primary objective is to assess the causal effect of a treatment on the distribution of survival times.
Commonly studied causal estimands include treatment-specific survival probabilities and hazard probabilities,
as well as contrasts such as survival probability differences and hazard ratios.
Under standard identification assumptions, including no unmeasured confounding and non-informative censoring, these estimands can be identified from observed data (Robins \& Rotnitzky 1992; Hern\'{a}n \& Robins 2020). The main challenge beyond identification lies in
constructing estimators that are robust to model misspecification and admit valid large-sample inference (e.g., confidence intervals).

Consider the static setup where a sample of observations are obtained on a binary treatment $A$, right-censored survival time $Y$ and non-censoring indicator $\Delta$,
and a vector of baseline covariates $X$, but no additional time-varying covariates are measured.
For this setup, there seems to be two strands of causal inference literatures which are arguably disconnected.
On one hand,
both weighted Kaplan–Meier (KM) for survival probabilities and weighted Breslow--Peto estimation\footnote{\label{ft:wBP}This method is a weighted version of Breslow's (1974) and Peto's (1972) modification of Cox's (1972) maximum partial likelihood estimator
to handle ties in survival times. See Section~\ref{sec:existing-wBP} for more details.} for hazard ratios as in Cox's proportional hazards models
have been proposed through general ideas of inverse probability weighted (IPW) estimation (Cole \& Hern\'{a}n 2004; Joffe et al. 2004; Xie \& Liu 2005).
The inverse weights are fitted values from a propensity score (PS) model for the probability of $A$ given $X$.
These methods have been widely used in practice, available from the popular R package \texttt{survival} (Therneau 2020),
but their statistical properties have been overlooked or studied only to a limited extent.
On the other hand,
semi-parametric theory and augmented IPW estimation have been developed in more general, dynamic setups with both baseline and time-varying covariates
(Robins \& Rotnitzky 1992; Hubbard et al. 2000; Bai et al. 2013).
While such general theory and methods are applicable to the static setup,
limited work has been done to specifically study the application and implications,
or to carefully connect the general theory with weighted Kaplan–Meier and Breslow--Peto estimation in practical use.

The purpose of our work is twofold. First, we re-examine the weighted Kaplan–Meier method
by formally linking it with the general framework of augmented IPW estimation including both point and variance estimation.
Second, to address limitations of existing weighted methods for survival analysis, we develop new methods and associated theory
through calibrated (CAL) estimation (Tan 2020; Ghosh \& Tan 2022).
For simplicity, we focus on the discrete-time setting, thereby avoiding technical issues in asymptotic theory related to continuous time.
Our main contributions are summarized as follows.
See Table \ref{tb:notation} for the notation of various estimators. \vspace{-.1in}
\begin{itemize}\addtolength{\itemsep}{-.1in}
\item We show that the weighted Kaplan–Meier estimator $\hat S_{1k,\mytext{wKM}}$, can be recast as an augmented IPW estimator
provided that 
(a) the non-censoring probabilities and survival probabilities are each modeled as constant with respect to the covariate vector $X$, leading to working models termed CCP and CSP respectively, and 
(b) the working-model parameters are estimated in a particular manner which differs from maximum likelihood (ML) estimation but conforms to calibrated estimation (Proposition \ref{pro:weighted-KM}).
Then the weighted Kaplan–Meier estimator is doubly robust, i.e., being consistent
if either the PS and CCP models are correct or the CSP model is correct,
and hence the consistency is \textit{not} ensured by only correct specification of PS model.
Moreover, we show that the model-robust variance estimator in R package \texttt{survival} algebraically coincides with the standard,
sample-variance based estimator from augmented IPW theory
under a mild normalization condition, and is nonparametrically consistent provided that
sampling variation is \textit{ignored} in the estimated propensity score (Proposition \ref{pro:weighted-KM-Vr}).
To our knowledge, all these results have not been appreciated in the literature.

\item We generalize the standard augmented IPW estimator by introducing an outcome regression (OR) model
for the conditional survival probability, separately from the CSP model, and then propose two calibrated augmented IPW estimators,
$\hat S_{1k,\mytext{CAL}}$ and $\hat S_{1k,\mytext{CAL,lin}}$, for the treatment-specific (marginal) survival probability.
While both estimators preserve pointwise double robustness, $\hat S_{1k,\mytext{CAL}}$ admits a consistent variance estimator as long as
the PS model is correctly specified and $\hat S_{1k,\mytext{CAL,lin}}$ admits a nonparametrically consistent variance estimator even when
the PS model is misspecified (Propositions \ref{pro:AIPW-CAL} and \ref{pro:AIPW-CAL-lin}).
Hence the method based on $\hat S_{1k,\mytext{CAL}}$ enables valid inference under correct PS and CCP models,
whereas that based on $\hat S_{1k,\mytext{CAL,lin}}$ enables doubly robust inference under either correct PS and CCP models or correct CSP model.
Furthermore, we formulate a hazard-ratio parameter that is targeted by
the weighted Brewlow--Peto estimator $\hat\theta_{\mytext{wBP}}$, and
propose two calibrated augmented IPW estimators, $\hat \theta_{\mytext{CAL}}$ and $\hat \theta_{\mytext{CAl,lin}}$ by incorporating risk-set OR models.
Both estimators achieve pointwise double robustness and lead to valid inference under correct PS and CCP models (Proposition \ref{pro:CAL-theta}),
and $\hat \theta_{\mytext{CAl,lin}}$ also enables valid inference under correct CSP and risk-set OR models.

\item We extend calibrated augmented IPW estimation to high-dimensional settings, where the number of covariate terms in PS and OR models
is close to or exceeds the sample size.
For survival probabilities, the regularized calibrated estimator $\hat S_{1k,\mytext{RCAL}}$
involves directly replacing unregularized estimation by Lasso regularized estimation when fitting PS and OR models.
In contrast, a direct extension for hazard ratio estimation requires Lasso regularized estimation at each time period,
which is computationally costly when cross validation is used for selecting Lasso tuning parameters.
To overcome this difficulty, we develop a novel regularized calibrated estimator, $\hat \theta_{\mytext{RCa}}$,
by exploiting a linearization technique to combine OR-based augmentation terms across time periods and then apply regularized estimation with OR models.
The resulting method is computationally tractable while achieving similar theoretical properties as the calibrated estimator in low-dimensional settings (Proposition \ref{pro:AIPW}).

\item We conduct a simulation study to evaluate the proposed methods and compare with existing ones in low-dimensional and high-dimensional settings.
We also present an empirical application to re-analyze the data used in Stroup et al. (2019) for
studying the effectiveness of adjunctive psychotropic treatments for patients with schizophrenia.
The calibrated methods yield coverage proportions closer to target ones in the simulation study,
and produce shorter confidence intervals (i.e., smaller SEs) in both simulation and empirical studies.
\end{itemize} \vspace{-.2in}

\begin{table} [H]
\caption{Notation of estimators studied} \label{tb:notation} \vspace{-.1in}
\begin{center}
\small
\renewcommand{\arraystretch}{0.9}
\resizebox{0.87\textwidth}{!}{\begin{tabular}{lll}
\hline
 Estimator of $\gamma$ for PS
 & \hspace{1in}Survival probability & Hazard ratio \\
  \hline
 $\begin{array}{l}
\text{General} \\
\hat{\gamma} \\
\\
\end{array}$
&
$\begin{array}{ll}
\hat{S}_{1k, \mytext{wKM}} & =  \underset{\text{(AIPW)}}{\hat{S}_{1k}(\hat{\gamma}, \hat{\rho}_{1, 1:k, \mytext{CAL}}, \hat{\eta}_{1, 1:k, \mytext{CAL}})} = \underset{\text{(ICE)}}{\hat{S}_{1k}(\hat{\eta}_{1, 1:k, \mytext{CAL}})} \\
 \hat{S}_{1k, \mytext{wKM}}  & = \underset{\text{(IPW)}}{\hat{S}_{1k}(\hat{\gamma}, \hat{\rho}_{1, 1:k, \mytext{CAL}})}~\text{if}~\sum_{i=1}^n A_i\hat{\pi}_i^{-1} = n
\end{array}$
&
$\begin{array}{l}
\hat{\theta}_{\mytext{wBP}}~(\text{Eq}~\ref{eq:weighted-BP})
\end{array} $ \\
\cline{2-3}
$\begin{array}{l}
\text{Calibrated} \\
\hat{\gamma}_{1, \mytext{CAL}} \\
\\
\end{array}$
&
$\begin{array}{ll}
\hat{S}_{1k, \mytext{CAL}} & \overset{\text{def}}{=} \hat{S}_{1k}(\hat{\gamma}_{1, \mytext{CAL}}, \hat{\rho}_{1, 1:k, \mytext{CAL}}, \hat{\eta}_{1, 1:k, \mytext{CAL}}, \hat{\alpha}_{1k, \mytext{CAL}}) \\
 \hat{S}_{1k, \mytext{CAL, lin}} & \overset{\text{def}}{=} \hat{S}_{1k}(\hat{\gamma}_{1, \mytext{CAL}}, \hat{\rho}_{1, 1:k, \mytext{CAL}}, \hat{\eta}_{1, 1:k, \mytext{CAL}}, \hat{\alpha}_{1k, \mytext{CAL, lin}}) \\
& = \hat{S}_{1k, \mytext{wKM}}
\end{array}$
&
$\begin{array}{l}
 \hat{\theta}_{\mytext{CAL}}~(\text{Eq}~\ref{eq:CAL-theta}) \\
\hat{\theta}_{\mytext{CAL, lin}} = \hat{\theta}_{\mytext{wBP}}
\end{array}$ \\
\cline{2-3}
$\begin{array}{l}
\text{Regularized} \\
\text{calibrated} \\
\hat{\gamma}_{1, \mytext{RCAL}}
\end{array}$
&
$\begin{array}{l}
 \hat{S}_{1k, \mytext{RCAL}} \overset{\text{def}}{=} \hat{S}_{1k}(\hat{\gamma}_{1, \mytext{RCAL}}, \hat{\rho}_{1, 1:k, \mytext{CAL}}, \hat{\eta}_{1, 1:k, \mytext{CAL}}, \hat{\alpha}_{1k, \mytext{RCAL}})~(\text{Eq}~\ref{eq:AIPW-CAL}) \\
 \hat{S}_{1k, \mytext{RCAL, lin}} \overset{\text{def}}{=} \hat{S}_{1k}(\hat{\gamma}_{1, \mytext{RCAL}}, \hat{\rho}_{1, 1:k, \mytext{CAL}}, \hat{\eta}_{1, 1:k, \mytext{CAL}}, \hat{\alpha}_{1k, \mytext{RCAL, lin}}) \\
\hat{S}_{1k, \mytext{RCw}} = \hat{S}_{1k, \mytext{wKM}}(\hat{\gamma}_{1, \mytext{RCAL}})~(\text{Eq}~\ref{eq:weighted-KM})
\end{array}$
&
$\begin{array}{l}
 \hat{\theta}_{\mytext{RCAL}} \\
\hat{\theta}_{\mytext{RCAL, lin}} \\
\hat{\theta}_{\mytext{RCw}} = \hat{\theta}_{\mytext{wBP}}(\hat{\gamma}_{1,\mytext{RCAL}})~(\text{Eq}~\ref{eq:RCAL-IPW-theta}) \\
\hat{\theta}_{\mytext{RCa}}~(\text{Eq}~\ref{eq:RCAL-AIPW-theta})
\end{array}$ \\
\hline
\end{tabular}}
\end{center}
\end{table}  \vspace{-.3in}
The remainder of the paper is organized as follows. Section~\ref{sec:setup} introduces the setup and identification conditions. Section~\ref{sec:existing} reviews existing methods for estimating survival probabilities and hazard ratios.
Section~\ref{sec:re-exam} reviews standard augmented IPW estimation and then re-examines weighted KM estimation.
Section~\ref{sec:CAL-AIPW} introduces calibrated estimation for survival probabilities and for hazard ratios.
Section~\ref{sec:high-dimension} extends the calibrated methods to high-dimensional settings. Section~\ref{sec:simulation} presents a simulation study, and Section~\ref{sec:empirical} provides a real-data application.

\section{Setup} \label{sec:setup}

Suppose that the observed data are $\{(Y_i, \Delta_i, A_i, X_i): i=1,\ldots,n\}$ from $n$ individuals,
where $Y_i = \min(U_i,C_i)$, $\Delta_i = 1\{U_i \le C_i\}$, $U_i$ is an (uncensored) event time such as death time, $C_i$ is a censoring time,
$A_i$ is a binary treatment and $X_i$ is a (baseline) covariate vector.
Assume that $\{(U_i,C_i, A_i,X_i): i=1,\ldots,n\}$ are independent and identically distributed copies of $(U,C,A,X)$, and hence
$\{(Y_i,\Delta_i,A_i,X_i): i=1,\ldots,n\}$ are independent and identically distributed copies of $(Y,\Delta,A,X)$ with $Y = \min(U,C)$ and $\Delta = 1\{U \le C\}$.

For causal inference, let $(U^{(0)},  U^{(1)})$ be the potential survival times and $(C^{(0)}, C^{(1)})$ be the potential censoring times.
By consistency, assume that $(U,C)$ is either $(U^{(0)}, C^{(0)})$ or  $(U^{(1)}, C^{(1)})$, depending on $A=0$ or 1.
Throughout, we focus on discrete-time data and
assume that there are discrete values, $0 = u_0 < u_1 < \cdots < u_K < u_{K+1}$, such that  $C^{(a)} \in \{u_0, u_1, \ldots, u_K\}$ and $U^{(a)} \in \{u_1, \ldots, u_K, u_{K+1}\}$ for $a=0,1$.
By this assumption, $P(U^{(a)} > u_0)=1$.
It is commonly of interest to estimate the treatment-specific survival probabilities $S_{ak} = P ( U^{(a)} >u_k )$,
hazard probabilities $q_{ak} = P ( U^{(a)} = u_k  | U^{(a)} \ge  u_k )$,
or other population quantities derived from these probabilities.
Causal effects can be defined as some
contrasts of these population quantities between treatment $a=0$ and 1, for example, the survival probability differences
$S_{1k} - S_{0k}$ or the hazard probability ratios $q_{1k} / q_{0k}$.

To ensure point identification, we assume throughout that
no unmeasured confounding (NUC) and non-informative censoring (NIC) hold:
\begin{align}
&  A \perp U^{(a)} | X, \quad a=0,1,  \label{eq:no-confounding} \\
& U^{(a)} \perp C^{(a)} | A=a, X, \quad a=0,1,  \label{eq:non-informative}
\end{align}
i.e., $ A $ and $ U^{(a)}$ are independent given $X$, and $U^{(a)}$ and $C^{(a)}$ are independent given $A=a$ and $X$ for $a=0,1$.
Under Assumptions~\eqref{eq:no-confounding}--\eqref{eq:non-informative},
the conditional survival probabilities $ P ( U^{(a)} > u_k | U^{(a)} \ge u_k, X)$ and hence also
$ P ( U^{(a)} > u_k |X)$ and $P( U^{(a)}>u_k)$, can be identified from the observed-data distribution for $k=1,\ldots,K$
(see footnote~\ref{ft:identification}).

\section{Existing weighted survival analysis} \label{sec:existing}

We describe existing methods of weighted survival analysis. The weights are defined from the fitted values in
a regression model $\pi(X;\gamma)$ for the propensity score $\pi^*(X) = P(A=1|X)$ (Rosenbaum \& Rubin 1983).
For an estimator $\hat\gamma$ of $\gamma$, the fitted propensity score is $\pi(X;\hat\gamma)$.
For example, consider logistic regression
\begin{align}
 \pi(X;\gamma) = \frac{ \exp ( \gamma^\T f(X) )} { 1+ \exp ( \gamma^\T f(X) )}, \label{eq:PS-model}
\end{align}
where $f(X)$ is a vector of known covariate functions (including an intercept) such as main effects and interactions, and
$\gamma=(\gamma_0,\gamma_1,\ldots,\gamma_p)^\T$ is a vector of unknown coefficients including intercept $\gamma_0$.
In general, let $\hat\gamma$ be an estimator of $\gamma$ which is consistent for the true value $\gamma^*$ if the model is correctly specified.
As a specific choice, the maximum likelihood (ML) estimator of $\gamma$, denoted as $\hat\gamma_{\mytext{ML}}$, is a solution to
\begin{align*}
 0 = \frac{1}{n} \sum_{i=1}^n \{ A_i - \pi(X_i; \gamma)\} f(X_i).
\end{align*}
The fitted propensity score is then $\pi (X; \hat\gamma_{\mytext{ML}})$.

\subsection{Weighted survival probabilities} \label{sec:weighted-KM}

Weighted Kaplan--Meier estimation has been proposed by incorporating inverse propensity-score weighting with Kaplan--Meier (KM) estimation (Kaplan \& Meier, 1958)
to estimate survival probabilities $S_{ak}$ (Cole \& Hernan 2004; Xie \& Liu 2005).
For estimating $S_{1k}$, $k=1,\ldots,K$, the hazard probabilities, $q_{1j}$, $j=1,\ldots,k$, are estimated as
\begin{align}
 \hat q_{1j} = \frac{\sum_{i \in J_j } A_i \hat w_{1i} }{ \sum_{i\in I_j } A_i \hat w_{1i} } ,  \label{eq:q-hat}
\end{align}
where $\hat w_{1i} = \pi (X_i; \hat\gamma)^{-1}$,
$I_j = \{1\le i\le n:  Y_i \ge u_j \}$ and $J_j = \{ 1\le i \le n: Y_i=u_j, \Delta_i=1\}$,
which represent, respectively, the risk set and event set at time $u_j$ (hence $I_j\setminus J_j$ represents the survival set at $u_j$).
The weighted KM estimator of $S_{1k}$ is
\begin{align}
\hat S_{1k,\mytext{wKM}} = \prod_{j=1}^k  ( 1- \hat q_{1j} ) . \label{eq:weighted-KM}
\end{align}
The dependency of $\hat q_{1j}$ and $\hat S_{1k,\mytext{wKM}}$ on $\hat\gamma$ is suppressed in the notation.
See Section~\ref{sec:weighted-KM-re} for a careful discussion of consistency of $\hat S_{1k,\mytext{wKM}}$.
For inference (e.g., confidence intervals), several variance estimators are available.
The variance estimator derived in Xie \& Liu (2005) is
\begin{align}
\hat V_{\mytext{b}} ( \hat S_{1k,\mytext{wKM}} ) =  \hat S_{1k,\mytext{wKM}} ^2 \sum_{j=1}^k
\frac{\sum_{i \in I_j } A_i \hat w_{1i}^2 }{ (\sum_{i\in I_j } A_i \hat w_{1i})^2 } \frac{\hat q_{1j} }{1- \hat q_{1j}} ,  \label{eq:weighted-KM-Vb}
\end{align}
which will be called a model-based variance estimator as indicated by Proposition~\ref{pro:weighted-KM-Vb} later.
The \texttt{survfit()} function 
in R package \texttt{survival} (Therneau 2020) provides two variance estimators for weighted KM estimation.
By our understanding and numerical evaluation, the model-based variance estimator in \texttt{survival}  is
\begin{align*}
\hat V_{\mytext{b0}} ( \hat S_{1k,\mytext{wKM}} ) =  \hat S_{1k,\mytext{wKM}} ^2 \sum_{j=1}^k
\frac{1}{ \sum_{i\in I_j } A_i \hat w_{1i} } \frac{\hat q_{1j} }{1- \hat q_{1j}} ,
\end{align*}
which is appropriate only for frequency weighted estimation, not inverse probability weighted (IPW) estimation (Therneau \& Grambsch 2000).
The model-robust variance estimator is
\begin{align}
\hat V_{\mytext{r}} ( \hat S_{1k,\mytext{wKM}} ) = \frac{\hat S_{1k,\mytext{wKM}} ^2 }{n^2}
\sum_{i=1}^n \hat\varphi_{1ki,\mytext{wKM}} ^2 , \label{eq:weighted-KM-Vr}
\end{align}
where
\begin{align*}
 \hat \varphi_{1ki,\mytext{wKM}} = n \sum_{j=1}^k \frac{ - A_i \hat w_{1i} } { \sum_{\tau \in I_j\setminus J_j } A_\tau \hat w_{1\tau} }
 1\{i\in I_j\} ( 1\{i\in J_j\} - \hat q_{1j})
\end{align*}
The additional factor of $n$ in $\hat \varphi_{1ki,\mytext{wKM}}$ is introduced to facilitate comparison with augmented IPW estimation in Section~\ref{sec:AIPW}.
For completeness, we directly derive the model-robust variance estimator through Taylor expansions in Supplement Section \ref{sec:technical-detail},
\textit{provided} that $\hat\gamma$ (hence each weight $\hat w_{1i}$) is treated as data-independent.
See Section~\ref{sec:weighted-KM-re} for a further discussion of consistency of $\hat V_{\mytext{r}} ( \hat S_{1k,\mytext{wKM}} )$.

For estimating the treatment-0 survival probabilities $S_{0k}$, $k=1,\ldots,K$, the hazard probabilities, $q_{0j}$, $j=1,\ldots,k$, are estimated as
\begin{align*}
 \hat q_{0j} = \frac{\sum_{i \in J_j } (1-A_i) \hat w_{0i} }{ \sum_{i\in I_j } (1-A_i) \hat w_{0i} } ,
\end{align*}
where $\hat w_{0i} = (1-\pi (X_i; \hat\gamma))^{-1} $.
The weighted KM estimator of $S_{0k}$ is
$\hat S_{0k,\mytext{wKM}} = \prod_{j=1}^k  ( 1- \hat q_{0j} )$.
Variance estimators can be obtained similarly as above.

\subsection{Weighted hazard ratios} \label{sec:existing-wBP}

Inverse propensity-score weighting has also been used in fitting Cox's (1972) proportional hazards models
to estimate hazard ratios between potential survival outcomes (Cole \& Hernan 2004; Joffe et al.~2004).
A popular, simple method is to fit a proportional hazards model with the treatment variable as the only regression term.
For example, the \texttt{coxph()} function in R package \texttt{survival} implements
weighted proportional hazard regression, including
Breslow's (1974) and Peto's (1972) modification of Cox's maximum partial likelihood estimator with option \texttt{ties="breslow"}.
In the discrete-time setting, the weighted Breslow--Peto (BP) estimator $\hat\theta_{\mytext{wBP}}$ is defined as a solution to
\begin{align*}
 0 = \sum_{k=1}^K \sum_{i \in J_k} \hat w_i
 \left( A_i - \frac{\sum_{\tau\in I_k} \hat w_\tau A_\tau \me^{\theta A_\tau}} {\sum_{\tau\in I_k} \hat w_\tau \me^{\theta A_\tau} } \right),
\end{align*}
where $\hat w_i = A_i \hat w_{1i} + (1-A_i) \hat w_{0i}$, i.e., $\hat w_i = \hat w_{1i}$ if $A_i=1$ or $\hat w_{0i}$ if $A_i=0$.
To help understanding, the preceding estimating equation can be rewritten as
\begin{align*}
 0 = \sum_{k=1}^K \left\{ \hat W_{1k} \hat q_{1k} - \frac{\hat W_{1k} \me^\theta}
 { \hat W_{1k} \me^\theta + \hat W_{0k} } (\hat W_{1k} \hat q_{1k}+ \hat W_{0k}\hat q_{0k})  \right\}
\end{align*}
or equivalently
\begin{align}
 0 = \sum_{k=1}^K  \frac{ \hat W_{1k} \hat W_{0k} }
 { \hat W_{1k} \me^\theta + \hat W_{0k} } ( \hat q_{1k} -  \hat q_{0k} \me^\theta )  ,  \label{eq:weighted-BP}
\end{align}
where $\hat W_{1k} = n^{-1} \sum_{i\in I_k} A_i \hat w_{1i}$, $\hat W_{0k} = n^{-1} \sum_{i\in I_k} (1-A_i) \hat w_{0i}$,
and $\hat q_{1k}$ and $\hat q_{0k}$ are defined in Section~\ref{sec:weighted-KM}.\footnote{\label{ft:invariant-q}
While $\hat q_{1k}$ and $\hat q_{0k}$ are invariant to rescaling of $\hat w_{1i}$'s and $\hat w_{0i}$'s,
equation \eqref{eq:weighted-BP} is affected by rescaling of $\hat w_{1i}$'s and $\hat w_{0i}$'s.
For example, if $\hat w_{1i} \equiv (n_1/n)^{-1}$ and $\hat w_{0i} \equiv ((n-n_1)/n)^{-1}$, with $n_1=\sum_{i=1}^n A_i$,
then resetting $\hat w_{1i} \equiv 1$ and $\hat w_{0i} \equiv 1$ in \eqref{eq:weighted-BP} may give a different solution for $\hat\theta_{\mytext{wBP}}$.}
From the expression \eqref{eq:weighted-BP}, $\hat\theta_{\mytext{wBP}}$ is expected to be a consistent estimator of $\theta$ in the following model of
proportional hazard probabilities (PHP):
\begin{align}
 q_{1k} = \me^\theta q_{0k} , \quad k=1,\ldots,K.   \label{eq:model-hazard-prob}
\end{align}
where $q_{ak} = P ( U^{(a)} = u_k  | U^{(a)} \ge  u_k )$ is a hazard probability.
If model \eqref{eq:model-hazard-prob} is misspecified, then $\hat\theta_{\mytext{wBP}}$ serves as a scalar summary of
log hazard probability ratios, $\log(q_{1k} / q_{0k})$ at individual times $u_k$, $k=1,\ldots,K$.
For completeness, we mention that a model of proportional hazard odds can also be considered:
$ q_{1k}/(1-q_{1k}) = \me^\vartheta q_{0k}/(1-q_{0k})$ , $k=1,\ldots,K$.
See Tan (2022, 2023) for a discussion of discrete-time survival models and associated estimation.

For inference, the R function \texttt{coxph()} from R package \texttt{survival} also provides two variance estimators for $\hat\theta_{\mytext{wBP}}$.
The model-based variance estimator is $\hat V_{\mytext{b0}} ( \hat\theta_{\mytext{wBP}})=
n^{-1} \hat H (\hat\theta_{\mytext{wBP}})^{-1}$, where
\begin{align*}
\hat H (\theta) = \sum_{k=1}^K \frac{ \hat W_{1k} \hat W_{0k} \me^\theta} { ( \hat W_{1k} \me^\theta + \hat W_{0k} )^2 }
 ( \hat W_{1k} \hat q_{1k}+ \hat W_{0k}\hat q_{0k} ) .
\end{align*}
We caution against the use of this variance estimator for two reasons.
First, $\hat V_{\mytext{b0}} ( \hat\theta_{\mytext{wBP}})$, like $\hat V_{\mytext{b0}} ( \hat S_{1k,\mytext{wKM}} )$, is not appropriate for IPW estimation.
Second, even for frequency weighted estimation,
$\hat V_{\mytext{b0}} ( \hat\theta_{\mytext{wBP}})$ tends to be asymptotically biased and conservative in the discrete-time setting (i.e., allowing tied event times)
as shown in Tan (2022, 2023).
The model-robust variance estimator is
\begin{align}
\hat V_{\mytext{r}} ( \hat\theta_{\mytext{wBP}})=
\hat H (\theta)^{-1} \hat G (\theta)
\hat H (\theta)^{-1} \Big|_{\hat\theta_{\mytext{wBP}}} ,  \label{eq:weighted-BP-Vr}
\end{align}
where
 $\hat G(\theta ) = n^{-2} \sum_{i=1}^n \{\sum_{k=1}^K \hat \varphi_{ki,\mytext{wBP}} (\theta) \}^2 $ with
\begin{align*}
 & \hat \varphi_{ki,\mytext{wBP}} (\theta) =
  \frac{ \hat W_{0k} } { \hat W_{1k} \me^\theta + \hat W_{0k} } A_i ( \hat w_{1i} 1\{i \in J_k\} - \hat q_{1k} 1\{i \in I_k\} ) \\
  & \quad - \frac{ \hat W_{1k} \me^\theta  } { \hat W_{1k} \me^\theta + \hat W_{0k} } (1-A_i) ( \hat w_{0i} 1\{i \in J_k\} - \hat q_{0k} 1\{i \in I_k\} ) \\
  & \quad + \frac{\hat q_{1k} -  \hat q_{0k} \me^\theta }{(\hat W_{1k} \me^\theta + \hat W_{0k})^2 }
 \left\{ \hat W_{0k}^2 ( A_i \hat w_{1i} 1\{i\in I_k\}- \hat W_{1k}) + \me^\theta \hat W_{1k}^2 ( (1-A_i) \hat w_{0i} 1\{i\in I_k\}- \hat W_{0k}) \right\} .
\end{align*}
This variance estimator is similar to an unweighted version in Tan (2022), Proposition 6, and can be verified through Taylor expansions, \textit{provided} that
$\hat\gamma$ is treated as data-independent.

\section{Re-examining weighted survival analysis} \label{sec:re-exam}

\subsection{Review: Augmented IPW estimation} \label{sec:AIPW}

We describe augmented IPW estimation for the treatment-1 survival probability $S_{1k}$ in the setup of Section~\ref{sec:setup}, adapted from
previous works (Robins \& Rotnitzky 1992; Hubbard et al.~2000; Bai et al.~2013).
The method involves estimating two distinct sets of unknown covariate functions. The first set includes the treatment propensity score
$ \pi^*(X) = P( A=1 |X)$ and the non-censoring probabilities $\pi^*_{1j}(X)$, $j=1,\ldots,k$, defined as
\begin{align*}
\pi^*_{1j} (X) & = P ( Y \ge u_j | Y \ge u_{j-1}, (Y,\Delta)\not=(u_{j-1},1), X, A=1 ) \\
& = P ( \text{not censored at $u_{j-1}$ } | \text{ in the risk set and event-free at $u_{j-1}$}, X, A=1 ),
\end{align*}
where the risk set at $u_j$ is $\{ Y \ge u_j\}$.
The second set consists of the conditional survival probabilities $m^*_{1j}(X)$, $j=1,\ldots,k$, defined as
\begin{align*}
 m^*_{1j} (X) &= P( (Y,\Delta)\not=(u_j,1) | Y \ge u_j, X, A=1) \\
& = P ( \text{event-free at $u_j$ } | \text{ in the risk set at $u_j$}, X, A=1 ).
\end{align*}
Under Assumptions~\eqref{eq:no-confounding}--\eqref{eq:non-informative},
identification is achieved such that
$\pi^*_{1j} (X) =P ( C^{(1)} \ge u_j | C^{(1)} \ge u_{j-1}, X, A=1)$, the conditional survival function of $C^{(1)}$ given $(X,A=1)$,
and $m^*_{1j} (X) =  P ( U^{(1)} > u_j | U^{(1)} \ge u_j, X)$, the conditional survival function of $U^{(1)}$ given $X$.
\enlargethispage{\baselineskip}
\footnote{\label{ft:identification}
For identification,
the second identity can be shown as follows (and similarly the first identity can be shown): $  m^*_{1j} (X) =  P( (Y,\Delta)\not=(u_j,1) | Y \ge u_j, X, A=1)
 = P( U^{(1)} > u_j | U^{(1)} \ge u_j, C^{(1)} \ge u_j, X, A=1)
 = P( U^{(1)} > u_j | U^{(1)} \ge u_j, X, A=1) = P( U^{(1)} > u_j | U^{(1)} \ge u_j, X)$,
where the third equality follows from NIC Assumption~\eqref{eq:non-informative} and the fourth equality follows from NUC Assumption~\eqref{eq:no-confounding}.}
In general, these covariate functions need to be estimated by postulating and fitting some (parametric) regression models,
denoted as $\pi(X; \gamma)$ for $\pi^*(X)$, $\pi_{1j} (X; \rho_{1j}) $ for $\pi^*_{1j} (X)$, and $m_{1j} (X; \eta_{1j})$ for $m^*_{1j} (X)$
for $j=1,\ldots,k$, where $\gamma$, $\rho_{1j}$, and $\eta_{1j}$ are parameter vectors.
In general, let $\hat\gamma$, $\{\hat\rho_{1j}: j=1,\ldots,k\}$, and $\{ \hat\eta_{1j}: j=1,\ldots,k\}$ be any estimators which are consistent for
their true values $\gamma^*$, $\{\rho^*_{1j}: j=1,\ldots,k\}$, and $\{\eta^*_{1j}: j=1,\ldots,k\}$ if the associated models are correctly specified
with $ \pi^*(X) = \pi(X;\gamma^*)$, $\pi^*_{1j}(X) = \pi_{1j} (X; \rho^*_{1j})$, and  $m^*_{1j}(X) = m_{1j} (X; \eta^*_{1j})$ for $j=1,\ldots,k$.
Then the fitted values $\pi(X;\hat\gamma)$, $\pi_{1j}(X;\hat\rho_{1j})$, and $m_{1j}(X;\hat\eta_{1j})$
are consistent for the unknown functions $\pi^*(X)$, $\pi^*_{1j} (X)$, and $m^*_{1j} (X)$ respectively.

For the treatment propensity score $\pi^*(X)$, logistic regression models are commonly used as described in Section~\ref{sec:existing}.
For the non-censoring probabilities $\{ \pi^*_{1j}(X): j=1,\ldots,k\}$
or the conditional survival probabilities $\{ m^*_{1j}(X): j=1,\ldots,k\}$,
regression models are usually specified with homogeneity (or dimension-reduction) assumptions on
some components of $\rho_{1j}$ or $\eta_{1j}$ over $j=1,\ldots,k$.
In particular, proportional hazards models for $\{ \pi^*_{1j}(X): j=1,\ldots,k\}$ or $\{ m^*_{1j}(X): j=1,\ldots,k\}$
assumes that the coefficients associated with non-intercept covariate terms are constant over $j$ in the hazard probabilities or odds.
See Tan (2023) for relevant regression models and inferences in the discrete-time setting.

Given any estimators $\hat\gamma$, $\hat\rho_{1,1:k}=\{\hat\rho_{1j}: j=1,\ldots,k\}$, and $\hat\eta_{1,1:k}=\{ \hat\eta_{1j}: j=1,\ldots,k\}$,
the augmented IPW estimator of $S_{1k}$ is
\begin{align}\label{eq:AIPW}
\hat S_{1k} (\hat\gamma, \hat\rho_{1,1:k}, \hat\eta_{1,1:k}) = \frac{1}{n} \sum_{i=1}^n \varphi_{1ki}(\hat\gamma, \hat\rho_{1,1:k}, \hat\eta_{1,1:k})
\end{align}
with
\begin{align} \label{eq:phi-ps}
\begin{split}
 & \quad \varphi_{1ki}(\hat\gamma, \hat\rho_{1,1:k}, \hat\eta_{1,1:k}) \\
 &= \frac{A_i }{\hat\pi_i } \frac{\overline R_{ki} }{\hat{\overline \pi}_{1ki} } 1\{ U^{(1)}_i > u_k\}
  - \sum_{\ell=0}^{k-1} \frac{A_i }{\hat\pi_i } \frac{\overline R_{\ell i} }{\hat{\overline \pi}_{1\ell i} }  \left( \frac{R_{\ell+1, i} }{\hat\pi_{1,\ell+1, i} }-1 \right) \hat{\underline m}_{1,\ell+1,i}  1\{ U^{(1)}_i > u_\ell \} \\
  &\qquad\qquad  - \left( \frac{A_i }{\hat\pi_i } -1 \right) \hat{\underline m}_{11i} ,
\end{split}
\end{align}
where the following notation is used:
$R_{\ell i} = 1\{ Y_i \ge u_\ell\}$, $\hat\pi_i = \pi(X_i; \hat\gamma)$, $\hat\pi_{1\ell i} = \pi_{1\ell} (X_i; \hat\rho_{1\ell}) $,
$ \hat m_{1\ell i} = m_{1\ell} (X_i; \hat\eta_{1\ell})$,
$\overline R_{\ell i} = \prod_{j=1}^\ell R_{ji} $, $ \hat{\overline \pi}_{1\ell i} = \prod_{j=1}^\ell \hat\pi_{1ji} $,
$\hat{\underline m}_{1\ell i} = \prod_{j=\ell}^k \hat m_{1ji} $
for $1 \le \ell \le k$, and
$\overline R_{0i} \equiv 1 $ and $\overline \pi_{10i} \equiv 1$.
(The dependency of $\hat{\underline m}_{1\ell i} $ on $k$ is suppressed in the notation.)
The expression \eqref{eq:phi-ps} uses the potential survival time $U_i^{(1)}$ only through the product $A_i R_{1\ell i} 1\{ U_i^{(1)} > u_\ell \}$,
which can be evaluated as $1 \{ Y_i \ge u_\ell, (Y_i,\Delta_i)\not=(u_\ell,1), A_i=1\}$ from the observed data.
Equivalently, the augmented IPW estimator can be expressed as \eqref{eq:AIPW} with
\begin{align} \label{eq:phi-or}
\begin{split}
&\quad \varphi_{1ki}(\hat\gamma, \hat\rho_{1,1:k}, \hat\eta_{1,1:k}) \\
&= \sum_{\ell=0}^{k-1} \frac{A_i }{\hat\pi_i } \frac{\overline R_{\ell+1, i} }{\hat{\overline \pi}_{1,\ell+1, i} } \hat{\underline m}_{1,\ell+2,i}
1\{ U^{(1)}_i > u_\ell \} ( 1\{U^{(1)}_i > u_{\ell+1}\}- \hat m_{1,\ell+1,i} )
  + \hat{\underline m}_{11i} ,
\end{split}
\end{align}
where $ \underline m_{1,k+1,i} \equiv 1$.
The augmented IPW estimator depends on fitted values from the two sets of regression models mentioned earlier.
For illustration, the explicit expressions for the special case of $k = 2$ are provided in Supplement Section~\ref{sec:aipw-k2}.
Retaining only the first term in the expression \eqref{eq:phi-ps} leads to the IPW estimator
\begin{align}
 \hat S_{1k}(\hat\gamma, \hat\rho_{1,1:k})
 = \frac{1}{n} \sum_{i=1}^n  \frac{A_i }{\hat\pi_i } \frac{\overline R_{ki} }{\hat{\overline \pi}_{1ki} } 1\{ U^{(1)}_i > u_k\} , \label{eq:IPW}
\end{align}
depending only on the fitted propensity scores $\hat\pi_i$ and non-censoring probabilities $\{\hat\pi_{1\ell i}: \ell=1,\ldots,k\}$.
Retaining only the last term in the expression \eqref{eq:phi-or} leads to the iterated-conditional-expectation (ICE) estimator
\begin{align}
 \hat S_{1k}( \hat\eta_{1,1:k})
 = \frac{1}{n} \sum_{i=1}^n   \hat{\underline m}_{11i}
  = \frac{1}{n} \sum_{i=1}^n  \hat m_{11i}\cdots\hat m_{1ki} ,  \label{eq:ICE}
\end{align}
depending only on the fitted conditional survival probabilities $\{\hat m_{1\ell i}: \ell=1,\ldots,k\}$.
The estimator $\hat S_{1k}(\hat\gamma, \hat\rho_{1,1:k})$ or $\hat S_{1k}( \hat\eta_{1,1:k})$ is consistent for $S_{1k}$
only if the associated models are correctly specified, $\pi(X;\gamma)$ and $\{\pi_{1j}(X; \rho_{1j}) :1\le j\le k\}$ in the first case or
$\{m_{1j}(X; \eta_{1j}) :1\le j\le k\}$ in the second case.
However, the augmented IPW estimator is doubly roust, i.e., being consistent for $S_{1k}$ if either
$\pi(X;\gamma)$ and $\{\pi_{1j}(X; \rho_{1j}) :1\le j\le k\}$ are correctly specified or
$\{m_{1j}(X; \eta_{1j}) :1\le j\le k\}$ are correctly specified.

Variance estimation for the augmented IPW estimator requires additional consideration
to properly account for sampling variations associated with the estimators $(\hat\gamma, \hat\rho_{1,1:k}, \hat\eta_{1,1:k})$.
By the current theory of augmented IPW estimation (e.g., Tsiatis 2006), there are three distinct scenarios. 
\begin{itemize}\addtolength{\itemsep}{-.1in}
\item
If the two sets of models, first $\pi(X;\gamma)$ and $\{\pi_{1j}(X; \rho_{1j}) :1\le j\le k\}$ and second
$\{m_{1j}(X; \eta_{1j}) :1\le j\le k\}$, are correctly specified, then
the asymptotic variance of $\hat S_{1k} (\hat\gamma, \hat\rho_{1,1:k}, \hat\eta_{1,1:k}) $
can be consistently estimated by
\begin{align}
 \hat V_{\mytext{r}}( \hat S_{1k} (\hat\gamma, \hat\rho_{1,1:k}, \hat\eta_{1,1:k}) )
 = \frac{1}{n^2} \sum_{i=1}^n \left\{ \varphi_{1ki}(\hat\gamma, \hat\rho_{1,1:k}, \hat\eta_{1,1:k}) -
 \hat S_{1k} (\hat\gamma, \hat\rho_{1,1:k}, \hat\eta_{1,1:k}) \right\}^2 ,  \label{eq:AIPW-var}
\end{align}
without being affected by the sampling variations from $(\hat\gamma, \hat\rho_{1,1:k}, \hat\eta_{1,1:k})$
(in other words, $(\hat\gamma, \hat\rho_{1,1:k}, \hat\eta_{1,1:k})$ can be treated as if they were data-independent).

\item If the models $\pi(X;\gamma)$ and $\{\pi_{1j}(X; \rho_{1j}) :1\le j\le k\}$ are correctly specified \textit{and} if
$(\hat\gamma, \hat\rho_{1,1:k})$ are ML estimators, then the variance estimator
(\ref{eq:AIPW-var}) is conservative, i.e., no smaller than the asymptotic variance of $\hat S_{1k} (\hat\gamma, \hat\rho_{1,1:k}, \hat\eta_{1,1:k}) $.

\item If the models $\{m_{1j}(X; \eta_{1j}) :1\le j\le k\}$, are correctly specified, then the variance estimator (\ref{eq:AIPW-var})
may be inconsistent in either direction.
\end{itemize} 
In practice, bootstrap or related methods can be used, assuming that the usual justification holds for such methods.

\subsection{Re-examining weighted KM estimation} \label{sec:weighted-KM-re}

With the postulated non-censoring and conditional survival probabilities
$\pi_{1j} (X; \rho_{1j})$ and $ m_{1j} (X;\eta_{1j})$ allowed to vary in $X$,
the augmented IPW estimator is apparently more complicated than the weighted KM estimator \eqref{eq:weighted-KM}.
There is, however, an interesting connection. We show that the augmented IPW estimator \textit{algebraically} reduces to the weighted KM estimator
if $\pi_{1j} (X; \rho_{1j})$ and $ m_{1j} (X;\eta_{1j})$, $j=1,\ldots,k$, are assumed to be constant in $X$
\textit{and} if the estimators $\hat\rho_{1j}$ and $\hat \eta_{1j}$, $j=1,\ldots,k$, are defined in a particular manner.

\begin{pro} \label{pro:weighted-KM}
Consider the following Constant non-Censoring Probability (CCP) model for $\pi^*_{1j}(X)$ and Constant Survival Probability (CSP) model
for $m^*_{1j}(X)$: 
\begin{align}
 & \pi_{1j}(X; \rho_{1j} ) = \rho_{1j}, \quad j=1,\ldots,k, \label{eq:const-pi}  \\
 & m_{1j} (X; \eta_{1j}) = \eta_{1j},  \quad j=1,\ldots,k, \label{eq:const-m}
\end{align}
where $\{\rho_{1j}: j=1,\ldots,k\}$ and $\{ \eta_{1j}: j=1,\ldots,k\}$ are unknown parameters (independent of $X$).
For any estimator $\hat\gamma$,
define $\hat\rho_{1,1:k,\mytext{CAL}}=\{\hat \rho_{1j,\mytext{CAL}}: j=1,\ldots,k\}$ and $\hat\eta_{1,1:k,\mytext{CAL}}=\{ \hat \eta_{1j,\mytext{CAL}}: j=1,\ldots,k\}$,
called calibrated estimators, as the solutions to the equations,
\begin{align}
& 0= \frac{1}{n} \sum_{i=1}^n  \frac{A_i }{\hat\pi_i } \overline R_{\ell i}
\left( \frac{R_{\ell+1, i} }{\rho_{1,\ell+1} }-1 \right) 1\{ U^{(1)}_i > u_\ell \} , \quad \ell=0,\ldots,k-1,  \label{eq:CAL-rho} \\
& 0 = \frac{1}{n} \sum_{i=1}^n  \frac{A_i }{\hat\pi_i } \overline R_{\ell+1, i}
1\{ U^{(1)}_i > u_\ell \} ( 1\{U^{(1)}_i > u_{\ell+1}\} -  \eta_{1,\ell+1} ), \quad \ell=0,\ldots,k-1,  \label{eq:CAL-eta}
\end{align}
where $\hat\pi_i = \pi (X_i;\hat\gamma)$ as before.
The dependency of $(\hat\rho_{1,1:k,\mytext{CAL}}, \hat\eta_{1,1:k,\mytext{CAL}})$ on $\hat\gamma$ is suppressed in the notation.
The augmented IPW estimator not only coincides with the ICE and weighted KM estimators,
\begin{align*}
& \quad \hat S_{1k} (\hat\gamma, \hat\rho_{1,1:k,\mytext{CAL}}, \hat\eta_{1,1:k,\mytext{CAL}})
= \hat S_{1k} (\hat\eta_{1,1:k,\mytext{CAL}})
= \hat S_{1k,\mytext{wKM}},
\end{align*}
but also, if the normalization $\sum_{i=1}^n A_i \hat\pi_i^{-1}=n$ is satisfied, coincides with the IPW estimator,
\begin{align*}
& \quad \hat S_{1k} (\hat\gamma, \hat\rho_{1,1:k,\mytext{CAL}}, \hat\eta_{1,1:k,\mytext{CAL}})
= \hat S_{1k} (\hat\gamma, \hat\rho_{1,1:k,\mytext{CAL}} ) ,
\end{align*}
where $\hat S_{1k} (\hat\gamma, \hat\rho_{1,1:k}, \hat\eta_{1,1:k})$, $\hat S_{1k} (\hat\gamma, \hat\rho_{1,1:k} )$ and $\hat S_{1k} (\hat\eta_{1,1:k})$ are defined in \eqref{eq:AIPW}, \eqref{eq:IPW} and \eqref{eq:ICE} respectively.
\end{pro}

As a consequence of Proposition~\ref{pro:weighted-KM}, the weighted KM estimator $\hat S_{1k,\mytext{wKM}}$ is
doubly robust, i.e., being consistent if in addition to Assumptions~\eqref{eq:no-confounding}--\eqref{eq:non-informative},
either PS model $\pi(X;\gamma)$ and CCP model \eqref{eq:const-pi} are correctly specified
or CSP model \eqref{eq:const-m} is correctly specified.
Therefore, the consistency of $\hat S_{1k,\mytext{wKM}}$ in general depends on
correct specification of either CCP model \eqref{eq:const-pi} or CSP model \eqref{eq:const-m},
in addition to that of PS model in the former case.\footnote{\label{ft:consistency-wKM}
Intuitively, this result can be explained by the fact that the hazard probabilities are estimated in \eqref{eq:q-hat} as constants $\hat q_{1j}$,
\textit{independently} of covariates, for weighted KM estimation.}
This result seems surprising and previously under-appreciated in two ways:
first $\hat S_{1k,\mytext{wKM}}$ is in general inconsistent if PS model is correct but CCP model is misspecified,
and second $\hat S_{1k,\mytext{wKM}}$ is consistent if CSP model is correctly even when PS model is misspecified.
The consistency result of Xie \& Liu (2005) was formally obtained in a setting that can be seen to implicitly
assume that PS model $\pi(X;\gamma)$ and CSP model \eqref{eq:const-m} are correctly specified.
For further understanding, see Supplement Section~\ref{sec:comparison-KM} for a comparison of weighted and unweighted KM estimation
in their consistency properties.

The augmented IPW representation for the weighted KM estimator is also helpful for studying variance estimation.
As mentioned in Section~\ref{sec:AIPW}, consistent variance estimation may be analytically difficult unless
\textit{all} the associated regression models are correctly specified. Nevertheless, with the particular models \eqref{eq:const-pi} and \eqref{eq:const-m}
and estimators $\{\hat \rho_{1j,\mytext{CAL}}: j=1,\ldots,k\}$ and $\{ \hat \eta_{1j,\mytext{CAL}}: j=1,\ldots,k\}$,
the variance estimator \eqref{eq:AIPW-var} coincides with the model-robust variance estimator $\hat V_{\mytext{r}} ( \hat S_{1k,\mytext{wKM}} )$ in \eqref{eq:weighted-KM-Vr}
under the condition $\sum_{i=1}^n A_i\hat{\pi}_i = n$. Moreover, $\hat V_{\mytext{r}} ( \hat S_{1k,\mytext{wKM}} )$
and, if $E( \sum_{i=1}^n A_i \hat\pi_i^{-1}) =n$,
the variance estimator \eqref{eq:AIPW-var} is nonparametrically consistent even when both
CCP and CSP models \eqref{eq:const-pi} and \eqref{eq:const-m} are misspecified,
\textit{provided} that $\hat\gamma$ is treated as data-independent.
With model misspecification, the point estimator $\hat S_{1k} (\hat\gamma, \hat\rho_{1,1:k,\mytext{CAL}}, \hat\eta_{1,1:k,\mytext{CAL}})$
converges in probability to a population quantity which may differ from $S_{1k}$.

\begin{pro} \label{pro:weighted-KM-Vr}
Suppose that $\hat\gamma$ were data-independent. Then the asymptotic variance of $\hat S_{1k,\mytext{wKM}}$
or equivalently $\hat S_{1k} (\hat\gamma, \hat\rho_{1,1:k,\mytext{CAL}}, \hat\eta_{1,1:k,\mytext{CAL}})$
can be consistently estimated by $\hat V_{\mytext{r}} ( \hat S_{1k,\mytext{wKM}} )$
or, if $E( \sum\limits_{i=1}^n A_i \hat\pi_i^{-1}) \allowbreak =n$, by $\hat V_{\mytext{r}}( \hat S_{1k} (\hat\gamma, \hat\rho_{1,1:k,\mytext{CAL}}, \hat\eta_{1,1:k,\mytext{CAL}}) )$
as defined in \eqref{eq:AIPW-var}.
Moreover, if $\sum_{i=1}^n A_i \hat\pi_i^{-1}=n$, then
 $\hat V_{\mytext{r}}( \hat S_{1k} (\hat\gamma, \hat\rho_{1,1:k,\mytext{CAL}}, \hat\eta_{1,1:k,\mytext{CAL}}) )$ algebraically coincides with
$\hat V_{\mytext{r}} ( \hat S_{1k,\mytext{wKM}} )$. 
\end{pro}

By Proposition~\ref{pro:weighted-KM-Vr}, in spite of model misspecification,
the sampling variations from $\{\hat \rho_{1j,\mytext{CAL}}: j=1,\ldots,k\}$ and $\{ \hat \eta_{1j,\mytext{CAL}}: j=1,\ldots,k\}$
do not affect variance estimation for $\hat S_{1k} (\hat\gamma, \hat\rho_{1,1:k,\mytext{CAL}}, \hat\eta_{1,1:k,\mytext{CAL}})$, while the sampling variation from $\hat\gamma$ were ignored.
As seen from the proof, the estimating equations \eqref{eq:CAL-rho} and \eqref{eq:CAL-eta} for $(\hat\rho_{1,1:k,\mytext{CAL}}, \hat\eta_{1,1:k,\mytext{CAL}})$
are precisely constructed to achieve this desirable property via calibrated estimation (Tan 2020; Ghosh \& Tan 2022).

Proposition~\ref{pro:weighted-KM-Vr}, however, leaves open the question of variance estimation while accounting for data-dependency of $\hat\gamma$.
In fact,
the usual theory guaranteeing the conservativeness of the variance estimator \eqref{eq:AIPW-var} under correctly specified models $\pi(X;\gamma)$ and $\{\pi_{1j}(X;\rho_{1j}): 1\le j\le k\}$ fitted by maximum likelihood (e.g., Tsiatis 2006) does not in general apply to the weighted KM estimator. Specifically, the ML estimators of  $\{\rho_{1j}: j=1,\ldots,k\}$ and $\{ \eta_{1j}: j=1,\ldots,k\}$
are defined as the solutions to
\begin{align}
& 0= \frac{1}{n} \sum_{i=1}^n A_i \overline R_{\ell i}
\left(  R_{\ell+1, i} - \rho_{1,\ell+1} \right) 1\{ U^{(1)}_i > u_\ell \} , \quad \ell=0,\ldots,k-1,  \label{eq:ML-rho} \\
& 0 = \frac{1}{n} \sum_{i=1}^n  A_i \overline R_{\ell+1, i}
1\{ U^{(1)}_i > u_\ell \} ( 1\{U^{(1)}_i > u_{\ell+1}-  \eta_{1,\ell+1} ), \quad \ell=0,\ldots,k-1,  \label{eq:ML-eta}
\end{align}
which are independent of $\hat\gamma$ because models \eqref{eq:const-pi} and \eqref{eq:const-m} are concerned about \textit{nested} conditional probabilities.
Hence the calibrated estimators
$\{\hat \rho_{1j,\mytext{CAL}}: j=1,\ldots,k\}$ and $\{ \hat \eta_{1j,\mytext{CAL}}: j=1,\ldots,k\}$
differ from the ML estimators except in the special case where the postulated propensity score $\pi(X;\gamma)$ is constant in $X$.

For completeness, we also point out that if CSP model \eqref{eq:const-m} is correctly specified but
PS model $\pi(X;\gamma)$ and CCP model \eqref{eq:const-pi} may be misspecified, then the model-based variance estimator
 $\hat V_{\mytext{b}} ( \hat S_{1k,\mytext{wKM}} )$ in \eqref{eq:weighted-KM-Vb} is consistent,
subject to a technical condition on $\hat\gamma$.
As discussed in Supplement Section~\ref{sec:comparison-KM}, if CSP model \eqref{eq:const-m} is correctly specified,
then both the unweighted and weighted KM estimators are consistent for $S_{1k}$
and hence, in this case, weighted KM estimation is not advantageous.

\begin{pro} \label{pro:weighted-KM-Vb}
If CSP model \eqref{eq:const-m} is correctly specified but PS model $\pi(X;\gamma)$ and CCP model \eqref{eq:const-pi} may be misspecified,
then the asymptotic variance of $\hat S_{1k,\mytext{wKM}}$
or equivalently $\hat S_{1k} (\hat\gamma, \hat\rho_{1,1:k,\mytext{CAL}}, \hat\eta_{1,1:k,\mytext{CAL}})$
can be consistently estimated by $\hat V_{\mytext{b}} ( \hat S_{1k,\mytext{wKM}} )$,
provided that $\hat\gamma$ admits a limit value $\breve\gamma$ satisfying $E( A \pi(X;\breve\gamma)^{-1}) =1$.
\end{pro}

Proposition~\ref{pro:weighted-KM-Vb} is more general than Xie \& Liu (2005), Proposition 2, where
both PS model $\pi(X;\gamma)$ and CSP model \eqref{eq:const-m} are assumed to be correctly specified.
If model $\pi(X;\gamma)$ is correctly specified, the technical condition  $E( A \pi(X;\breve\gamma)^{-1}) =1$ is satisfied by any consistent estimator $\hat\gamma$,
for example, the ML estimator. In this case, Proposition~\ref{pro:weighted-KM-Vb} agrees with Xie \& Liu (2005).
However, if PS model $\pi(X;\gamma)$ is misspecified as allowed in Proposition~\ref{pro:weighted-KM-Vb}, then whether $E( A \pi(X;\breve\gamma)^{-1}) =1$ is satisfied
depends on the actual definition of $\hat\gamma$. In particular, this condition is satisfied by the calibrated estimator of $\hat\gamma$
defined later in \eqref{eq:CAL-gamma}, but it may not be satisfied by
the ML estimator of $\gamma$ except in the special case where the postulated propensity score $\pi(X;\gamma)$ is constant in $X$.

\section{Calibrating weighted survival analysis} \label{sec:CAL-AIPW}

\subsection{Calibrated estimation for survival probabilities} \label{sec:CAL-surv-prob}

From Section~\ref{sec:weighted-KM-re}, consistent variance estimation remains to be developed for the weighted KM estimator $\hat S_{1k,\mytext{wKM}}$
while accounting for the data-dependency of $\hat\gamma$
if PS model $\pi(X;\gamma)$ and CCP model $\{\pi_{1j}(X; \rho_{1j}) :1\le j\le k\}$
in \eqref{eq:const-pi} are correctly specified but CSP model $\{m_{1j}(X; \eta_{1j}) :1\le j\le k\}$ in \eqref{eq:const-m} may be misspecified,
which is the scenario where weighted KM estimation is advantageous over the unweighted version.
The usual theory of augmented IPW estimation indicating conservative variance estimation does \textit{not} apply to $\hat S_{1k,\mytext{wKM}}$
because the associated estimators for $\{\rho_{1j}: j=1,\ldots,k\}$ are not ML estimators, even when $\hat\gamma$ is an ML estimator.
To tackle this issue, we propose a modification of $\hat S_{1k,\mytext{wKM}}$
to achieve consistent variance estimation analytically
accounting for estimation of $\gamma$ and $\{(\rho_{1j}, \eta_{1j}): j=1,\ldots,k\}$
in the aforementioned scenario through calibrated estimation (Tan 2020; Ghosh \& Tan 2022).
In fact, a direct application of calibrated estimation to the augmented IPW estimator
$\hat S_{1k} (\hat\gamma, \hat\rho_{1,1:k}, \hat\eta_{1,1:k})$ involves a technical difficulty in
our setting with non-constant $\pi(X;\gamma)$ but
constant $\pi_{1j}(X;\rho_{1j})$ and $m_{1j}(X;\eta_{1j})$ as in CCP and CSP model
\eqref{eq:const-pi} and \eqref{eq:const-m}. See Supplement Section~\ref{sec:original-CAL} for details.

First, we generalize the augmented IPW estimator $\hat S_{1k} (\hat\gamma, \hat\rho_{1,1:k}, \hat\eta_{1,1:k})$ in \eqref{eq:AIPW} as follows,
where the non-censoring probability model, $\pi_{1j}(X; \rho_{1j})$,
and survival probability model, $m_{1j}(X; \eta_{1j})$, may not necessarily be the ``constant'' models \eqref{eq:const-pi} and \eqref{eq:const-m}.
Consider an outcome regression (OR) model, denoted as $\mu_{1k}(X;\alpha_{1k})$, for $\mu^*_{1k} (X) = P ( U^{(1)} > u_k | X)$, for example:
\begin{align}
 \mu_{1k} (X; \alpha_{1k}) =  \frac{ \exp (\alpha_{1k}^\T f(X)) }{ 1+ \exp(\alpha_{1k}^\T f(X)) }, \label{eq:OR-model}
\end{align}
where $f(X)$ is the same vector of covariate functions as in propensity score model \eqref{eq:PS-model} and $\alpha_{1k}$ is a parameter vector.
By definition, $\mu^*_{1k}(X)$ is related to the nested conditional probabilities as
$ \mu^*_{1k} (X) = \prod_{j=1}^k m^*_{1j}(X)$.
However, $\mu_{1k} (X; \alpha_{1k})$, as a working regression model which may be misspecified,
is not required to satisfy the same relationship with the regression models $\{ m_{1j}(X;\eta_{1j}): 1\le j \le k\}$,
$ \mu_{1k} (X; \alpha_{1k}) = \prod_{j=1}^k m_{1j} (X; \eta_{1j})$.
Given any estimators $\hat\gamma$, $\hat\rho_{1,1:k}=\{\hat\rho_{1j}: j=1,\ldots,k\}$, $\hat\eta_{1,1:k}=\{ \hat\eta_{1j}: j=1,\ldots,k\}$,
and $\hat\alpha_{1k}$, a new augmented IPW estimator of $S_{1k}$ is
\begin{align}\label{eq:AIPW-new}
\hat S_{1k} (\hat\gamma, \hat\rho_{1,1:k}, \hat\eta_{1,1:k}, \hat\alpha_{1k}) =
\frac{1}{n} \sum_{i=1}^n \varphi_{1ki}(\hat\gamma, \hat\rho_{1,1:k}, \hat\eta_{1,1:k}, \hat\alpha_{1k})
\end{align}
with
\begin{align} \label{eq:phi-ps-new}
\begin{split}
 & \quad \varphi_{1ki}(\hat\gamma, \hat\rho_{1,1:k}, \hat\eta_{1,1:k}, \hat\alpha_{1k}) \\
 &= \frac{A_i }{\hat\pi_i } \frac{\overline R_{ki} }{\hat{\overline \pi}_{1ki} } 1\{ U^{(1)}_i > u_k\}
  - \sum_{\ell=0}^{k-1} \frac{A_i }{\hat\pi_i } \frac{\overline R_{\ell i} }{\hat{\overline \pi}_{1\ell i} }  \left( \frac{R_{\ell+1, i} }{\hat\pi_{1,\ell+1, i} }-1 \right) \hat{\underline m}_{1,\ell+1,i}  1\{ U^{(1)}_i > u_\ell \} \\
  &\qquad\qquad  - \left( \frac{A_i }{\hat\pi_i } -1 \right) \hat\mu_{1ki} ,
\end{split}
\end{align}
where $\hat \mu_{1ki} = \mu_{1k}(X_i; \hat\alpha_{1k})$, in addition to the notation used in \eqref{eq:phi-ps}.
%
%
It is helpful to rewrite
$ \varphi_{1ki} (\hat\gamma, \hat\rho_{1,1:k}, \hat\eta_{1,1:k}, \hat\alpha_{1k})
 = \frac{A_i}{\hat\pi_i} \hat U^{(1)}_{ki} (\hat\rho_{1,1:k}, \hat\eta_{1,1:k}) - \left( \frac{A_i }{\hat\pi_i } -1 \right) \hat\mu_{1ki}
$,
where
\begin{align*}
 \hat U^{(1)}_{ki} (\hat\rho_{1,1:k}, \hat\eta_{1,1:k}) =
 \frac{\overline R_{ki} }{\hat{\overline \pi}_{1ki} } 1\{ U^{(1)}_i > u_k\}
  - \sum_{\ell=0}^{k-1} \frac{\overline R_{\ell i} }{\hat{\overline \pi}_{1\ell i} }  \left( \frac{R_{\ell+1, i} }{\hat\pi_{1,\ell+1,i} }-1 \right) \hat{\underline m}_{1,\ell+1}  1\{ U^{(1)}_i > u_\ell \},
\end{align*}
which can be evaluated from the observed data and treated as an \textit{imputed} value of $1\{U^{(1)}_i > u_k\}$ for each individual $i$ with $A_i=1$.
Then the new augmented IPW estimator of $S_{1k}$ takes a cross-sectional form of augmented IPW estimation with imputed survival outcomes:
\begin{align}
 & \hat S_{1k} (\hat\gamma, \hat\rho_{1,1:k}, \hat\eta_{1,1:k}, \hat\alpha_{1k}) =
\frac{1}{n} \sum_{i=1}^n \left\{
\frac{A_i}{\hat\pi_i} \hat U^{(1)}_{ki} (\hat\rho_{1,1:k}, \hat\eta_{1,1:k}) - \left( \frac{A_i }{\hat\pi_i } -1 \right) \hat\mu_{1ki} \right\}.  \label{eq:AIPW-cross-section}
\end{align}
If $\hat\mu_{1ki} = \hat{\underline m}_{11i}$, then \eqref{eq:AIPW-new} and \eqref{eq:phi-ps-new}
reduce to \eqref{eq:AIPW} and \eqref{eq:phi-ps}.
However, $\hat\mu_{1ki}$ in \eqref{eq:AIPW-new} is allowed to differ from $\hat{\underline m}_{11i}$.
Such a disentanglement of $\hat\mu_{1ki}$ from the \textit{implied} $\hat{\underline m}_{11i}$
is useful when the regression models $\{ m_{1j}(X;\eta_{1j}): 1\le j \le k\}$ may be misspecified.

In place of $\hat S_{1k,\mytext{wKM}}$,
our calibrated augmented IPW estimator of $S_{1k}$ is
\begin{align}
  \hat S_{1k,\mytext{CAL}} =\hat S_{1k} (\hat\gamma_{1,\mytext{CAL}}, \hat\rho_{1,1:k,\mytext{CAL}}, \hat\eta_{1,1:k,\mytext{CAL}}, \hat\alpha_{1k,\mytext{CAL}}) ,
  \label{eq:AIPW-CAL}
\end{align}
where $\hat\rho_{1,1:k,\mytext{CAL}}=\{\hat \rho_{1j,\mytext{CAL}}: j=1,\ldots,k\}$ and $\hat\eta_{1,1:k,\mytext{CAL}}=\{ \hat \eta_{1j,\mytext{CAL}}: j=1,\ldots,k\}$,
associated with CCP and CSP models \eqref{eq:const-pi} and \eqref{eq:const-m}, are defined as in Proposition~\ref{pro:weighted-KM}
with $\hat\gamma= \hat\gamma_{1,\mytext{CAL}}$.
The dependency of $(\hat\rho_{1,1:k,\mytext{CAL}}, \hat\eta_{1,1:k,\mytext{CAL}}, \hat\alpha_{1k,\mytext{CAL}})$ on $\hat\gamma_{1,\mytext{CAL}}$ is suppressed in the notation.
The estimator $\hat\gamma_{1,\mytext{CAL}}$ is a calibrated estimator of $\gamma$ in model \eqref{eq:PS-model}, defined as a solution to
the estimating equation
\begin{align}
 0 = \frac{1}{n} \sum_{i=1}^n \left\{ \frac{A_i}{\pi(X_i;\gamma)} -1 \right\} f(X_i) .  \label{eq:CAL-gamma}
\end{align}
Equivalently, $\hat\gamma_{1,\mytext{CAL}}$ can also be determined as a minimizer of the convex loss function (Tan 2020)
\begin{align}
\ell_{\mytext{CAL}}(\gamma) = \frac{1}{n}\sum_{i=1}^n \{A_i\me^{-f^\T(X_i)\gamma} + (1 - A_i)f^\T(X_i)\gamma\}. \label{eq:loss-CAL-gamma}
\end{align}
The estimator $\hat\alpha_{1k,\mytext{CAL}}$ is a calibrated estimator of $\alpha_{1k}$ in model \eqref{eq:OR-model}, defined as a solution to
\begin{align}
 0 = \frac{1}{n} \sum_{i=1}^n A_i \frac{1-\hat\pi_{1i,\mytext{CAL}}} {\hat\pi_{1i,\mytext{CAL}} }
  \left\{ \hat U^{(1)}_{ki,\mytext{CAL}} - \mu_{1k} (X_i;\alpha_{1k}) \right\} f(X_i),  \label{eq:CAL-alpha}
\end{align}
where $\hat\pi_{1i,\mytext{CAL}} = \pi(X_i; \hat\gamma_{1,\mytext{CAL}})$ and
$\hat U^{(1)}_{ki,\mytext{CAL}}= \hat U^{(1)}_{ki} (\hat\rho_{1,1:k,\mytext{CAL}}, \hat\eta_{1,1:k,\mytext{CAL}})$.
Similarly to $\hat{\gamma}_{1, \mytext{CAL}}$, $\hat{\alpha}_{1k,\mytext{CAL}}$ can also be equivalently determined as a minimizer of the convex loss function
\begin{align}
\ell_{\mytext{CAL}}(\alpha_{1k}) = \frac{1}{n}\sum_{i=1}^n A_i\frac{1 - \hat{\pi}_{1i, \mytext{CAL}}}{\hat{\pi}_{1i, \mytext{CAL}}}[ -\hat{U}^{(1)}_{ki, \mytext{CAL}}f^\T(X_i)\alpha_{1k} + \log\{1 + \me^{f^\T(X_i)\alpha_{1k}}\} ]. \label{eq:loss-CAL-alpha}
\end{align}
The estimator $\hat\alpha_{1k,\mytext{CAL}}$ is a weighted ML estimator for logistic regression with response
$ \hat U^{(1)}_{ki,\mytext{CAL}}$, covariate terms $f(X_i)$, and weight $(1-\hat\pi_{1i,\mytext{CAL}})/\hat\pi_{1i,\mytext{CAL}} $
in the treated group $\{A_i=1\}$.
Hence the estimator $\hat S_{1k,\mytext{CAL}}$ can be seen as an application of the method in Tan (2020)
with PS model \eqref{eq:PS-model}
and OR model \eqref{eq:OR-model} \textit{after} the imputed response $ \hat U^{(1)}_{ki,\mytext{CAL}}$ is obtained.

Similarly to $S_{1k,\mytext{wKM}}$,
the point estimator $\hat S_{1k,\mytext{CAL}}$ can be shown to be doubly robust for $S_{1k}$, i.e., being consistent in two distinct scenarios:
either PS and CCP models \eqref{eq:PS-model} and \eqref{eq:const-pi} are correctly specified,
or CSP and OR models \eqref{eq:const-m} and \eqref{eq:OR-model} are correctly specified.
In our specific setting, CSP model \eqref{eq:const-m} implies OR model \eqref{eq:OR-model},
which allows $\mu^*_{1k}(X)$ to be constant in $X$ because $f(\cdot)$  is assumed to include an intercept.
Hence the second scenario is just that CSP model \eqref{eq:const-m} is correctly specified.
As discussed in Supplement Section~\ref{sec:comparison-KM},
the consistency in the first scenario (although not the second) represents a real advantage
over unweighted KM estimation.

Compared with weighted KM estimation, a benefit of $\hat S_{1k,\mytext{CAL}}$ is that in the first scenario above, consistent variance estimation
 can be obtained for $\hat S_{1k,\mytext{CAL}}$ such that the sampling variations from estimation of
the nuisance parameters $(\gamma, \rho_{1,1:k}, \eta_{1,1:k}, \alpha_{1k})$ can be safely ignored,
as indicated by the asymptotic expansion \eqref{eq:CAL-expan}.\footnote{\label{ft:advtange-cal-over-wkm}
An additional advantage of $\hat S_{1k,\mytext{CAL}}$ over $\hat S_{1k,\mytext{wKM}}$ is
that $\hat S_{1k,\mytext{CAL}}$ tends to achieve smaller variance than $\hat S_{1k,\mytext{wKM}}$,
because $\hat S_{1k,\mytext{CAL}}$ involves a flexible augmentation term in \eqref{eq:AIPW-cross-section}
with $\mu_{1k}(X;\alpha_{1k})$ non-constant in $X$,
whereas the augmentation term for $\hat S_{1k,\mytext{wKM}}$ is restricted with constant $\mu_{1k}(X;\alpha_{1k}) = \prod_{j=1}^k \eta_{1j} $
as implied from CCP model \eqref{eq:const-m}.
}
The consistency of the variance estimator \eqref{eq:AIPW-CAL-Vr} does not require that \textit{all} relevant regression models be correctly specified,
as in the usual theory for the consistency of the variance estimator \eqref{eq:AIPW-var} (e.g., Tsiatis 2006).

\begin{pro} \label{pro:AIPW-CAL}
Suppose that PS model \eqref{eq:PS-model} is correctly specified, but CCP, CSP or OR model \eqref{eq:const-pi}, \eqref{eq:const-m}
or \eqref{eq:OR-model} may be misspecified. Then under standard regularity conditions as $n\to\infty$ and the dimension of $f$, denoted as $\dim(f)$, and $K$ are fixed,
 $\hat S_{1k,\mytext{CAL}}$ satisfies the asymptotic expansion
\begin{align}
 \hat S_{1k,\mytext{CAL}} &= \hat S_{1k} (\breve\tau_{1,\mytext{CAL}} )
 + o_p (n^{-1/2}) = \frac{1}{n} \sum_{i=1}^n \varphi_{1ki} ( \breve\tau_{1,\mytext{CAL}} ) + o_p (n^{-1/2}), \label{eq:CAL-expan}
\end{align}
where $\hat S_{1k} (\tau)$ and $\varphi_{1ki}(\tau)$ are defined as
$\hat S_{1k} (\hat\tau) $ and $\varphi_{1ki}(\hat\tau)$ in \eqref{eq:AIPW-new} and \eqref{eq:phi-ps-new}
with $\hat\tau = (\hat\gamma , \hat\rho_{1,1:k}, \hat\eta_{1,1:k}, \allowbreak \hat\alpha_{1k})$ replaced by $\tau=(\gamma, \rho_{1,1:k}, \eta_{1,1:k}, \alpha_{1k})$,
and $\breve\tau_{1,\mytext{CAL}} = (\breve\gamma_{1,\mytext{CAL}}, \breve\rho_{1,1:k,\mytext{CAL}},$ $ \breve\eta_{1,1:k,\mytext{CAL}}, \breve\alpha_{1k,\mytext{CAL}})$
are the limit values of $\hat\tau_{1,\mytext{CAL}}=(\hat\gamma_{1,\mytext{CAL}}, \hat\rho_{1,1:k,\mytext{CAL}}, \hat\eta_{1,1:k,\mytext{CAL}}, \hat\alpha_{1k,\mytext{CAL}})$ (see the proof in Supplement Section~\ref{sec:technical-detail}  for details).
Moreover, the asymptotic variance of $\hat S_{1k,\mytext{CAL}}$ can be consistently estimated by
\begin{align}
 \hat V_{\mytext{r}}( \hat S_{1k,\mytext{CAL}} )
 = \frac{1}{n^2} \sum_{i=1}^n \left\{ \varphi_{1ki}(\hat\tau_{1,\mytext{CAL}}) -\hat S_{1k,\mytext{CAL}} \right\}^2 .  \label{eq:AIPW-CAL-Vr}
\end{align}
\end{pro}

By Proposition~\ref{pro:AIPW-CAL},
if PS and CCP models \eqref{eq:PS-model} and \eqref{eq:const-pi} are correctly specified, then $\hat S_{1k,\mytext{CAL}}$ is consistent for $S_{1k}$, and
an asymptotic $(1-c)$-confidence interval for $S_{1k}$ is
$ \hat S_{1k,\mytext{CAL}} \pm z_{c/2} \sqrt{\hat V_{\mytext{r}}( \hat S_{1k,\mytext{CAL}} ) }$, where
$z_{c/2}$ is the $(1-c/2)$-quantile of $\N(0,1)$.
Correct specification of both PS and CCP models \eqref{eq:PS-model} and \eqref{eq:const-pi} ensures the consistency of $\hat S_{1k,\mytext{CAL}}$ for $S_{1k}$,
whereas correct specification of PS model \eqref{eq:PS-model} alone ensures the consistency of $\hat V_{\mytext{r}}( \hat S_{1k,\mytext{CAL}} )$
for the asymptotic variance of $\hat S_{1k,\mytext{CAL}}$.

The asymptotic expansion \eqref{eq:CAL-expan} is achieved mainly because
the estimating equations \eqref{eq:CAL-rho}, \eqref{eq:CAL-eta}, \eqref{eq:CAL-gamma}, and \eqref{eq:CAL-alpha} for
$(\hat\rho_{1,1:k,\mytext{CAL}}, \hat\eta_{1,1:k,\mytext{CAL}}, \hat\gamma_{1,\mytext{CAL}}, \hat\alpha_{1k,\mytext{CAL}})$
are carefully chosen such that if PS model \eqref{eq:PS-model} is correctly specified, then the limit values $\breve\gamma_{1,\mytext{CAL}}=\gamma^*$,
$\breve\rho_{1,1:k,\mytext{CAL}}$, $\breve\eta_{1,1:k,\mytext{CAL}}$, and $\breve\alpha_{1k,\mytext{CAL}}$
satisfy the following calibration equations:
\begin{align*}
& 0 = E \left\{\frac{\partial }{\partial \tau} \hat S_{1k} (\tau) \right\}
 \Big|_{\breve\tau_{1,\mytext{CAL}} } ,
\end{align*}
where $\tau= (\gamma, \rho_{1,1:k}, \eta_{1,1:k },\alpha_{1k})$.
A similar asymptotic expansion as \eqref{eq:CAL-expan} also holds and underlies Proposition~\ref{pro:weighted-KM-Vr}
for the weighted KM estimator $\hat S_{1k,\mytext{wKM}}$
or equivalently the augmented IPW estimator $\hat S_{1k} (\hat\gamma, \hat\rho_{1,1:k,\mytext{CAL}}, \hat\eta_{1,1:k,\mytext{CAL}})$
except that only calibration equations for $(\rho_{1,1:k}, \eta_{1,1:k })$ are involved, with non-random $\hat\gamma$.

There are interesting implications when the logistic OR model \eqref{eq:OR-model} for $\mu^*_{1k}(X)$ is replaced by
a linear OR model:
\begin{align}
 \mu_{1k} (X; \alpha_{1k}) = \alpha_{1k}^\T f(X) . \label{eq:OR-model-lin}
\end{align}
Our estimator, $\hat S_{1k,\mytext{CAL,lin}}$, is defined similarly as $\hat S_{1k,\mytext{CAL}}$ in \eqref{eq:AIPW-CAL}.
The calibrated estimator $\hat\alpha_{1k,\mytext{CAL}}$, defined as a solution to \eqref{eq:CAL-alpha}
or a minimizer of
\begin{align}
\ell_{\mytext{CAL}}(\alpha_{1k}) = \frac{1}{n}\sum_{i=1}^n A_i\frac{1 - \hat{\pi}_{1i, \mytext{CAL}}}{\hat{\pi}_{1i, \mytext{CAL}}}[ -\hat{U}^{(1)}_{ki, \mytext{CAL}}f^\T(X_i)\alpha_{1k} + \frac{1}{2}\{f^\T(X_i)\alpha_{1k}\}^2 ], \label{eq:loss-CAL-alpha-lin}
\end{align}
becomes a weighted least-squares estimator for linear regression with response
$ \hat U^{(1)}_{ki,\mytext{CAL}}$, covariate terms $f(X_i)$, and weight $(1-\hat\pi_{1i,\mytext{CAL}})/\hat\pi_{1i,\mytext{CAL}} $
in the treated group $\{A_i=1\}$.
The point estimator $\hat S_{1k,\mytext{CAL,lin}}$ with linear OR model \eqref{eq:OR-model-lin} can be shown to be doubly robust
for $S_{1k}$, i.e., being consistent if either PS and CCP models \eqref{eq:PS-model} and \eqref{eq:const-pi} are correctly specified,
or CSP and OR models \eqref{eq:const-m} and \eqref{eq:OR-model-lin} (or just CSP model \eqref{eq:const-m}) are correctly specified.
As discussed earlier, the consistency in the first scenario is of more interest than in the second scenario.
Moreover, Proposition~\ref{pro:AIPW-CAL} can be extended as follows for variance estimation.
Hence in spite of the lack of range restriction to $[0,1]$, model \eqref{eq:OR-model-lin}
can be used, similarly to model \eqref{eq:OR-model}, as a working regression model to facilitate consistent variance estimation.

\begin{pro} \label{pro:AIPW-CAL-lin}
Suppose that PS, CCP, CSP, or linear OR model  \eqref{eq:PS-model}, \eqref{eq:const-pi}, \eqref{eq:const-m},
or \eqref{eq:OR-model-lin} may be misspecified. Then the conclusions in Proposition~\ref{pro:AIPW-CAL}
remain valid for the estimator
\begin{align*}
\hat S_{1k,\mytext{CAL,lin}} =\hat S_{1k} (\hat\gamma_{1,\mytext{CAL}}, \hat\rho_{1,1:k,\mytext{CAL}}, \hat\eta_{1,1:k,\mytext{CAL}}, \hat\alpha_{1k,\mytext{CAL,lin}}),
\end{align*}
i.e., defined as $\hat S_{1k,\mytext{CAL}}$ in \eqref{eq:AIPW-CAL}
with $\hat\alpha_{1k,\mytext{CAL}}$ replaced by $\hat\alpha_{1k,\mytext{CAL,lin}}$. 
\end{pro}

We provide two additional remarks about Proposition~\ref{pro:AIPW-CAL-lin}. First, by equation \eqref{eq:CAL-gamma} for $\hat\gamma_{1,\mytext{CAL}}$
and the linearity of $\mu_{1k}(\cdot;\alpha_{1k})$ in $\alpha_{1k}$, it is shown in Supplement Section~\ref{sec:technical-detail} that
\begin{align}
 \hat S_{1k,\mytext{CAL,lin}} =\hat S_{1k} (\hat\gamma_{1,\mytext{CAL}}, \hat\rho_{1,1:k,\mytext{CAL}}, \hat\eta_{1,1:k,\mytext{CAL}})
 = \hat S_{1k,\mytext{wKM}} (\hat\gamma_{1,\mytext{CAL}}),  \label{eq:AIPW-reduction}
\end{align}
where $\hat S_{1k} (\hat\gamma_{1,\mytext{CAL}}, \hat\rho_{1,1:k,\mytext{CAL}}, \hat\eta_{1,1:k,\mytext{CAL}})$
and $\hat S_{1k,\mytext{wKM}} (\hat\gamma_{1,\mytext{CAL}})$
are $\hat S_{1k} (\hat\gamma, \hat\rho_{1,1:k,\mytext{CAL}}, \hat\eta_{1,1:k,\mytext{CAL}})$
and $\hat S_{1k,\mytext{wKM}}$ in Section~\ref{sec:weighted-KM-re} with $\hat\gamma= \hat\gamma_{1,\mytext{CAL}}$.
Incidentally, a benefit from this reduction is by the expression \eqref{eq:weighted-KM} for $\hat S_{1k,\mytext{wKM}}$,
the estimator $\hat S_{1k,\mytext{CAL,lin}}$ is necessarily non-increasing in $k$.
Such a monotonicity may not hold for $\hat S_{1k,\mytext{CAL}}$ in finite samples when logistic OR model \eqref{eq:OR-model} is used.
By Proposition~\ref{pro:AIPW-CAL-lin},  the variance estimator \eqref{eq:AIPW-CAL-Vr} is
consistent for the asymptotic variance of $\hat S_{1k,\mytext{wKM}} (\hat\gamma_{1,\mytext{CAL}})$
and therefore serves as a correction to $\hat V_{\mytext{r}}( \hat S_{1k,\mytext{wKM}} )$ with $\hat\gamma= \hat\gamma_{1,\mytext{CAL}}$ in \eqref{eq:weighted-KM-Vr},
in further accounting for sampling variation from $\hat\gamma$.

Second, while Proposition~\ref{pro:AIPW-CAL} assumes that PS model \eqref{eq:PS-model} is correctly specified,
Proposition~\ref{pro:AIPW-CAL-lin} shows that the variance estimator from \eqref{eq:AIPW-CAL-Vr}
is nonparametrically consistent for the asymptotic variance of $\hat S_{1k,\mytext{CAL,lin}}$, even if all associated models are misspecified
and $\hat S_{1k,\mytext{CAL,lin}}$ may be inconsistent for $S_{1k}$,
converging to a population quantity different from $S_{1k}$.
Moreover, the assumption of either PS and CCP models \eqref{eq:PS-model} and \eqref{eq:const-pi} being correctly specified
or CSP and OR models \eqref{eq:const-m} and \eqref{eq:OR-model-lin} (or just CSP model \eqref{eq:const-m}) being correctly specified
ensures that $\hat S_{1k,\mytext{CAL,lin}}$ is consistent for $S_{1k}$.
Hence the point estimator $\hat S_{1k,\mytext{CAL,lin}}$ with  the variance estimator from \eqref{eq:AIPW-CAL-Vr}
can be said to enable doubly robust inference about $S_{1k}$, even though
CSP model \eqref{eq:const-m} may often be misspecified in practice.
See Tan (2020), Remark 9, for a related discussion.

To conclude this subsection, a calibrated estimator, $\hat S_{0k,\mytext{CAL}}$, can be defined for treatment-0 survival probabilities $S_{0k}$,
similarly to $\hat S_{1k,\mytext{CAL}}$ for $S_{1k}$, by replacing $A_i$ and $\pi(\cdot;\gamma)$ with $1-A_i$ and $1-\pi(\cdot;\gamma)$.
In particular, a treatment-0 calibrated estimator of $\gamma$, denoted as $\hat\gamma_{0,\mytext{CAL}}$,
is defined as a solution to the estimating equation
\begin{align}
 0 = \frac{1}{n} \sum_{i=1}^n \left\{ \frac{1-A_i}{1-\pi(X;\gamma)} -1 \right\} f(X_i) ,  \label{eq:CAL-gamma-tr0}
\end{align}
or a minimizer of the convex loss function
\begin{align*}
\ell_{\mytext{CAL}}(\gamma) = \frac{1}{n}\sum_{i=1}^n \{(1 - A_i)\me^{f^\T(X_i)\gamma} - A_if^\T(X_i)\gamma\}.
\end{align*}
If PS  model \eqref{eq:PS-model} is correctly specified, then $\hat\gamma_{1,\mytext{CAL}}$ and $\hat\gamma_{0,\mytext{CAL}}$
are expected to be consistent for the same true value $\gamma^*$. Otherwise, the two estimators may have different asymptotic limits.
In practice, comparison of the estimators $\hat\gamma_{1,\mytext{CAL}}$ and $\hat\gamma_{0,\mytext{CAL}}$ and the corresponding fitted values
$ \pi(X_i; \hat\gamma_{1,\mytext{CAL}})$ and $ \pi(X_i; \hat\gamma_{0,\mytext{CAL}})$ can be used for model diagnosis.
See Chan, Yam and Zhang (2016, Section 2.3), Vermeulen and Vansteelandt (2015, Section 5.1), and Tan (2020, Section 3.5) for related discussions.

\subsection{Calibrated estimation for hazard ratios} \label{sec:CAL-hazard}

The hazard probabilities $q_{ak}$ are related to the survival probabilities as
\begin{align}
  q_{ak} = \frac{ S_{a,k-1} - S_{ak} } {S_{a,k-1} }, \quad k=1,\ldots,K, \;a=0,1. \label{eq:q-S-relation}
\end{align}
Hence inference about individual hazard ratios $q_{1k}/q_{0k}$ can be deduced, using the delta method, from the calibrated estimators of $S_{ak}$
and associated variance estimators.
In this subsection, we mainly discuss calibrated estimation for the parameter $\breve\theta_{\mytext{wBP}}$
defined as a solution to the population version of \eqref{eq:weighted-BP},
\begin{align}
 0 = \sum_{k=1}^K \frac{ W_{1k} W_{0k} }
 { W_{1k} \me^\theta + W_{0k} } ( q_{1k} - q_{0k} \me^\theta )  ,  \label{eq:weighted-BP-pop}
\end{align}
where $W_{ak} = E [ P^{-1}(A=a|X) 1\{A=a, Y \ge u_k\} ]
  = E\{ P( Y \ge u_k | A=a,X )\}$ is the asymptotic limit of $\hat W_{ak}$ if model $\pi(\cdot;\gamma)$ is correctly specified.\footnote{\label{ft:onW}
If  $A \perp (U^{(a)}, C^{(a)}) |X$, which is more restrictive than
Assumptions \eqref{eq:no-confounding} and \eqref{eq:non-informative} combined, then $A \perp Y^{(a)} |X$ and
$W_{ak}$ reduces to $P ( Y^{(a)} \ge u_k)$, where $Y^{(a)} = \min ( U^{(a)}, C^{(a)})$.
Nevertheless, $W_{ak}$ is a well-defined population quantity under Assumptions \eqref{eq:no-confounding} and \eqref{eq:non-informative} alone,
which allows 
$A \not\perp Y^{(a)} |X$.
The equality of the two expressions for $W_{ak}$ can be verified by the law of iterated expectations (Tan 2006, Theorem 7).}
If the treatment-1 and treatment-0 hazard probabilities are proportional as in PHP model \eqref{eq:model-hazard-prob},
then the common value $\theta^*$ of the log hazard probability ratios, $\log(q_{1k}/ q_{0k})$, is identified by $\breve\theta_{\mytext{wBP}}$. Otherwise,
$\breve\theta_{\mytext{wBP}}$ represents a scalar summary of the heterogeneous log hazard probability ratios.

The existing estimator $\hat\theta_{\mytext{wBP}}$ is obtained by solving \eqref{eq:weighted-BP-pop}
after estimating $W_{ak}$ by $\hat W_{ak}$ and $q_{ak}$ by $\hat q_{ak}$. The estimator $\hat q_{ak}$ is derived from the weighted KM estimator
as prescribed by \eqref{eq:q-S-relation}, i.e.,
$ \hat q_{ak} = ( \hat S_{a,k-1,\mytext{wKM}} - \hat S_{ak,\mytext{wKM}} ) / \hat S_{a,k-1,\mytext{wKM}} $.
From this relationship and the discussion after Proposition~\ref{pro:weighted-KM}, there are two distinct scenarios.
First, if PS and CCP models \eqref{eq:PS-model} and \eqref{eq:const-pi}, in both treatment-1 and treatment-0 versions, are correctly specified, then
$\hat S_{ak,\mytext{wKM}} $ and $\hat W_{ak}$ are consistent for $S_{ak}$ and $W_{ak}$, and hence
$\hat\theta_{\mytext{wBP}}$ is consistent for $\breve\theta_{\mytext{wBP}}$.
Second, if CSP model \eqref{eq:const-m}, in both treatment-1 and treatment-0 versions,
is correctly specified but models \eqref{eq:PS-model} and \eqref{eq:const-pi} may be misspecified, then
$\hat S_{ak,\mytext{wKM}}$ is also consistent for $S_{ak}$ while $\hat W_{ak}$ may be inconsistent for $W_{ak}$
and $\hat\theta_{\mytext{wBP}}$ may be inconsistent for $\breve\theta_{\mytext{wBP}}$.
Nevertheless, in either scenario, if further PHP model \eqref{eq:model-hazard-prob} is valid, then
the true value $\theta^*$ satisfies equation \eqref{eq:weighted-BP-pop} and
$\hat\theta_{\mytext{wBP}}$ is consistent for $\theta^*$. 
In the second scenario, from the discussion in Supplement Section~\ref{sec:comparison-KM},
similar properties are also achieved by the unweighted version of $\hat\theta_{\mytext{wBP}}$.
Hence it is of main interest to consider the first scenario.

The model-robust variance estimator $\hat V_{\mytext{r}} (\hat\theta_{\mytext{wBP}})$ in \eqref{eq:weighted-BP-Vr}
is nonparametrically consistent for the asymptotic variance of $\hat\theta_{\mytext{wBP}}$,
\textit{provided} that the data-dependency of $\hat\gamma$ were removed.
There are two sources of such data-dependency,
from the difference $(\hat q_{1k} - \me^\theta\hat q_{0k})$
or from the pre-factor $\frac{ \hat W_{1k} \hat W_{0k} }{ \hat W_{1k} \me^\theta + \hat W_{0k} }$, for each $k$ in \eqref{eq:weighted-BP}.
To tackle this limitation, we exploit calibrated estimation for $\{(W_{1k}, W_{0k}): k=1,\ldots,K\}$,
in addition to that for $\{(S_{1k}, S_{0k}): k=1,\ldots,K\}$,
and develop both a consistent estimator of $\breve\theta_{\mytext{wBP}}$ and a consistent variance estimator
in the aforementioned scenario where models \eqref{eq:PS-model} and \eqref{eq:const-pi} are correctly specified.

Instead of the IPW estimator $\hat W_{ak}$, we employ an augmented IPW estimator for risk-set weight $W_{ak}$.
For $a=0,1$, consider a risk-set OR model, denoted as $\nu_{ak}(X;\beta_{ak})$, for $\nu^*_{ak}(X) = P( Y \ge u_k | A=1,X )$, for example:
\begin{align}
 \nu_{ak} (X; \beta_{ak}) =  \frac{ \exp (\beta_{ak}^\T f(X)) }{ 1+ \exp(\beta_{ak}^\T f(X)) }, \label{eq:OR-model-W}
\end{align}
where $f(X)$ is the same as in PS model \eqref{eq:PS-model} and $\beta_{ak}$ is a parameter vector.
Given any estimators $\hat\gamma$ and $\hat\beta_{1k}$, an augmented IPW estimator of $W_{1k}$ is
\begin{align}
 \hat W_{1k} ( \hat\gamma,\hat\beta_{1k} ) = \frac{1}{n} \sum_{i=1}^n \psi_{1ki} ( \hat\gamma,\hat\beta_{1k} ), \label{eq:AIPW-W}
\end{align}
with
\begin{align*}
 \psi_{1ki} ( \hat\gamma,\hat\beta_{1k}) = \frac{A_i}{\hat\pi_i} 1\{Y_i \ge u_k\} - \left( \frac{A_i}{\hat\pi_i}-1 \right) \hat\nu_{1ki},
\end{align*}
where $\hat \nu_{1ki} = \nu_{1k}(X_i; \hat\beta_{1k})$, in addition to the notation used in \eqref{eq:phi-ps}.
Our calibrated augmented IPW estimator of $W_{1k}$ is
\begin{align}
  \hat W_{1k,\mytext{CAL}} =\hat W_{1k} (\hat\gamma_{1,\mytext{CAL}}, \hat\beta_{1k,\mytext{CAL}}) ,
  \label{eq:AIPW-CAL-W}
\end{align}
where $\hat\gamma_{1,\mytext{CAL}}$ is defined in \eqref{eq:CAL-gamma} and
 $\hat\beta_{1k,\mytext{CAL}}$ is defined as a solution to
\begin{align}
 0 = \frac{1}{n} \sum_{i=1}^n A_i \frac{1-\hat\pi_{1i,\mytext{CAL}}} {\hat\pi_{1i,\mytext{CAL}} }
  \left[ 1\{Y_i\ge u_k\} - \nu_{1k} (X_i;\beta_{1k}) \right] f(X_i),  \label{eq:CAL-beta}
\end{align}
with $\hat\pi_{1i,\mytext{CAL}} = \pi(X_i; \hat\gamma_{1,\mytext{CAL}})$ as in \eqref{eq:CAL-alpha}.
Similarly to $\hat{\gamma}_{1, \mytext{CAL}}$ and $\hat{\alpha}_{1k, \mytext{CAL}}$,  $\hat\beta_{1k,\mytext{CAL}}$ can also be equivalently determined as a minimizer of the convex loss function
\begin{align}
\ell_{\mytext{CAL}}(\beta_{1k}) = \frac{1}{n}\sum_{i=1}^n A_i\frac{1 - \hat{\pi}_{1i, \mytext{CAL}}}{\hat{\pi}_{1i, \mytext{CAL}}}[1\{Y_i \geq u_k\}f^\T(X_i)\beta_{1k} - \log\{1 + \me^{f^\T(X_i)\beta_{1k}}\} ]. \label{eq:loss-CAL-beta}
\end{align}
The estimator $\hat W_{1k,\mytext{CAL}}$, similar to $\hat S_{1k,\mytext{CAL}}$,
is an application of the method in Tan (2020) with PS model \eqref{eq:PS-model}
and OR model \eqref{eq:OR-model-W} for the binary outcome $1\{Y_i \ge u_k\}$.
A technical relaxation is that $A$ and $Y^{(a)}$ may be confounded given $X$, and hence
$W_{ak}$ may differ from $P ( Y^{(a)} \ge u_k)$ (See footnote~\ref{ft:identification}).
Our calibrated augmented IPW estimator, $\hat W_{0k,\mytext{CAL}}$, can be similarly defined for $W_{0k}$
by replacing $A_i$ and $\pi(\cdot;\gamma)$ with $1-A_i$ and $1-\pi(\cdot;\gamma)$.

Our calibrated estimator, denoted as $\hat\theta_{\mytext{CAL}}$, for $\breve\theta_{\mytext{wBP}}$ is defined as a solution to
\begin{align}
 0 = \sum_{k=1}^K \left\{  \frac{ \hat W_{1k,\mytext{CAL}} \hat W_{0k,\mytext{CAL}} }
 { \hat W_{1k,\mytext{CAL}} \me^\theta + \hat W_{0k,\mytext{CAL}} } ( \hat q_{1k,\mytext{CAL}} - \hat q_{0k,\mytext{CAL}} \me^\theta ) \right\} ,
 \label{eq:CAL-theta}
\end{align}
with $\hat q_{ak,\mytext{CAL}} = (\hat S_{a,k-1,\mytext{CAL}} - \hat S_{ak,\mytext{CAL}} )/\hat S_{a,k-1,\mytext{CAL}} $ by \eqref{eq:q-S-relation},
where $\hat S_{ak,\mytext{CAL}}$ and $\hat W_{ak,\mytext{CAL}}$ are calibrated augmented IPW estimators,
for example, in \eqref{eq:AIPW-CAL} and \eqref{eq:AIPW-CAL-W} for $a=1$.
The point estimator $\hat\theta_{\mytext{CAL}}$ can be shown to be doubly robust for $\breve\theta_{\mytext{wBP}}$, i.e., being consistent
in two scenarios: either PS and CCP models \eqref{eq:PS-model} and \eqref{eq:const-pi}
or CSP and OR models \eqref{eq:const-m}, \eqref{eq:OR-model}, and \eqref{eq:OR-model-W}
(or just CSP and OR models \eqref{eq:const-m} and \eqref{eq:OR-model-W})
are correctly specified, including both treatment-1 and treatment-0 versions.
Moreover, as an improvement over $\hat\theta_{\mytext{wBP}}$ in the first scenario,
if model \eqref{eq:PS-model} alone is correctly specified, then
consistent variance estimation can be obtained for $\hat\theta_{\mytext{CAL}}$ by using the asymptotic expansions (or influence functions)
for $\hat S_{ak,\mytext{CAL}}$ and $\hat W_{ak,\mytext{CAL}}$. Let
\begin{align*}
 & \hat \varphi_{ki} (\theta) =
  \frac{ \hat W_{0k,\mytext{CAL}} \hat W_{1k,\mytext{CAL}}} { \hat W_{1k,\mytext{CAL}} \me^\theta + \hat W_{0k,\mytext{CAL}} }
  \left(\frac{ -\hat\varphi_{1ki} + (1-\hat q_{1k,\mytext{CAL}}) \hat\varphi_{1,k-1,i} }{ \hat S_{1,k-1,\mytext{CAL}} } \right) \\
  & \quad - \frac{ \hat W_{0k,\mytext{CAL}} \hat W_{1k,\mytext{CAL}} \me^\theta  } { \hat W_{1k,\mytext{CAL}} \me^\theta + \hat W_{0k,\mytext{CAL}} }
  \left( \frac{ - \hat\varphi_{0ki} + (1-\hat q_{0k,\mytext{CAL}}) \hat\varphi_{0,k-1,i} }{ \hat S_{0,k-1,\mytext{CAL}} } \right)  \\
  & \quad + \frac{\hat q_{1k,\mytext{CAL}} -  \hat q_{0k,\mytext{CAL}} \me^\theta }{(\hat W_{1k,\mytext{CAL}} \me^\theta + \hat W_{0k,\mytext{CAL}})^2 }
 \left\{ \hat W_{0k,\mytext{CAL}}^2 (\hat\psi_{1ki} - \hat W_{1k,\mytext{CAL}} ) +
 \me^\theta \hat W_{1k,\mytext{CAL}}^2 (\hat\psi_{0ki} - \hat W_{0k,\mytext{CAL}} ) \right\} ,
\end{align*}
where $\hat\varphi_{1ki} = \varphi_{1ki}(\hat\gamma_{1,\mytext{CAL}}, \hat\rho_{1,1:k,\mytext{CAL}}, \hat\eta_{1,1:k,\mytext{CAL}}, \hat\alpha_{1k,\mytext{CAL}})$
as in Proposition~\ref{pro:AIPW-CAL},
$\hat\psi_{1ki} =  \psi_{1ki} ( \hat\gamma_{1,\mytext{CAL}},$ $\hat\beta_{1k,\mytext{CAL}}) $,
and $\hat\varphi_{0ki}$ and $\hat \psi_{0ki}$ are similarly defined.

\begin{pro} \label{pro:CAL-theta}
Suppose that PS model \eqref{eq:PS-model} is correctly specified.
Then under standard regularity conditions as $n\to\infty$ and $\dim(f)$ and $K$ are fixed,
$\hat\theta_{\mytext{CAL}}$ satisfies the asymptotic expansion
\begin{align}
 \hat\theta_{\mytext{CAL}} - \breve\theta_{\mytext{CAL}}
 = H_{\mytext{CAL}} (\theta)^{-1}\;  \sum_{k=1}^K \breve\varphi_{ki} (\theta)  \Big|_{\breve\theta_{\mytext{CAL}}} + o_p(n^{-1/2}), \label{eq:CAL-theta-expan}
\end{align}
where $\breve\theta_{\mytext{CAL}}$ is the limit value defined as the solution of \eqref{eq:weighted-BP-pop} with
$W_{ak}$ and $q_{ak}$ replaced by the limit values of $\hat W_{ak,\mytext{CAL}}$ and $\hat q_{ak,\mytext{CAL}}$,
$\breve\varphi_{ki} (\theta)$ is the limit version of $\hat\varphi_{ki} (\theta)$,
and $H_{\mytext{CAL}} (\theta)$ is the limit version of $\hat H_{\mytext{CAL}} (\theta)$ below  (see the proof in Supplement Section~\ref{sec:technical-detail} for details).
Moreover, the asymptotic variance of $\hat\theta_{\mytext{CAL}}$ can be consistently estimated by
\begin{align}
\hat V_{\mytext{r}} ( \hat\theta_{\mytext{CAL}})=
\hat H_{\mytext{CAL}} (\theta)^{-1} \hat G_{\mytext{CAL}} (\theta )
\hat H_{\mytext{CAL}} (\theta)^{-1} \Big|_{\breve\theta_{\mytext{CAL}}} ,  \label{eq:CAL-theta-Vr}
\end{align}
where $\hat G_{\mytext{CAL}} (\theta ) = n^{-2} \sum_{i=1}^n \{\sum_{k=1}^K \hat \varphi_{ki} (\theta) \}^2 $ and
\begin{align*}
\hat H_{\mytext{CAL}} (\theta) = \sum_{k=1}^K \frac{\hat W_{1k,\mytext{CAL}} \hat W_{0k,\mytext{CAL}} \me^\theta}
 { ( \hat W_{1k,\mytext{CAL}} \me^\theta + \hat W_{0k,\mytext{CAL}} )^2 }
 ( \hat W_{1k,\mytext{CAL}} \hat q_{1k,\mytext{CAL}}+ \hat W_{0k,\mytext{CAL}}\hat q_{0k,\mytext{CAL}} ).
\end{align*}
\end{pro}

\vspace{.05in}
By Proposition~\ref{pro:CAL-theta},
if PS and CCP models \eqref{eq:PS-model} and \eqref{eq:const-pi}, in both treatment-1 and treatment-0 versions, are correctly specified, then
$\hat\theta_{\mytext{CAL}}$ is consistent for $\breve\theta_{\mytext{wBP}}$ (i.e., $\breve\theta_{\mytext{CAL}}=\breve\theta_{\mytext{wBP}})$, and
an asymptotic $(1-c)$-confidence interval for $\breve\theta_{\mytext{wBP}}$ is
$ \hat \theta_{\mytext{CAL}} \pm z_{c/2} \sqrt{\hat V_{\mytext{r}}( \hat \theta_{\mytext{CAL}} ) }$.
The confidence interval is also of asymptotic level $1-c$ for the true value $\theta^*$ if PHP model \eqref{eq:model-hazard-prob} is valid
about homogeneous hazard probability ratios.
In particular, if models \eqref{eq:PS-model} and \eqref{eq:const-pi} are correctly specified, then
the Wald rejection region, $ |\hat \theta_{\mytext{CAL}}| \ge z_{c/2}  \sqrt{\hat V_{\mytext{r}}( \hat \theta_{\mytext{CAL}} ) } $,
achieves an asymptotic level $c$ for testing the null hypothesis $\theta^*=0$, i.e., the distributions of $U^{(1)}$ and $U^{(0)}$ are the same.
This test is more desirable than the weighted log-rank test in Xie \& Liu (2005),
which relies on CSP model \eqref{eq:const-m}, in both treatment-1 and treatment-0 versions, being correctly specified.
In the latter scenario, the unweighted log-rank test is also valid similarly to unweighted KM estimation as discussed in
Supplement Section~\ref{sec:comparison-KM}.

Similarly to the replacement of logistic regression \eqref{eq:OR-model}
by linear regression \eqref{eq:OR-model-lin} in calibrated estimation of $S_{1k}$,
the logistic OR model \eqref{eq:OR-model-W} can also be replaced by a linear OR model in calibrated estimation of $W_{1k}$:
\begin{align}
 \nu_{ak} (X; \beta_{ak}) = \beta_{ak}^\T f(X) . \label{eq:OR-model-W-lin}
\end{align}
The corresponding estimator, for example $\hat W_{1k,\mytext{CAL,lin}}$, is defined similarly as $\hat W_{1k,\mytext{CAL}}$ in \eqref{eq:AIPW-CAL-W},
where $\hat\beta_{1k,\mytext{CAL}}$, defined as a solution to \eqref{eq:CAL-beta}
or a minimizer of
\begin{align}
\ell_{\mytext{CAL}}(\beta_{1k}) = \frac{1}{n}\sum_{i=1}^n A_i\frac{1 - \hat{\pi}_{1i, \mytext{CAL}}}{\hat{\pi}_{1i, \mytext{CAL}}}
 [1\{Y_i \geq u_k\}- f^\T(X_i)\beta_{1k}]^2/2 ,  \label{eq:loss-CAL-beta-lin}
\end{align}
becomes a weighted least-squares estimator.
Our calibrated estimator, $\hat\theta_{\mytext{CAL,lin}}$,
is defined as a solution to \eqref{eq:CAL-theta}, with $\hat S_{ak,\mytext{CAL}}$ and $\hat W_{ak,\mytext{CAL}}$
replaced by $\hat S_{ak,\mytext{CAL,lin}}$ and $\hat W_{ak,\mytext{CAL,lin}}$.
In addition to $\hat S_{ak,\mytext{CAL,lin}} $ reducing to $\hat S_{ak,\mytext{wKM}} (\hat\gamma_{a,\mytext{CAL}})$ as in \eqref{eq:AIPW-reduction},
the estimator $\hat W_{ak,\mytext{CAL,lin}}$ also reduces to $\hat W_{ak} (\hat\gamma_{a,\mytext{CAL}})$, i.e.,
the IPW estimator $\hat W_{ak} $ with $\hat\gamma=\hat\gamma_{a,\mytext{CAL}}$.
Hence $\hat\theta_{\mytext{CAL,lin}}$ reduces to $\hat\theta_{\mytext{wBP}}$ based on $\hat\gamma_{a,\mytext{CAL}}$, i.e.,
\begin{align}
 \hat\theta_{\mytext{CAL,lin}} = \hat\theta_{\mytext{wBP}} (\hat\gamma_{a,\mytext{CAL}}), \label{eq:theta-reduction}
\end{align}
where $\hat\theta_{\mytext{wBP}}(\hat\gamma_{a,\mytext{CAL}})$ is $\hat\theta_{\mytext{wBP}}$ in Section~\ref{sec:existing-wBP} with $\hat w_{ai}$ defined from
$\hat\gamma=\hat\gamma_{a,\mytext{CAL}}$ for $a=0,1$.
Moreover, Proposition~\ref{pro:CAL-theta} can be extended such that the variance estimator from \eqref{eq:CAL-theta-Vr}
is nonparametrically consistent for the asymptotic variance of $\hat\theta_{\mytext{CAL,lin}}$
even if PS model \eqref{eq:PS-model} is misspecified.
For this reason, the point estimator $\hat \theta_{\mytext{CAL,lin}}$ with the variance estimator from \eqref{eq:CAL-theta-Vr}
enables doubly robust inference about $\breve\theta_{\mytext{wBP}}$
if either PS and CCP models \eqref{eq:PS-model} and \eqref{eq:const-pi} or
CSP and OR models \eqref{eq:const-m}, \eqref{eq:OR-model-lin}, and \eqref{eq:OR-model-W-lin}
(or just CSP and OR models \eqref{eq:const-m} and \eqref{eq:OR-model-W-lin})
are correctly specified.

\section{High-dimensional estimation} \label{sec:high-dimension}

We propose regularized calibrated estimation for survival probabilities $S_{ak}$ and
the hazard-ratio parameter $\breve{\theta}_{\mytext{wBP}}$ in high-dimensional settings, where the dimension of $f(X)$ is close to or exceeds the sample size.
In particular, to be computationally efficient, we develop a novel regularized calibrated estimator of $\breve{\theta}_{\mytext{wBP}}$ through a ``linearization'' technique
instead of directly incorporating regularization at each time period.

\textbf{Estimation of survival probabilities.}\;
Given PS, CCP, CSP, and OR models \eqref{eq:PS-model}, \eqref{eq:const-pi}, \eqref{eq:const-m} and \eqref{eq:OR-model},
our regularized calibrated estimator for $S_{ak}$ is defined as $\hat{S}_{1k, \mytext{RCAL}} = \hat{S}_{1k}(\hat{\gamma}_{1, \mytext{RCAL}}, \hat{\rho}_{1, 1:k, \mytext{RCAL}},$ $ \hat{\eta}_{1, 1:k, \mytext{RCAL}}, \hat{\alpha}_{1k, \mytext{RCAL}})$, where $(\hat{\gamma}_{1, \mytext{RCAL}}, \hat{\alpha}_{1k, \mytext{RCAL}})$ are regularized versions of the calibrated estimators $(\hat{\gamma}_{1, \mytext{CAL}}, \hat{\alpha}_{1k, \mytext{CAL}})$, and $(\hat{\rho}_{1j, \mytext{RCAL}}, \hat{\eta}_{1j, \mytext{RCAL}})$ are solutions of \eqref{eq:CAL-rho} and \eqref{eq:CAL-eta} with $\hat{\pi}_i = \pi(X_i; \hat{\gamma}_{1, \mytext{RCAL}})$ for $j = 1, \ldots, k$.
The estimator $\hat{\gamma}_{1, \mytext{RCAL}}$ is a regularized calibrated estimator of $\gamma$ (Tan 2020), defined as a minimizer of the penalized loss function $\ell_{\mytext{CAL}}(\gamma) + \lambda \lVert \gamma_{1:p} \rVert_1$, where $\ell_{\mytext{CAL}}(\gamma)$ is defined in \eqref{eq:loss-CAL-gamma}, $\lambda \geq 0$ is a tuning parameter, $\lVert \cdot \rVert_1$ denotes the $L_1$ norm, and $\gamma_{1:p} = (\gamma_1, \ldots, \gamma_p)^\T$ excludes the intercept $\gamma_0$.
Similarly, $\hat{\alpha}_{1k, \mytext{RCAL}}$ is defined as a minimizer of the penalized loss function $\ell_{\mytext{CAL}}(\alpha_{1k}) + \lambda_k \lVert \alpha_{1k, 1:p} \rVert_1$, where $\ell_{\mytext{CAL}}(\alpha_{1k})$ is defined in \eqref{eq:loss-CAL-alpha}, but with $\hat{\pi}_{1i, \mytext{CAL}}$ and  $\hat{U}_{ki, \mytext{CAL}}^{(1)}$ replaced by $\hat{\pi}_{1i, \mytext{RCAL}} = \pi(X_i; \hat{\gamma}_{1, \mytext{RCAL}})$ and $\hat{U}_{ki, \mytext{RCAL}}^{(1)} = \hat{U}_{ki}^{(1)}(\hat{\rho}_{1, 1:k, \mytext{RCAL}}, \hat{\eta}_{1, 1:k, \mytext{RCAL}})$.
With linear OR model \eqref{eq:OR-model-lin}, the regularized version of $\hat{S}_{1k, \mytext{CAL, lin}}$, denoted as $\hat{S}_{1k, \mytext{RCAL, lin}}$, is defined in the same way as $\hat{S}_{1k, \mytext{RCAL}}$, except that the loss function for estimating $\hat{\alpha}_{1k, \mytext{RCAL}}$ is \eqref{eq:loss-CAL-alpha-lin} instead of  \eqref{eq:loss-CAL-alpha}. Under suitable regularity and sparsity conditions
in the high-dimensional setting, it can be shown as in Tan (2020) that
$\hat{S}_{1k, \mytext{RCAL}}$ and $\hat{S}_{1k, \mytext{RCAL, lin}}$ satisfy similar properties as in Propositions \ref{pro:AIPW-CAL} and \ref{pro:AIPW-CAL-lin}
for their non-regularized counterparts $\hat{S}_{1k, \mytext{CAL}}$ and $\hat{S}_{1k, \mytext{CAL, lin}}$ respectively. For treatment-0, a regularized calibrated estimator, $\hat{S}_{0k, \mytext{RCAL}}$, can be defined for $S_{0k}$, similarly to $\hat{S}_{1k, \mytext{RCAL}}$, by replacing $A_i$ and $\pi(\cdot; \gamma)$ with $1 - A_i$ and $1 - \pi(\cdot; \gamma)$.

\textbf{Estimation of hazard-ratio parameter (first method).}\; We discuss two methods for regularized estimation of $\breve{\theta}_{\mytext{wBP}}$.
The first method incorporates regularization at each time period and
proceeds similarly as in Section~\ref{sec:CAL-hazard} for obtaining $\hat{\theta}_{\mytext{CAL}}$.
Specifically, we first compute the regularized calibrated augmented IPW estimator for $W_{ak}$,  denoted as $\hat{W}_{ak, \mytext{RCAL}}$, and then define the regularized calibrated estimator $\hat{\theta}_{\mytext{RCAL}}$ for $\breve{\theta}_{\mytext{wBP}}$ as a solution to \eqref{eq:CAL-theta}, replacing $\hat{S}_{ak, \mytext{CAL}}$ and $\hat{W}_{ak, \mytext{CAL}}$ with $\hat{S}_{ak, \mytext{RCAL}}$ and $\hat{W}_{ak, \mytext{RCAL}}$.
Given PS model \eqref{eq:PS-model} and logistic OR model \eqref{eq:OR-model-W}, $\hat{W}_{1k, \mytext{RCAL}}$ is defined as $\hat{W}_{1k, \mytext{RCAL}} = \hat{W}_{1k}(\hat{\gamma}_{1, \mytext{RCAL}}, \hat{\beta}_{1k, \mytext{RCAL}})$, where $\hat{\beta}_{1k, \mytext{RCAL}}$ is the minimizer of the penalized loss function $\ell_{\mytext{CAL}}(\beta_{1k}) + \lambda_k \lVert \beta_{1k, 1:p} \rVert_1$, with $\ell_{\mytext{CAL}}(\beta_{1k})$ defined in \eqref{eq:loss-CAL-beta}.
A regularized calibrated estimator, $\hat{W}_{0k, \mytext{RCAL}}$, can be defined for $W_{0k}$, similarly to $\hat{W}_{1k, \mytext{RCAL}}$, by replacing $A_i$ and $\pi(\cdot; \gamma)$ with $1 - A_i$ and $1 - \pi(\cdot; \gamma)$.
With linear OR model \eqref{eq:OR-model-W-lin}, the regularized version of $\hat{W}_{1k, \mytext{CAL, lin}}$, denoted as $\hat{W}_{1k, \mytext{RCAL, lin}}$, is defined in the same way as $\hat{W}_{1k, \mytext{RCAL}}$, except that the loss function for estimating $\hat{\beta}_{1k, \mytext{RCAL}}$ is \eqref{eq:loss-CAL-beta-lin} instead of  \eqref{eq:loss-CAL-beta}.
For treatment-0, $\hat{W}_{0k, \mytext{CAL, lin}}$ is similarly defined.
The regularized calibrated estimator $\hat{\theta}_{\mytext{RCAL, lin}}$ is defined as a solution to \eqref{eq:CAL-theta}, with $\hat{S}_{ak, \mytext{CAL}}$ and $\hat{W}_{ak, \mytext{CAL}}$ replaced by $\hat{S}_{ak, \mytext{RCAL, lin}}$ and $\hat{W}_{ak, \mytext{RCAL, lin}}$.
Then Proposition~\ref{pro:CAL-theta} and related discussion can be extended to $\hat{\theta}_{\mytext{RCAL}}$  and  $\hat{\theta}_{\mytext{RCAL, lin}}$
under suitable regularity and sparsity conditions in the high-dimensional setting.

The first method is time-consuming for two reasons. First, it requires fitting $4K$ distinct OR models
\eqref{eq:OR-model} and \eqref{eq:OR-model-W}, $\mu_{ak}(X;\alpha_{ak})$ and $\nu_{ak}(X;\beta_{ak})$ for $a=0,1$ and $k=1,\ldots,K$.
Second,  the method estimates each coefficient vector $\alpha_{ak}$ or $\beta_{ak}$ separately, which is challenging in high-dimensional settings due to the need for selecting tuning parameters (e.g., through cross-validation).
The second issue can be addressed by adopting group Lasso
(Yuan \& Lin, 2006) to estimate $(\alpha_{ak}: k = 1, \ldots, K)$ or $(\beta_{ak}: k = 1, \ldots, K)$ simultaneously for $a=0,1$.
For example, for treatment-1, a group regularized calibrated estimators $\hat{\alpha}_{1, \mytext{RCALg}} = (\hat{\alpha}_{1k, \mytext{RCALg}}: k =1, \ldots, K)$ for $(\alpha_{1k}: k = 1, \ldots, K)$ is jointly as a minimizer to $\sum_{k=1}^K \ell_{\mytext{CAL}}(\alpha_{1k}) + \lambda \sum_{j=1}^p \lVert \alpha_{1\cdot j} \rVert_2$,
where $\alpha_{1\cdot j} = (\alpha_{1k, j} : k = 1, \ldots, K)^\T$ consists of $K$ coefficients associated with the covariate term $f_j(X)$  across $k$.
However, the problem of estimating $4K$ coefficient vectors remains computationally expensive and statistically undesirable.

\textbf{Estimation of hazard-ratio parameter (second method).}\;
To address the computational difficulty of the first method,
we develop a novel regularized calibrated estimator for $\breve{\theta}_{\mytext{wBP}}$
by combining OR-based augmentation terms across time periods before incorporating regularization estimation with OR models.

First, consider the weighted-only regularized calibrated estimators for $S_{1k}$ and $W_{1k}$, defined as follows:
\begin{align}
\hat{S}_{1k, \mytext{RCw}} = \frac{1}{n}\sum_{i=1}^n \frac{A_i}{\hat{\pi}_{1i,\mytext{RCAL}}}\hat{U}_{ki,\mytext{RCAL}}^{(1)}, \quad
\hat{W}_{1k, \mytext{RCw}} = \frac{1}{n}\sum_{i=1}^n \frac{A_i}{\hat{\pi}_{1i, \mytext{RCAL}}}1\{Y_i \geq u_k\} ,
\label{eq:IPW-RCAL}
\end{align}
by using the inverse PS weighted terms, without the OR-based augmentation terms, in $\hat{S}_{1k, \mytext{RCAL}}$ and $\hat{W}_{1k, \mytext{RCAL}}$.
The weighted-only regularized calibrated estimators $\hat{S}_{0k, \mytext{RCw}}$ and $\hat{W}_{0k, \mytext{RCw}}$ for $S_{0k}$ and $W_{0k}$
are defined similarly to $\hat{S}_{1k, \mytext{RCw}}$ and $\hat{W}_{1k, \mytext{RCw}}$
by replacing $A_i$ and $\pi(\cdot; \gamma)$ with $1 - A_i$ and $1 - \pi(\cdot; \gamma)$.
By substituting $(\hat{S}_{ak, \mytext{RCw}}, \hat{W}_{ak, \mytext{RCw}})$ for $(\hat{S}_{ak, \mytext{CAL}}, \hat{W}_{ak, \mytext{CAL}})$
in \eqref{eq:CAL-theta}, consider the weighted-only regularized calibrated estimator, $\hat\theta_{\mytext{RCw}}$, defined as a solution to
\begin{align}
 0 = \sum_{k=1}^K \left\{  \frac{ \hat W_{1k,\mytext{RCw}} \hat W_{0k,\mytext{RCw}} }
 { \hat W_{1k,\mytext{RCw}} \me^\theta + \hat W_{0k,\mytext{RCw}} } ( \hat q_{1k,\mytext{RCw}} - \hat q_{0k,\mytext{RCw}} \me^\theta ) \right\} ,
  \label{eq:RCAL-IPW-theta}
\end{align}
with $\hat q_{ak,\mytext{RCw}} = (\hat S_{a,k-1,\mytext{RCw}} - \hat S_{ak,\mytext{RCw}} )/\hat S_{a,k-1,\mytext{RCw}} $.
To derive another augmented estimator based on $\hat\theta_{\mytext{RCw}}$,
we use a ``linearization'' of the estimating function in \eqref{eq:RCAL-IPW-theta}, which are nonlinear functionals of
the weighted-only estimators $(\hat{S}_{ak, \mytext{RCw}}, \hat{W}_{ak, \mytext{RCw}})$, $k=1,\ldots,K, a=0,1$.

In Supplement Section~\ref{sec:technical-detail}, we show that the following asymptotic expansion holds:
 \begin{align}
& \sum_{k=1}^K \frac{\hat{W}_{1k, \mytext{RCw}}\hat{W}_{0k, \mytext{RCw}}}{\hat{W}_{1k, \mytext{RCw}}\me^{\theta} + \hat{W}_{0k, \mytext{RCw}}}(\hat{q}_{1k, \mytext{RCw}} - \hat{q}_{0k, \mytext{RCw}}\me^{\theta})  \nonumber \\
& = \frac{1}{n}\sum_{i=1}^n \frac{A_i}{\hat{\pi}_{1i, \mytext{RCAL}}}B_{1i}(\theta, \breve{S}_{1, \mytext{RCw}}, \breve{W}_{\mytext{RCw}}) - \frac{1}{n}\sum_{i=1}^n \frac{1 - A_i}{1 - \hat{\pi}_{0i, \mytext{RCAL}}}B_{0i}(\theta, \breve{S}_{0, \mytext{RCw}}, \breve{W}_{\mytext{RCw}}) + o_P(n^{-1/2}), \label{eq:expand}
\end{align}
where
 $\breve q_{ak,\mytext{RCw}} = (\breve S_{a,k-1,\mytext{RCw}} - \breve S_{ak,\mytext{RCw}} )/\breve S_{a,k-1,\mytext{RCw}} $,
 $\breve{S}_{ak, \mytext{RCw}}$ and $\breve{W}_{ak, \mytext{RCw}}$ are the limit values of $\hat{S}_{ak, \mytext{RCw}}$ and $\hat{W}_{ak, \mytext{RCw}}$,
 $\breve{S}_{a, \mytext{RCw}} = \{\breve{S}_{ak, \mytext{RCw}}: k = 0, 1, \ldots, K\}$ ($\breve{S}_{a0, \mytext{RCw}} \equiv 1$), $\breve{W}_{\mytext{RCw}} = \{\breve{W}_{ak, \mytext{RCw}}: k = 1, \ldots, K, a = 0, 1\}$, and
\begin{align*}
B_{1i}(\theta, \breve{S}_{1, \mytext{RCw}}, \breve{W}_{\mytext{RCw}})  = & \sum_{k=1}^K \frac{\breve{W}_{1k, \mytext{RCw}}\breve{W}_{0k, \mytext{RCw}}}{( \breve{W}_{1k, \mytext{RCw}}\me^{\theta} + \breve{W}_{0k, \mytext{RCw}} )\breve{S}_{1, k-1, \mytext{RCw}} }\{ (1 - \breve{q}_{1k, \mytext{RCw}})\hat{U}_{k-1, i}^{(1)} - \hat{U}_{ki}^{(1)}\}   \\
& + \sum_{k=1}^K \frac{\breve{q}_{1k, \mytext{RCw}} - \breve{q}_{0k, \mytext{RCw}}\me^{\theta}}{(\breve{W}_{1k, \mytext{RCw}}\me^{\theta} + \breve{W}_{0k, \mytext{RCw}})^2}\breve{W}_{0k, \mytext{RCw}}^21\{ Y_i \geq u_k\}, \\
B_{0i}(\theta, \breve{S}_{0, \mytext{RCw}}, \breve{W}_{\mytext{RCw}})  = & \me^{\theta}\sum_{k=1}^K \frac{\breve{W}_{1k, \mytext{RCw}}\breve{W}_{0k, \mytext{RCw}}}{( \breve{W}_{1k, \mytext{RCw}}\me^{\theta} + \breve{W}_{0k, \mytext{RCw}} )\breve{S}_{0, k-1, \mytext{RCw}} }\{ (1 - \breve{q}_{0k, \mytext{RCw}})\hat{U}_{k-1, i}^{(0)} - \hat{U}_{ki}^{(0)}\}   \\
& - \me^{\theta}\sum_{k=1}^K \frac{\breve{q}_{1k, \mytext{RCw}} - \breve{q}_{0k, \mytext{RCw}}\me^{\theta}}{(\breve{W}_{1k, \mytext{RCw}}\me^{\theta} + \breve{W}_{0k, \mytext{RCw}})^2}\breve{W}_{1k, \mytext{RCw}}^21\{ Y_i \geq u_k\}.
\end{align*}
A key point is that, by direct calculation using the definitions \eqref{eq:IPW-RCAL},
estimating equation \eqref{eq:RCAL-IPW-theta} for $\hat{\theta}_{\mytext{RCw}}$ is equivalent to
setting to $0$ the leading term in \eqref{eq:expand}
with $(\breve{S}_{a, \mytext{RCw}}, \breve{W}_{\mytext{RCw}})$ replaced by $(\hat{S}_{a, \mytext{RCw}}, \hat{W}_{\mytext{RCw}})$:
\begin{align}
0 = \frac{1}{n}\sum_{i=1}^n \frac{A_i}{\hat{\pi}_{1i, \mytext{RCAL}}}B_{1i}(\theta, \hat{S}_{1, \mytext{RCw}}, \hat{W}_{\mytext{RCw}}) - \frac{1}{n}\sum_{i=1}^n \frac{1 - A_i}{1 - \hat{\pi}_{0i, \mytext{RCAL}}}B_{0i}(\theta, \hat{S}_{0, \mytext{RCw}}, \hat{W}_{\mytext{RCw}}), \label{eq:IPW-theta}
\end{align}
where $\hat{S}_{a, \mytext{RCw}} = \{\hat{S}_{ak, \mytext{RCw}}: k = 0, 1, \ldots, K\}$ ($\hat{S}_{a0, \mytext{RCw}} \equiv 1$) and $\hat{W}_{\mytext{RCw}} = \{\hat{W}_{ak, \mytext{RCw}}: k = 1, \ldots, K, a = 0, 1\}$.
Hence $\hat{\theta}_{\mytext{RCw}}$ is equivalently a solution to \eqref{eq:IPW-theta}.

By incorporating an OR-based augmentation term into \eqref{eq:IPW-theta},
we propose a new augmented regularized calibrated estimator for $\breve{\theta}_{\mytext{wBP}}$. For $a = 0, 1$,
consider an OR model, denoted as $g_a(x; \zeta_a)$, for the ``mean outcome'' $g_a^*(x) = E\{B_{ai}(\breve\theta_{\mytext{RCw}},\breve{S}_{a, \mytext{RCw}},  \breve{W}_{\mytext{RCw}}) | X_i=x)\}$, for example:
\begin{align}
g_a(X; \zeta_a) = \zeta_a^\T f(X), \label{eq:combine-OR-model}
\end{align}
where $\breve\theta_{\mytext{RCw}}$ is the limit value defined as the solution of \eqref{eq:weighted-BP-pop} with
$W_{ak}$ and $q_{ak}$ replaced by $\breve W_{ak,\mytext{RCw}}$ and $\breve q_{ak,\mytext{RCw}}$, the limit values
of $\hat W_{ak,\mytext{RCw}}$ and $\hat q_{ak,\mytext{RCw}}$.
Then our augmented regularized calibrated estimator, denoted as $\hat{\theta}_{\mytext{RCa}}$, for $\breve{\theta}_{\mytext{wBP}}$ is the solution of the following equation
\begin{align}
0 = & \frac{1}{n}\sum_{i=1}^n \left\{\frac{A_i}{\hat{\pi}_{1i, \mytext{RCAL}}}B_{1i}(\theta, \hat{S}_{1, \mytext{RCw}}, \hat{W}_{\mytext{RCw}}) - \left(\frac{A_i}{\hat{\pi}_{1i, \mytext{RCAL}}} - 1\right)\hat{g}_{1i, \mytext{RCAL}}\right\} \nonumber \\
& - \frac{1}{n}\sum_{i=1}^n \left\{\frac{1 - A_i}{1 - \hat{\pi}_{0i, \mytext{RCAL}}}B_{0i}(\theta, \hat{S}_{0, \mytext{RCw}}, \hat{W}_{\mytext{RCw}})  - \left(\frac{1 - A_i}{1 - \hat{\pi}_{0i, \mytext{RCAL}}} - 1\right)\hat{g}_{0i, \mytext{RCAL}}\right\} \nonumber \\
= & \sum_{k=1}^K \left\{  \frac{ \hat W_{1k,\mytext{RCw}} \hat W_{0k,\mytext{RCw}} }
 { \hat W_{1k,\mytext{RCw}} \me^\theta + \hat W_{0k,\mytext{RCw}} } ( \hat q_{1k,\mytext{RCw}} - \hat q_{0k,\mytext{RCw}} \me^\theta ) \right\} \nonumber \\
& - \frac{1}{n}\sum_{i=1}^n \left\{\left(\frac{A_i}{\hat{\pi}_{1i, \mytext{RCAL}}} - 1\right)\hat{g}_{1i, \mytext{RCAL}} - \left(\frac{1 - A_i}{1 - \hat{\pi}_{0i, \mytext{RCAL}}} - 1\right)\hat{g}_{0i, \mytext{RCAL}}\right\} ,
 \label{eq:RCAL-AIPW-theta}
\end{align}
where $\hat{g}_{ai, \mytext{RCAL}} = g_a(X_i; \hat{\zeta}_{a, \mytext{RCAL}})$ and $\hat{\zeta}_{a, \mytext{RCAL}}$ is a minimizer of the penalized loss function $\ell(\zeta_a) + \lambda \lVert \zeta_{a, 1:p} \rVert_1$ with
\begin{align*}
& \ell(\zeta_1) = \frac{1}{2n}\sum_{i=1}^n A_i\frac{1 - \hat{\pi}_{1i, \mytext{RCAL}}}{\hat{\pi}_{1i, \mytext{RCAL}}}\{B_{1i}(\hat{\theta}_{\mytext{RCw}}, \hat{S}_{1, \mytext{RCw}}, \hat{W}_{\mytext{RCw}}) - f^\T(X)\zeta_1\}^2, \\
& \ell(\zeta_0) = \frac{1}{2n}\sum_{i=1}^n (1 - A_i)\frac{\hat{\pi}_{0i, \mytext{RCAL}}}{1 - \hat{\pi}_{0i, \mytext{RCAL}}}\{B_{0i}(\hat{\theta}_{\mytext{RCw}}, \hat{S}_{0, \mytext{RCw}}, \hat{W}_{\mytext{RCw}}) - f^\T(X)\zeta_0\}^2.
\end{align*}
The second equality in \eqref{eq:RCAL-AIPW-theta} follows from the equivalence between \eqref{eq:RCAL-IPW-theta} and \eqref{eq:IPW-theta}.
By similar reasoning as in Ghosh and Tan (2022), the following result can be obtained.
Let $\hat\delta_i (\theta)$ be the $i$th summand in \eqref{eq:RCAL-AIPW-theta}, i.e.,
\begin{align*}
\hat\delta_i(\theta) = & \left\{\frac{A_i}{\hat{\pi}_{1i, \mytext{RCAL}}}B_{1i}(\theta, \hat{S}_{1, \mytext{RCw}}, \hat{W}_{\mytext{RCw}}) - \left(\frac{A_i}{\hat{\pi}_{1i, \mytext{RCAL}}} - 1\right)\hat{g}_{1i, \mytext{RCAL}}\right\} \nonumber \\
& - \left\{\frac{1 - A_i}{1 - \hat{\pi}_{0i, \mytext{RCAL}}}B_{0i}(\theta, \hat{S}_{0, \mytext{RCw}}, \hat{W}_{\mytext{RCw}})  - \left(\frac{1 - A_i}{1 - \hat{\pi}_{0i, \mytext{RCAL}}} - 1\right)\hat{g}_{0i, \mytext{RCAL}}\right\},
\end{align*}
and let $\breve \delta_i(\theta)$ be the limit version of $\hat\delta_i(\theta)$, with $(\hat{S}_{a, \mytext{RCw}}, \hat{W}_{\mytext{RCw}})$ and
$(\hat\gamma_{a, \mytext{RCAL}}, \hat{\rho}_{a, 1:k, \mytext{RCAL}}, \hat{\eta}_{a, 1:k, \mytext{RCAL}}, \allowbreak \hat{\zeta}_{a, \mytext{RCAL}})$ replaced by their limit values.
Note that $B_{ai}(\theta, \hat{S}_{a, \mytext{RCw}},\hat{W}_{\mytext{RCw}})$ depends on $\hat\gamma_{a,\mytext{RCAL}}$ through $\hat{\pi}_{ai,\mytext{RCAL}}$
and on $(\hat{\rho}_{a, 1:k, \mytext{RCAL}}, \hat{\eta}_{a, 1:k, \mytext{RCAL}})$ through $\hat U^{(a)}_{ki, \mytext{RCAL}}$,
in addition to depending on $(\hat{S}_{a, \mytext{RCw}},\hat{W}_{\mytext{RCw}})$.

\begin{pro}\label{pro:AIPW}
Under suitable regularity and sparsity conditions in the high-dimensional setting as $n, \dim(f) \to\infty$ and $K$ is fixed,
$\hat\theta_{\mytext{RCa}}$ converges in probability to the same limit value $\breve\theta_{\mytext{RCw}}$ as $\hat\theta_{\mytext{RCw}}$ and
satisfies the asymptotic expansion
\begin{align}
 \hat\theta_{\mytext{RCa}} - \breve \theta_{\mytext{RCw}} =
  H_{\mytext{RCw}} (\theta)^{-1} \cdot \frac{1}{n} \sum_{i=1}^n \breve\delta_i(\theta) \Big|_{\breve\theta_{\mytext{RCw}}} + o_p(n^{-1/2}) , \label{eq:RCAL-theta-expan}
\end{align}
where $H_{\mytext{RCw}} (\theta)$ is the limit version of $\hat H_{\mytext{RCw}} (\theta)$ below  (see the proof for details).
Moreover, the asymptotic variance of $\hat\theta_{\mytext{RCa}}$ can be consistently estimated by
\begin{align}
\hat V_{\mytext{r}} ( \hat\theta_{\mytext{RCa}})=
\hat H_{\mytext{RCw}} (\theta)^{-1} \hat G_{\mytext{RCa}} (\theta )
\hat H_{\mytext{RCw}} (\theta)^{-1}  \Big|_{\hat\theta_{\mytext{RCa}}},  \label{eq:RCAL-theta-Vr}
\end{align}
where $\hat G_{\mytext{RCa}} (\theta ) = n^{-2} \sum_{i=1}^n \hat\delta_i(\theta)^2 $ and
\begin{align*}
\hat H_{\mytext{RCw}} (\theta) = \sum_{k=1}^K \frac{\hat W_{1k,\mytext{RCw}} \hat W_{0k,\mytext{RCw}} \me^\theta}
 { ( \hat W_{1k,\mytext{RCw}} \me^\theta + \hat W_{0k,\mytext{RCw}} )^2 }
 ( \hat W_{1k,\mytext{RCw}} \hat q_{1k,\mytext{RCw}}+ \hat W_{0k,\mytext{RCw}}\hat q_{0k,\mytext{RCw}} ).
\end{align*}
\end{pro}

\vspace{.05in}
Proposition \ref{pro:AIPW} is similar to
Proposition~\ref{pro:AIPW-CAL-lin} about $\hat S_{1k,\mytext{CAL,lin}}$ with linear OR model \eqref{eq:OR-model-lin}
and the discussion about $\hat\theta_{\mytext{CAL,lin}}$ with linear OR models \eqref{eq:OR-model-lin} and \eqref{eq:OR-model-W-lin}.
Nonparametrically, the asymptotic expansion \eqref{eq:RCAL-theta-expan} is valid, and
the variance estimator from \eqref{eq:RCAL-theta-Vr}
is consistent for the asymptotic variance of $\hat\theta_{\mytext{RCa}}$.
Moreover, $\hat\theta_{\mytext{RCa}}$ can be shown to be doubly robust for $\breve\theta_{\mytext{wBP}}$,
in being consistent (i.e., $\breve\theta_{\mytext{RCw}}=\breve\theta_{\mytext{wBP}}$)
if either PS and CCP models \eqref{eq:PS-model} and \eqref{eq:const-pi} or CSP and OR models \eqref{eq:const-m} and \eqref{eq:combine-OR-model}
are correctly specified, including both treatment-1 and treatment-0 versions.
Therefore, in either of the two cases,
an asymptotic $(1 - c)$-confidence interval for $\breve{\theta}_{\mytext{wBP}}$ is
$\hat{\theta}_{\mytext{RCa}} \pm z_{c/2}\sqrt{\hat V_{\mytext{r}} ( \hat\theta_{\mytext{RCa}})}$.

Compared with Propositions~\ref{pro:AIPW-CAL} and \ref{pro:CAL-theta},
PS model \eqref{eq:PS-model} is allowed to be misspecified in Proposition \ref{pro:AIPW}.
This difference is due to the choice that (nonlinear) logistic OR models \eqref{eq:OR-model} and \eqref{eq:OR-model-W} are used
in the augmentation terms for $\hat S_{1k,\mytext{CAL}}$ and $\hat\theta_{\mytext{CAL}}$,
whereas linear OR model \eqref{eq:combine-OR-model} is used for $\hat\theta_{\mytext{RCa}}$ here.
If OR model \eqref{eq:combine-OR-model} is replaced by a (nonlinear) generalized linear model
and regularized calibrated estimation of coefficient vector $\zeta_a$ is defined as in Tan (2020), then
similar properties as in Proposition \ref{pro:AIPW} can be obtained in general only when model \eqref{eq:PS-model} is correctly specified.
This would be more aligned with Propositions~\ref{pro:AIPW-CAL} and \ref{pro:CAL-theta}.

For completeness, we connect the regularized with non-regularized estimators. Consider the low-dimensional setting with Lasso tuning parameters set to 0,
so that $\hat\gamma_{a,\mytext{RCAL}} = \hat\gamma_{a,\mytext{CAL}}$ and $\hat{\pi}_{1i,\mytext{RCAL}}=\hat{\pi}_{1i,\mytext{CAL}}$.
Then the augmented estimator $\hat\theta_{\mytext{RCa}}$ reduces to the weighted-only estimator $\hat\theta_{\mytext{RCw}}$, because the augmentation terms
in \eqref{eq:AIPW} vanish directly by the definition of $(\hat\gamma_{1,\mytext{CAL}},\hat\gamma_{0,\mytext{CAL}})$, i.e.,
\begin{align*}
 0 = \frac{1}{n}\sum_{i=1}^n \left\{\left(\frac{A_i}{\hat{\pi}_{1i,\mytext{RCAL}}} - 1\right)\hat{g}_{1i, \mytext{RCAL}}
 - \left(\frac{1 - A_i}{1 - \hat{\pi}_{0i, \mytext{RCAL}}} - 1\right)\hat{g}_{0i, \mytext{RCAL}}\right\} .
\end{align*}
In estimating equation \eqref{eq:IPW-RCAL} for $\hat\theta_{\mytext{RCw}}$, note that
$\hat{S}_{ak, \mytext{RCw}} = \hat{S}_{ak, \mytext{wKM}} (\hat \gamma_{a,\mytext{CAL}}) = \hat{S}_{ak, \mytext{CAL,lin}} $ and
$\hat{W}_{ak, \mytext{RCw}}  = \hat{W}_{ak} (\hat \gamma_{a,\mytext{CAL}}) = \hat{W}_{ak, \mytext{CAL,lin}}$ from the discussion about $\hat\theta_{\mytext{CAL,lin}}$ in Section~\ref{sec:CAL-hazard}.
Then $\hat\theta_{\mytext{RCw}}$ coincides with $\hat\theta_{\mytext{CAL,lin}}$. From these observations, the point estimator
$\hat\theta_{\mytext{RCa}}$ reduces to $\hat\theta_{\mytext{CAL,lin}}$ in the low-dimensional setting.
Furthermore, it can be shown by direct calculation (details omitted) that
the variance estimator \eqref{eq:RCAL-theta-Vr} for $\hat\theta_{\mytext{RCa}}$
is also identical to the variance estimator \eqref{eq:CAL-theta-Vr} for $\hat\theta_{\mytext{CAL,lin}}$.

\section{Simulation study} \label{sec:simulation}
We conduct a simulation study to evaluate the performance of the following estimators for survival probabilities and
hazard-ratio parameter $\breve{\theta}_{\mytext{wBP}}$ in both low-dimensional and high-dimensional settings.
In the low-dimensional setting, where $p = 10$ and $n = 400$ or $1000$,  we compare $\hat{S}_{ak, \mytext{CAL}}$ and $\hat{S}_{ak, \mytext{CAL,lin}}$ against $\hat{S}_{ak, \mytext{KM}}$ and $\hat{S}_{ak, \mytext{wKM}}$ for survival probability estimation, and compare $\hat{\theta}_{\mytext{CAL}}$ and $\hat{\theta}_{\mytext{CAL,lin}}$
and, with Lasso penalty set to 0, $\hat{\theta}_{\mytext{RCw}}$ and $\hat{\theta}_{\mytext{RCa}}$ against $\hat{\theta}_{\mytext{BP}}$ and $\hat{\theta}_{\mytext{wBP}}$ for $\breve{\theta}_{\mytext{wBP}}$ estimation.
In the high-dimensional setting, where $p = 200$ and $n = 400$ or $1000$, we assess the performance of $\hat{S}_{ak, \mytext{RCAL}}$ and $\hat{S}_{ak, \mytext{RCAL,lin}}$ in comparison to $\hat{S}_{ak, \mytext{KM}}$ and $\hat{S}_{ak, \mytext{wKM}}$ for survival probability estimation, and compare $\hat{\theta}_{\mytext{RCw}}$ and $\hat{\theta}_{\mytext{RCa}}$ with $\hat{\theta}_{\mytext{BP}}$ and $\hat{\theta}_{\mytext{wBP}}$ for $\breve{\theta}_{\mytext{wBP}}$ estimation.
In this setting, $\hat{S}_{1k, \text{wKM}}$ and $\hat{\theta}_{\text{wBP}}$ are computed with $\hat{\gamma}$ obtained by
regularized maximum likelihood (RML).
Due to computational constraints, we do not compute $\hat{\theta}_{\mytext{RCAL}}$ or $\hat{\theta}_{\mytext{RCAL,lin}}$.
All estimates of coefficient vectors $\gamma, \alpha_{ak}, \beta_{ak}, \zeta_a$ are computed using the R package \texttt{RCAL} (Tan and Sun 2020).
The Lasso tuning parameters are selected via 5-fold cross-validation in the high-dimensional setting.

In the simulation, the treatment variable $A$ is generated from a Bernoulli distribution with probability 0.5.
The first 10 covariates of $X$ follow a truncated multivariate normal distribution in each treatment group: $X^{1:10} | A = 0 \sim N(\mu_0, \Sigma_0)$ and $X^{1:10} | A = 1 \sim N(\mu_1, \Sigma_1)$, where $\mu_0$ and $\mu_1$ are both $10\times 1$ vectors with all components being $-0.25/10$ and $0.25/10$ respectively. These covariates are truncated element-wise on the interval  $(-2.5 / \sqrt{10}, 2.5 / \sqrt{10})$.
The remaining covariates are independently generated from a standard normal distribution.
For the covariances matrices $\Sigma_0$ and $\Sigma_1$, we consider two cases: (C1) $\Sigma_0 = \Sigma_1 = (0.5^{|i - j|})_{i, j = 1, \ldots, 10}$; (C2) $\Sigma_0 =  4/3(0.5^{|i - j|})_{i, j = 1, \ldots, 10}$ and $\Sigma_1 = 2 / 3(0.5^{|i - j|})_{i, j = 1, \ldots, 10}$.
We show in Supplement Section \ref{sec:addtional-simulation} that PS model \eqref{eq:PS-model} with $f(X) = X$ is correctly specified under (C1) and misspecified under (C2).  Furthermore, we show that the limit value of $\hat{\gamma}_{1, \mytext{RCAL}}$ is nonzero only for the first 10 covariates.
Given $X$, a continuous event time $\tilde{U}$ is generated as Weibull with shape and scale parameters 2 and $\exp(0.5X_1 + \cdots + 0.5X_{10})$ or respectively 1 and $\exp(0.5X_1 + \cdots + 0.5X_{10})$ in treatment 0 and treatment 1.
A continuous censoring time $\tilde{C}$ is generated as 4 times Beta(2, 2) in treatment 0, and uniform(0, 4) in treatment 1.
The observed data $(Y, \Delta)$ are then obtained, where $\Delta = 1\{\tilde{U}\leq \tilde{C}\}$ and $Y$ is integer-valued, defined by discretizing $\tilde{Y} = \min(\tilde{U}, \tilde{C})$ in intervals of length .01 (Tan 2022).

Table \ref{tb:s} summarizes the estimates of $S_{1k}$ at $u_k = 60, 90, 120$, and Table \ref{tb:theta} summarizes the estimates of $\breve{\theta}_{\mytext{wBP}}$, based on 2000 repeated simulations with $n = 1000$ and $p = 10$ or $p = 200$ under (C1).
Results for the estimation of $S_{1k}$ and  $\breve{\theta}_{\mytext{wBP}}$ with $n = 400$ under (C1), as well as for $n = 400$ or $1000$ under (C2), and for the estimation of  $S_{0k}$, are presented in the Supplement, and similar conclusions can be drawn as discussed below.
A caveat under (C2) with misspecified model \eqref{eq:PS-model} is
that the bias of each estimator is defined against its limit value (or target value), not the true value (Supplement Tables \ref{tb:target-S} and \ref{tb:target-theta}).


From Table \ref{tb:s}, we have the following findings.
(1) $\hat{S}_{1k, \mytext{KM}}$ leads to considerable biases and under-coverage.
(2) Under low dimension $(p = 10)$,
$\hat{S}_{1k, \mytext{CAL}}$ and $\hat{S}_{1k, \mytext{CAL,lin}}$ methods perform similarly to each other
and slightly improve over $\hat{S}_{1k, \mytext{wKM}}$ method
in achieving shorter confidence intervals and coverage proportions closer to the nominal.
(3) Under high dimension $(p = 200)$,  all methods exhibit under-coverage, possibly due to insufficient sample size.
Nevertheless, $\hat{S}_{1k, \mytext{RCAL}}$ and $\hat{S}_{1k, \mytext{RCAL,lin}}$ noticeably improve over
$\hat{S}_{1k, \mytext{wKM}}$
in achieving smaller biases, shorter confidence intervals, and coverage proportions closer to the nominal. \vspace{-.1in}

\begin{table}[!htbp]
\caption{Summary for estimation of $S_{1k}$ at $u_k = 60, 90, 120$ with $n = 1000$ under (C1).} \vspace{-.15in} \label{tb:s}
\begin{center}
\small
\renewcommand{\arraystretch}{0.8}
\resizebox{0.9\textwidth}{!}{\begin{tabular}{lccccccccccccccccc}
\hline
& \multicolumn{5}{c}{$u_k = 60$} && \multicolumn{5}{c}{$u_k = 90$} && \multicolumn{5}{c}{$u_k = 120$}   \\
\cline{2-6}\cline{8-12}\cline{14-18}
& Bias & $\sqrt{\text{Var}}$ & $\sqrt{\text{EVar}}$ & Cov90(L90) & Cov95(L95) & & Bias & $\sqrt{\text{Var}}$ & $\sqrt{\text{EVar}}$ & Cov90(L90) & Cov95(L95) & & Bias & $\sqrt{\text{Var}}$ & $\sqrt{\text{EVar}}$ & Cov90(L90) & Cov95(L95) \\
\hline
& \multicolumn{17}{c}{$p = 10$} \\
$\hat{S}_{1k, \mytext{KM}}$ &   0.030  &  0.024  &  0.023  &  0.632 ( 0.076 )  &  0.732 ( 0.091 ) &  &   0.031  &  0.024  &  0.024  &  0.647 ( 0.078 )  &  0.753 ( 0.093 ) &  &   0.029  &  0.023  &  0.023  &  0.655 ( 0.077 )  &  0.773 ( 0.092 )\\
$\hat{S}_{1k, \mytext{wKM}}$ &   0.000  &  0.023  &  0.024  &  0.905 ( 0.078 )  &  0.959 ( 0.093 ) &  &   0.000  &  0.022  &  0.024  &  0.914 ( 0.077 )  &  0.956 ( 0.092 ) &  &  -0.001  &  0.022  &  0.023  &  0.915 ( 0.075 )  &  0.958 ( 0.089) \\
$\hat{S}_{1k, \mytext{CAL}}$ &   0.000  &  0.023  &  0.022  &  0.884 ( 0.074 )  &  0.944 ( 0.088 ) &  &   0.000  &  0.022  &  0.022  &  0.901 ( 0.073 )  &  0.946 ( 0.087 ) &  &  -0.001  &  0.021  &  0.022  &  0.903 ( 0.071 )  &  0.949 ( 0.085 )\\
$\hat{S}_{1k, \mytext{CAL,lin}}$ &   0.000  &  0.023  &  0.022  &  0.885 ( 0.074 )  &  0.947 ( 0.088 ) &  &   0.000  &  0.022  &  0.022  &  0.898 ( 0.073 )  &  0.944 ( 0.087 ) &  &  -0.001  &  0.022  &  0.022  &  0.904 ( 0.071 )  &  0.949 ( 0.085 )\\

& \multicolumn{17}{c}{$p = 200$} \\
$\hat{S}_{1k, \mytext{KM}}$ &   0.030  &  0.023  &  0.023  &  0.637 ( 0.076 )  &  0.739 ( 0.091 ) &  &  0.030  &  0.024  &  0.024  &  0.644 ( 0.078 )  &  0.744 ( 0.093 ) &  &  0.029  &  0.023  &  0.023  &  0.662 ( 0.077 )  &  0.774 ( 0.092 )\\
$\hat{S}_{1k, \mytext{wKM}}$ &   0.024  &  0.023  &  0.023  &  0.716 ( 0.077 )  &  0.824 ( 0.091 ) &  &  0.024  &  0.023  &  0.024  &  0.727 ( 0.078 )  &  0.824 ( 0.093 ) &  &  0.023  &  0.023  &  0.023  &  0.756 ( 0.077 )  &  0.844 ( 0.091 )\\
$\hat{S}_{1k, \mytext{RCAL}}$ &   0.011  &  0.023  &  0.022  &  0.846 ( 0.073 )  &  0.911 ( 0.086 ) &  &  0.009  &  0.023  &  0.022  &  0.860 ( 0.074 )  &  0.927 ( 0.088 ) &  &  0.008  &  0.022  &  0.022  &  0.871 ( 0.073 )  &  0.938 ( 0.086 )\\
$\hat{S}_{1k, \mytext{RCAL,lin}}$ &   0.011  &  0.023  &  0.022  &  0.847 ( 0.073 )  &  0.912 ( 0.086 ) &  &  0.009  &  0.023  &  0.022  &  0.863 ( 0.074 )  &  0.926 ( 0.088 ) &  &  0.008  &  0.022  &  0.022  &  0.871 ( 0.073 )  &  0.938 ( 0.087 )\\

\hline
\end{tabular}}
\end{center}
\setlength{\baselineskip}{0.2\baselineskip}
\vspace{-.15in}\noindent{\tiny
\textbf{Note}: Bias is the Monte Carlo bias of the point estimates against the true (data-generating) values. Var are the Monte Carlo variance of the point estimates. EVar is the mean of the variance estimates. Cov90(L90) and Cov95(L95) denote coverage proportion and average length of the 90\% and 95\% confidence intervals.}
\vspace{-.3in}
\end{table}

\begin{table}[!htbp]
  \caption{Summary for estimation of $\breve{\theta}_{\mytext{wBP}}$ with $n = 1000$ under (C1).}
  \vspace{-.15in}
  \label{tb:theta}
  \begin{center}
  \small
  \renewcommand{\arraystretch}{0.8}
  \resizebox{0.9\textwidth}{!}{
  \begin{tabular}{lccccccccccc}
  \hline
  & \multicolumn{6}{c}{$p = 10$} &&  \multicolumn{4}{c}{$p = 200$} \\
  \cline{2-7}\cline{9-12}
  & $\hat{\theta}_{\mytext{BP}}$
  & $\hat{\theta}_{\mytext{wBP}}$
  & $\hat{\theta}_{\mytext{CAL}}$
  & $\hat{\theta}_{\mytext{CAL,lin}}$
  & $\hat{\theta}_{\mytext{RCw}}$
  & $\hat{\theta}_{\mytext{RCa}}$ & 
  & $\hat{\theta}_{\mytext{BP}}$
  & $\hat{\theta}_{\mytext{wBP}}$
  & $\hat{\theta}_{\mytext{RCw}}$
  & $\hat{\theta}_{\mytext{RCa}}$ \\
  \hline
  Bias
  & -0.191 & -0.001 &  0.000 & -0.001 & -0.001 & -0.001 & 
  & -0.189 & -0.152 & -0.158 & -0.042 \\
  $\sqrt{\text{Var}}$
  & 0.074 & 0.060 & 0.060 & 0.060 & 0.060 & 0.060 & 
  & 0.073 & 0.069 & 0.067 & 0.061 \\
  $\sqrt{\text{EVar}}$
  & 0.074 & 0.075 & 0.059 & 0.060 & 0.075 & 0.060 & 
  & 0.074 & 0.074 & 0.074 & 0.059 \\
  Cov90(L90)
  & 0.172 (0.242) & 0.961 (0.246) & 0.895 (0.195) & 0.898 (0.197) & 0.957 (0.246) & 0.898 (0.197) & 
  & 0.166 (0.242) & 0.324 (0.243) & 0.281 (0.243) & 0.814 (0.195) \\
  Cov95(L95)
  & 0.258 (0.288) & 0.985 (0.294) & 0.945 (0.232) & 0.949 (0.235) & 0.985 (0.294) & 0.949 (0.235) & 
  & 0.263 (0.288) & 0.458 (0.289) & 0.418 (0.289) & 0.886 (0.233) \\
  \hline
  \end{tabular}}
  \end{center}
\setlength{\baselineskip}{0.2\baselineskip}
\vspace{-.15in}\noindent{\tiny
\textbf{Note}: Bias is the Monte Carlo bias relative to the true value of $\breve{\theta}_{\text{wBP}}$ in \eqref{eq:weighted-BP-pop}. Var is the Monte Carlo variance, EVar is the mean estimated variance. Cov90(L90) and Cov95(L95) denote coverage proportion and average length of the 90\% and 95\% confidence intervals.}
  \vspace{-.15in}
  \end{table}

From Table \ref{tb:theta}, we have the following findings.
(1) $\hat{\theta}_{\mytext{BP}}$ leads to considerable biases and under-coverage.
(2) Under low dimension,
$\hat{\theta}_{\mytext{CAL}}$, $\hat{\theta}_{\mytext{CAL,lin}}$ and $\hat{\theta}_{\mytext{RCa}}$ methods
perform noticeably better than $\hat{\theta}_{\mytext{wBP}}$ and $\hat{\theta}_{\mytext{RCw}}$ 
in producing shorter confidence intervals, variance estimates closer to the Monte Carlo variances, and coverage proportions closer to the nominal.
Note that $\hat{\theta}_{\mytext{CAL,lin}}$ and $\hat{\theta}_{\mytext{RCa}}$ give identical results (including point and variance estimates)
as explained in Section \ref{sec:high-dimension}.
(3) Under high dimension, all methods exhibit under-coverage, possibly due to insufficient sample size.
The method $\hat{\theta}_{\mytext{RCa}}$ performs the best, with smallest bias and, on average,
coverage proportions closest to the target ones and shortest confidence intervals. \vspace{-.1in}

\section{Empirical application}\label{sec:empirical}
We illustrate the proposed methods by studying the effectiveness of adjunctive psychotropic treatments for patients with schizophrenia, using a dataset derived from the U.S. national Medicaid Analytic eXtract files covering January 1, 2001 to December 31, 2010. This dataset was previously analyzed by Stroup et al. (2019). The original cohort includes four treatment groups corresponding to initiating treatment with an antidepressant, a benzodiazepine, a mood stabilizer, or another antipsychotic.
To focus on binary treatment comparisons, we define initiating treatment with an antidepressant as treatment 0 and initiating treatment with a benzodiazepine as treatment 1.
We consider three main outcomes analyzed in Stroup et al. (2019), namely
hospitalization for a mental disorder (primary), emergency department (ED) visits for a mental disorder, and all-cause mortality. The total sample size is 43,058, and the covariate vector X includes 101 baseline variables after converting categorical variables to dummy variables.

We consider two settings for PS and OR model specifications: a low-dimensional setting and a high-dimensional setting. In the low-dimensional setting, working models include main effects from $X$ only. In the high-dimensional setting,  working models
include  additionally 2275 two-way interactions from $X$ excluding those either with nonzero values less than 344 (i.e., 0.8\% of the sample size 43,058)
or perfectly collinear with other terms.

In the low-dimensional setting, we apply $\hat{S}_{ak, \mytext{KM}}$, $\hat{S}_{ak, \mytext{wKM}}$ and $\hat{S}_{ak, \mytext{CAL,lin}}$ to estimate survival probabilities. For estimating $\breve{\theta}_{\mytext{wBP}}$, we apply $\hat{\theta}_{\mytext{BP}}$, $\hat{\theta}_{\mytext{wBP}}$ and $\hat{\theta}_{\mytext{CAL,lin}}$.
We do not apply calibrated estimators based on logistic OR models \eqref{eq:OR-model} and \eqref{eq:OR-model-W} (i.e., $\hat{S}_{ak, \mytext{CAL}}$ and  $\hat{\theta}_{\mytext{CAL}}$) because the corresponding performance tends to be similar to that of the estimators based on linear OR models. In addition, we do not apply $\hat{\theta}_{\mytext{RCa}}$ because it coincides with $\hat{\theta}_{\mytext{CAL,lin}}$ in this setting, and we exclude $\hat{\theta}_{\mytext{RCw}}$ due to its inferior performance observed in the simulation study.

In the high-dimensional setting, we apply  $\hat{S}_{ak, \mytext{wKM}}$ and $\hat{S}_{ak, \mytext{RCAL,lin}}$ to estimate survival probabilities. For estimating  $\breve{\theta}_{\mytext{wBP}}$, we apply $\hat{\theta}_{\mytext{wBP}}$ and $\hat{\theta}_{\mytext{RCa}}$.
As in the simulation study, $\hat{S}_{1k, \text{wKM}}$ and $\hat{\theta}_{\text{wBP}}$ are computed with $\hat{\gamma}$ obtained by RML.
For similar reasons as in the low-dimensional setting, we do not apply $\hat{S}_{ak, \mytext{RCAL}}$ or $\hat{\theta}_{\mytext{RCw}}$. Due to computational constraints as in the simulation study, we exclude $\hat{\theta}_{\mytext{RCAL}}$ and $\hat{\theta}_{\mytext{RCAL,lin}}$ in the high-dimensional setting.

Survival probabilities are evaluated at $\{30\times k: k =1, 2, \ldots, 12\}$ for Hospitalization and ED visits outcomes. For mortality outcome, exact multiples of 30 do not correspond to observed times. Therefore, we estimate survival probabilities at the nearest available times: 30, 64, 91, 121, 151, 180, 210, 242, 274, 300, 330, and 361 days. 
We summarize the results from two aspects: comparison between low dimension and high dimension, and comparison between existing weighted methods and calibrated methods. A discussion on the effects of covariate adjustment is provided in Supplement Section~\ref{sec:addtional-empirical}, and the key finding from there is as follows.
For the Hospitalization and Mortality outcomes, weighting primarily affects survival probability estimation associated with treatment 1, leading to substantial differences between $\hat{\theta}_{\mytext{wBP}}$ and $\hat{\theta}_{\mytext{BP}}$, whereas for the ED visits outcome, weighting affects survival probability estimation for both treatments, and the effects partially cancel out for estimation of the hazard-ratio parameter. 

\begin{table}[!htbp]
\caption{Summary for estimation of  $\breve{\theta}_{\mytext{BP}}$} \label{tb:empirical-theta} \vspace{-.15in}
\begin{center}
\small
\renewcommand{\arraystretch}{0.7}
\resizebox{0.8\textwidth}{!}{\begin{tabular}{lccccccccccccc}
\hline
  \multirow{3}{*}{Outcome} &  \multicolumn{6}{c}{Point Est} & &  \multicolumn{6}{c}{SE} \\
   \cline{2-7}\cline{9-14}
  & \multicolumn{3}{c}{Low-Dim} & & \multicolumn{2}{c}{High-Dim} & & \multicolumn{3}{c}{Low} & &  \multicolumn{2}{c}{High} \\
  \cline{2-4}\cline{6-7}\cline{9-11}\cline{13-14}
  & $\hat{\theta}_{\mytext{BP}}$ & $\hat{\theta}_{\mytext{wBP}}$ & $\hat{\theta}_{\mytext{CAL,lin}}$ & & $\hat{\theta}_{\mytext{wBP}}$ & $\hat{\theta}_{\mytext{RCa}}$ &
  & $\hat{\theta}_{\mytext{BP}}$ & $\hat{\theta}_{\mytext{wBP}}$ & $\hat{\theta}_{\mytext{CAL,lin}}$ & & $\hat{\theta}_{\mytext{wBP}}$ & $\hat{\theta}_{\mytext{RCa}}$  \\
  \hline
Hospitalization & 0.194 & 0.262 & 0.238& & 0.254 & 0.218 & & 0.025  & 0.030 & 0.027 & & 0.029 & 0.026 \\
ED Visits & 0.192 & 0.219 & 0.192 & & 0.218 & 0.190 & & 0.022 & 0.027 & 0.025  & & 0.026 & 0.023 \\
Mortality & 0.495 & 0.195 & 0.208 & & 0.253 & 0.322 & & 0.089 & 0.109 & 0.101  & & 0.102  & 0.095 \\
\hline
\end{tabular}}
\end{center}
\setlength{\baselineskip}{0.2\baselineskip}
\vspace{-.15in}\noindent{\tiny
\textbf{Note}: Point Est is point estimate, and SE is standard error. Low-Dim means that only main-effects are included in working models, and High-Dim means that both main-effects and interactions are included in working models.
}
\vspace{-.15in}
\end{table}

\begin{figure}[!htbp]
\centering
\includegraphics[scale=0.52]{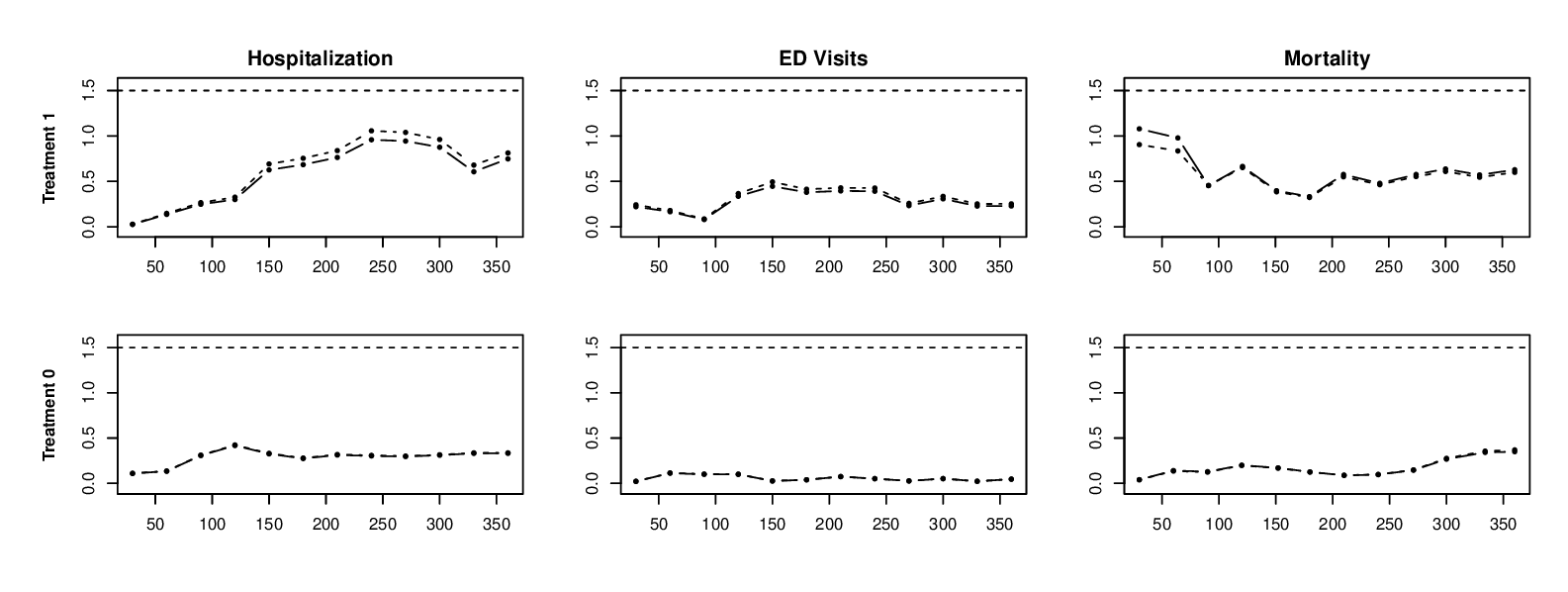} \vspace{-.48in}
\caption{Standardized absolute difference between $\hat{S}_{ak, \mytext{CAL,lin}}$ and $\hat{S}_{ak, \mytext{RCAL,lin}}$. Solid line is standardized by $\text{SE}(\hat{S}_{ak, \mytext{CAL,lin}})$ and dashed line is standardized by $\text{SE}(\hat{S}_{ak, \mytext{RCAL,lin}})$. A horizontal line is placed at 1.5.} \label{fig:calsvprob} \vspace{-.1in}
\end{figure}

\textbf{Comparison between low dimension and high dimension.}\; Figure \ref{fig:calsvprob} presents the standardized absolute difference between $\hat{S}_{ak, \mytext{CAL,lin}}$ and  $\hat{S}_{ak, \mytext{RCAL,lin}}$. Across outcomes and treatment groups, the two estimators yield similar point estimates
when evaluated against SEs.
A consistent pattern is observed for the estimation of $\breve{\theta}_{\mytext{wBP}}$ from Table \ref{tb:empirical-theta}, where $\hat{\theta}_{\mytext{CAL,lin}}$ and  $\hat{\theta}_{\mytext{RCa}}$ also produce comparable point estimates relative to SEs.
These findings suggest that a PS model including only main effects may be ``nearly correct'', since otherwise the point estimates would likely differ. This conclusion is further supported by the covariate balance results in Supplement Table \ref{tb:covbal}, which shows that extending the PS model to include interaction terms does not materially improve the balance of the interactions, indicating that the interaction terms are already approximately balanced once main effects are controlled for.

In spite of similar point estimates, the high-dimensional methods exhibit notable efficiency gains over low-dimensional ones.
Except at time points 30 and 64 for treatment 1 in the Mortality outcome,
$\hat{S}_{ak, \mytext{RCAL,lin}}$ generally has SEs no larger than those of $\hat{S}_{ak, \mytext{CAL,lin}}$.
\footnote{\label{ft:empirical-result}
The exceptions can be explained by the fact that
the SEs are extremely small at time points 30 and 64. Specifically, as reported in Supplement Table \ref{tb:mort-tr1},
for treatment 1 in the Mortality outcome, the SEs of $\hat{S}_{ak, \mytext{RAIPW}}$ and $\hat{S}_{ak, \mytext{CAL,lin}}$ are $4\times 10^{-4}$ and $3\times 10^{-4}$ respectively at time 30, and $5\times 10^{-4}$ and $4\times 10^{-4}$ at time 64. These slight absolute differences
lead to the reversal in the standardized comparison.}
This can be seen from Figure \ref{fig:calsvprob}:
the dashed lines, where the absolute difference is scaled by $\text{SE}(\hat{S}_{ak, \mytext{RCAL,lin}})$, are consistently no lower than the solid lines, where the absolute difference is scaled by $\text{SE}(\hat{S}_{ak, \mytext{CAL,lin}})$.
From Table \ref{tb:empirical-theta}, $\hat{\theta}_{\mytext{RCa}}$ also consistently has smaller SEs than $\hat{\theta}_{\mytext{CAL,lin}}$.
For example, for the ED visits outcome, the relative efficiency of $\hat{\theta}_{\mytext{RCa}}$ over $\hat{\theta}_{\mytext{CAL,lin}}$
is $(.025/.023)^2 = 1.18$,
The efficiency gain is due to regularization estimation using more flexible working models, as observed in an empirical application
with only administrative censoring in Tan (2020).

Finally, similar conclusions are obtained when comparing between low- and high-dimensional applications of $\hat{S}_{ak, \mytext{wKM}}$
for survival probability estimation from Supplement Figure \ref{fig:mlesvprob}, and when comparing those of
$\hat{\theta}_{\mytext{wBP}}$ for the estimation of $\breve{\theta}_{\mytext{wBP}}$ from Table \ref{tb:empirical-theta}.

\textbf{Comparison between existing weighted methods and calibrated methods.}\;  Figure \ref{fig:mlecalsvprob-low} presents the standardized absolute difference between $\hat{S}_{ak, \mytext{wKM}}$ and $\hat{S}_{ak, \mytext{CAL,lin}}$ in the low dimensional setting. For treatment 0 across all three outcomes, the two estimators yield very similar point estimates relative to SEs. Although the differences for treatment 1 are relatively larger, they remain within two SEs.
A consistent pattern is observed for the estimation of  $\breve{\theta}_{\mytext{wBP}}$ from Table \ref{tb:empirical-theta}, where $\hat{\theta}_{\mytext{wBP}}$ and $\hat{\theta}_{\mytext{CAL,lin}}$ produce similar point estimates.

Beyond similar point estimates, $\hat{S}_{ak, \mytext{CAL,lin}}$ generally has SEs no larger than those of $\hat{S}_{ak, \mytext{wKM}}$.
This can be seen from Figure \ref{fig:mlecalsvprob-low}: the dashed lines, where
the absolute differences are scaled by $\text{SE}(\hat{S}_{ak, \mytext{CAL,lin}})$, are consistently no lower than the solid lines, where the differences are scaled by $\text{SE}(\hat{S}_{ak, \mytext{wKM}})$.
Analogously, for the estimation of $\breve{\theta}_{\mytext{wBP}}$ in Table \ref{tb:empirical-theta},
$\hat{\theta}_{\mytext{CAL,lin}}$ has smaller SEs than $\hat{\theta}_{\mytext{wBP}}$.
These efficiency gains can be attributed to two factors.  First, inverse probability weights constructed via calibrated estimation tend to be less variable than those obtained from ML estimation (Tan 2020). Second,
augmented IPW estimation via OR model \eqref{eq:OR-model-lin} in $\hat{S}_{ak, \mytext{CAL,lin}}$ often achieves smaller variance than
that via the implied, constant OR model in $\hat{S}_{ak, \mytext{wKM}}$ (footnote \ref{ft:advtange-cal-over-wkm}).

Finally, similar conclusions hold in the high dimensional setting. As shown in Supplement Figure \ref{fig:mlecalsvprob-high}, $\hat{S}_{ak, \mytext{wKM}}$ and $\hat{S}_{ak, \mytext{RCAL,lin}}$ yield comparable survival probability estimates relative to SEs. From Table \ref{tb:empirical-theta},
$\hat{\theta}_{\mytext{wBP}}$ and $\hat{\theta}_{\mytext{RCa}}$ produce similar point estimates for $\breve{\theta}_{\mytext{wBP}}$ relative to SEs,
but the calibrated estimator $\hat{\theta}_{\mytext{RCa}}$ achieve smaller SEs than $\hat{\theta}_{\mytext{wBP}}$, for example, by a factor of $(.026 / .023)^2 = 1.28$ for the ED visits outcome.

\begin{figure}[!htbp]
\centering
\includegraphics[scale=0.52]{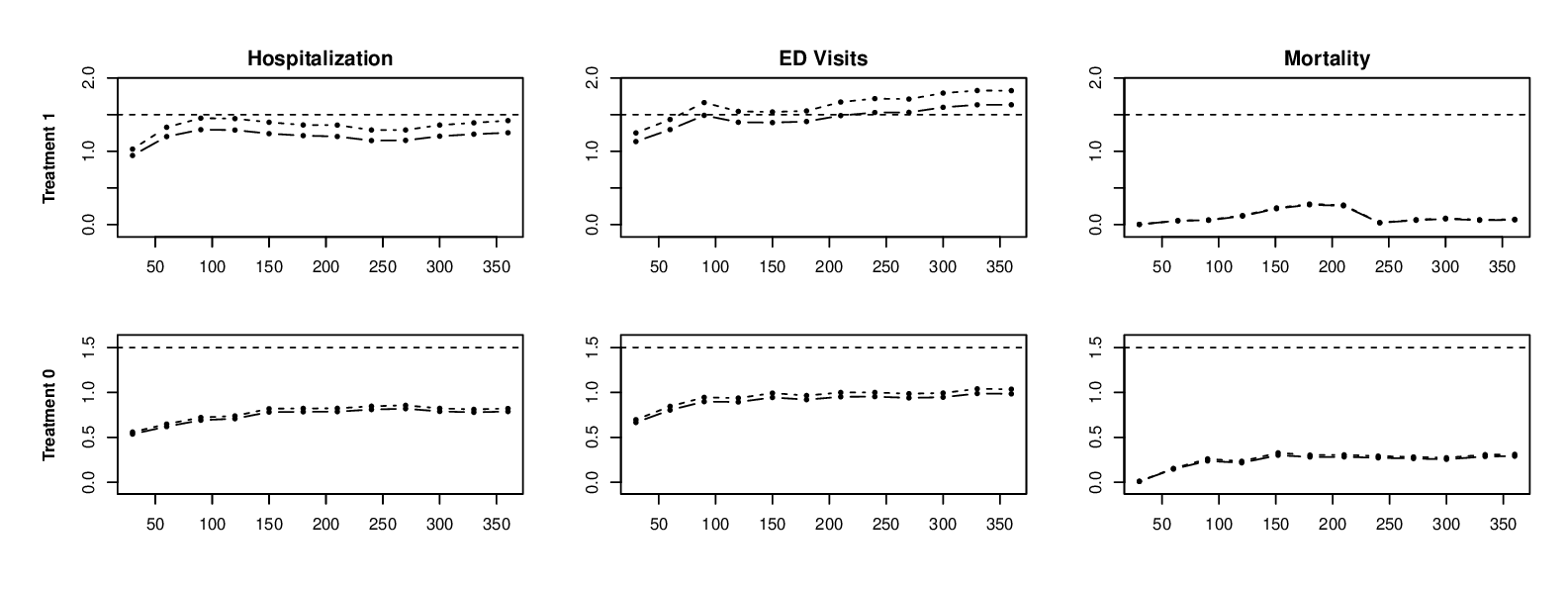} \vspace{-.48in}
\caption{Standardized absolute difference between $\hat{S}_{ak, \mytext{wKM}} $ and $\hat{S}_{ak, \mytext{CAL,lin}}$. Solid line is standardized by $\text{SE}(\hat{S}_{ak, \mytext{wKM}})$ and dashed line is standardized by $\text{SE}(\hat{S}_{ak, \mytext{CAL,lin})}$. A horizontal line is placed at 1.5.} \label{fig:mlecalsvprob-low} \vspace{-.1in}
\end{figure}

\textbf{Summary.}\;  There are three main findings from the empirical results. 
First, as discussed in Supplement Section~\ref{sec:addtional-empirical} and summarized earlier, covariate adjustment has a meaningful impact on the estimation of survival probabilities and hazard-ratio parameter $\breve{\theta}_{\mytext{wBP}}$, but this impact is outcome and treatment specific.
Second, expanding from low-dimensional to high-dimensional working models does not materially change point estimates, suggesting that a main-effect PS model is adequate in this application, a conclusion further supported by covariate balance diagnostics. At the same time, regularized estimation consistently achieves efficiency gains, yielding smaller SEs. Third, existing methods and calibrated methods produce similar point estimates across outcomes and low- and high-dimensional settings, but calibrated methods uniformly exhibit improved efficiency, both for survival probability estimation and $\breve{\theta}_{\mytext{wBP}}$ estimation.

\vspace{.3in}
\centerline{\bf\Large References}
\begin{description}\addtolength{\itemsep}{-.15in}

\item Bai, X., Tsiatis, A., and O'Brien, S. (2013) Doubly-robust Estimators of Treatment-specific Survival Distributions in Observational Studies with Stratified Sampling, {\it Biometrics}, 69, 830–839.

\item Breslow, N. (1974) Covariance analysis of censored survival data, {\it Biometrics}, 30, 89–100.

\item Chan, K., Yam, S., and Zhang, Z. (2016). Globally efficient non-parametric inference of average treatment effects by empirical balancing calibration weighting, {\it Journal of the Royal Statistical Society Series B}, 78, 673–700.

\item Cole, S. and Hern\'{a}n, M. (2004) Adjusted survival curves with inverse probability weights, {\it Computer Methods and Programs in Biomedicine}, 75, 45–49.

\item Cox, D. (1972) Regression models and life tables (with discussion), {\it Journal of the Royal Statistical Society Series B}, 34, 187–220.

\item Ghosh, S. and Tan, Z. (2022) Doubly robust semiparametric inference using regularized calibrated estimation with high-dimensional data, {\it Bernoulli}, 28, 1675–1703.

\item Hern\'{a}n, M. and Robins, J. (2020) {\it Causal Inference: What If}, Boca Raton: Chapman \& Hall/CRC.

\item Hubbard, A., van der L., and Robins, J. (2000) Nonparametric locally efficient estimation of the treatment specific survival distributions with right censored data and covariates in observational studies,  In {\it Statistical Models in Epidemiology, the Environment, and Clinical Trials} (pp. 135–177), New York: Springer.

\item Joffe, R., MacQueen, G., Marriott, M., and Young, L. (2004) A prospective, longitudinal study of percentage of time spent ill in patients with bipolar I or bipolar II disorders, {\it International Journal of Bipolar Disorders}, 6, 62–66.

\item Kaplan, E. and Meier, P. (1958) Nonparametric estimation from incomplete observations, {\it Journal of the American Statistical Association}, 53, 457–481.

\item Peto, R. (1972) Contribution to the discussion of Cox (1972): regression models and life tables, {\it Journal of the Royal Statistical Society Series B},  34, 205–207.

\item Robins, J. and Rotnitzky, A. (1992) Recovery information and adjustment for dependent censoring using surrogate markers, In {\it AIDS epidemiology: methodological issues} (, 297–331), Boston: Springer.

\item Rosenbaum, P. and Rubin, D. (1983) The Central Role of the Propensity Score in Observational Studies for Causal Effects, {\it Biometrika}, 70, 41–55.

\item Stroup, T., Gerhard, T., Crystal, S., Huang, C., and Tan, Z. (2019) Comparative Effectiveness of Adjunctive Psychotropic Medications in Patients With Schizophrenia, {\it JAMA psychiatry}, 76, 508–515.

\item Tan, Z. (2006) A Distributional Approach for Causal Inference Using Propensity Scores, {\it Journal of the American Statistical Association}, 101, 1607–1618.

\item Tan, Z. (2020) Model-assisted inference for treatment effects using regularized calibrated estimation with high-dimensional data, {\it Annals of Statistics}, 48, 811–837.

\item Tan, Z. and Sun, B. (2020) {\it RCAL: Regularized calibrated estimation}, R package version 2.0.

\item Tan, Z. (2022) Analysis of odds, probability, and hazard ratios: From 2 by 2 tables to two-sample survival data, {\it Journal of Statistical Planning and Inference}, 221, 248–265.

\item Tan, Z. (2023) Consistent and robust inference in hazard probability and odds models with discrete-time survival data, {\it Lifetime Data Analysis}, 29, 555–584.

\item Therneau, T. (2020) {\it A package for survival analysis in R}, R package version 3.2-7.

\item Therneau, T. and Grambsch, P. (2000) {\it Modeling survival data: extending the Cox model}, New York: Springer.

\item Tsiatis, A. (2006) {\it Semiparametric Theory and Missing Data}, New York: Springer.

\item Vermeulen, K. and Vansteelandt, S. (2015) Bias-reduced doubly robust estimation, {\it Journal of the American Statistical Association}, 110, 1024–1036.

\item Xie, J. and Liu, C. (2005) Adjusted Kaplan–Meier estimator and log-rank test with inverse probability of treatment weighting for survival data, {\it Statistics in Medicine}, 24, 3089–3110.

\item Yuan, M. and Lin, Y. (2006) Model selection and estimation in regression with grouped variables, {\it Journal of the Royal Statistical Society Series B}, 68, 49–67.
\end{description}

\clearpage

\setcounter{page}{1}

\setcounter{section}{0}
\setcounter{equation}{0}

\setcounter{table}{0}
\setcounter{figure}{0}

\setcounter{pro}{0}
\renewcommand{\thepro}{S\arabic{pro}}

\setcounter{lem}{0}
\renewcommand{\thelem}{S\arabic{lem}}

\setcounter{thm}{0}
\renewcommand{\thethm}{S\arabic{thm}}

\setcounter{ass}{0}
\renewcommand{\theass}{S\arabic{ass}}

\renewcommand{\thesection}{\Roman{section}}
\renewcommand{\theequation}{S\arabic{equation}}

\renewcommand\thetable{S\arabic{table}}
\renewcommand\thefigure{S\arabic{figure}}

\begin{center}
{\Large Supplementary Material for}\\
{\Large ``Re-examining and calibrating weighted survival analysis for causal inference''}

\vspace{.1in} 
Wenfu Xu, Yi Zhang, Tobias Gerhard, Zhiqiang Tan

\end{center}

\section{Augmented IPW estimator for $k = 2$} \label{sec:aipw-k2}
We provide the explicit expressions for the augmented IPW estimator in the special case of $k=2$, as mentioned in Section~\ref{sec:AIPW} of the main paper.
The augmented IPW estimator \eqref{eq:AIPW} with the expression \eqref{eq:phi-ps} becomes
\begin{align}
\hat S_{12}  (\hat\gamma, \hat\rho_{1,1:2}, \hat\eta_{1,1:2})
=  \frac{1}{n} \sum_{i=1}^n & \Bigg\{ \frac{A_i }{\hat\pi_i } \frac{R_{1i} R_{2i} }{\hat\pi_{11i} \hat\pi_{12i} } 1\{ U^{(1)}_i > u_2\}
 - \frac{A_i }{\hat\pi_i } \frac{R_{1 i} }{\hat\pi_{11 i} }  \left( \frac{R_{2 i} }{\hat\pi_{12 i} }-1 \right)\hat m_{12i}  1\{ U^{(1)}_i > u_1 \} \nonumber \\
 & \quad - \frac{A_i }{\hat\pi_i } \left( \frac{R_{1 i} }{\hat\pi_{11 i} } -1 \right)\hat m_{12i}\hat m_{11i}
 - \left( \frac{A_i }{\hat\pi_i } -1 \right)\hat m_{12i}\hat m_{11i} \Bigg\} . \label{eq:AIPW-2period-a}
\end{align}
The augmented IPW estimator with the expression \eqref{eq:phi-or} becomes
\begin{align}
\hat S_{12}  (\hat\gamma, \hat\rho_{1,1:2}, \hat\eta_{1,1:2})
=  \frac{1}{n} \sum_{i=1}^n & \Bigg\{ \frac{A_i }{\hat\pi_i } \frac{R_{1i} R_{2i} }{\hat\pi_{11i} \hat\pi_{12i} } 1\{ U^{(1)}_i > u_1 \}  (1\{ U^{(1)}_i > u_2\} - \hat m_{12i}) \nonumber \\
 & \quad + \frac{A_i }{\hat\pi_i } \frac{R_{1 i} }{\hat\pi_{11 i} }\hat m_{12i} ( 1\{ U^{(1)}_i > u_1 \} -\hat m_{11i}) + \hat m_{12i}\hat m_{11i} \Bigg\}. \label{eq:AIPW-2period-b}
\end{align}
The two expressions can be easily verified to be equivalent to each other.

\section{Comparison of weighted and unweighted KM estimation}  \label{sec:comparison-KM}

We compare weighted and unweighted KM estimation in their consistency properties.
Formally, the unweighted KM estimator of $S_{1k}$ (defined by setting $\hat w_{1i}\equiv 1$) is doubly robust, but with respect to the following sets of models:
either $\pi^*(X)$ is constant in $X$ (hence $\hat\pi_i$ is constant in $i$) and model \eqref{eq:const-pi} holds, or model \eqref{eq:const-m} holds,
in addition to the unadjusted version of Assumptions~\eqref{eq:no-confounding}--\eqref{eq:non-informative}:
\begin{align}
&  A \perp U^{(1)}  \label{eq:no-confounding-b} \\
& U^{(1)} \perp C^{(1)} | A=1 , \label{eq:non-informative-b}
\end{align}
We make two simple but important observations (see Supplement Section \ref{sec:technical-detail} for proofs). \vspace{-.05in}
\begin{itemize}\addtolength{\itemsep}{-.1in}
\item[(i)] If CCP model \eqref{eq:const-pi} is correctly specified (i.e., $\pi^*_{1j}(X) = \rho^*_{1j}$ for some constant $\rho^*_{1j}$,
$j=1,\ldots,k$), then Assumption \eqref{eq:non-informative} implies the unadjusted version \eqref{eq:non-informative-b},
provided that the behavior of $C^{(1)}$ is restricted up to time $u_k$
(i.e., $C^{(1)}$ is replaced by $1\{C^{(1)} \ge u_j\}$ for any $j=1,\ldots,k$).

\item[(ii)] If CSP model \eqref{eq:const-m} is correctly specified (i.e., $m^*_{1j}(X) = \eta^*_{1j}$ for some constant $\eta^*_{1j}$,
$j=1,\ldots,k$), then Assumption \eqref{eq:non-informative} implies the unadjusted version \eqref{eq:non-informative-b},
and Assumption~\eqref{eq:no-confounding} implies the unadjusted version \eqref{eq:no-confounding-b},
provided that the behavior of $U^{(1)}$ is restricted up to time $u_{k+1}$
(i.e., $U^{(1)}$ is replaced by $1\{U^{(1)} > u_j\}$ for any $j=1,\ldots,k$).
\end{itemize} \vspace{-.05in}
From these observations, we obtain the following comparison provided that Assumptions~\eqref{eq:no-confounding}--\eqref{eq:non-informative} hold.
First, if CSP model \eqref{eq:const-m} is correctly specified, then both the weighted and unweighted KM estimators
are consistent for $S_{1k}$.
However, if PS model $\pi(X;\gamma)$ and CCP model \eqref{eq:const-pi} are correctly specified but CSP model \eqref{eq:const-m} may be misspecified,
then the weighted KM estimator is consistent while the unweighted KM estimator is in general inconsistent.
Hence compared with the unweighted KM estimator, the weighted KM estimator is advantageous
when propensity score model $\pi(X;\gamma)$ and CCP model \eqref{eq:const-pi} are correctly specified
but CSP model \eqref{eq:const-m} may be misspecified.

\section{Calibrated estimation with augmented IPW estimation}  \label{sec:original-CAL}

We discuss application of calibrated estimation to the original augmented IPW estimator
$\hat S_{1k} (\hat\gamma, \hat\rho_{1,1:k}, \allowbreak \hat\eta_{1,1:k})$, and compare with calibrated estimation to the new augmented IPW estimator
$\hat S_{1k} (\hat\gamma, \hat\rho_{1,1:k}, \allowbreak \hat\eta_{1,1:k}, \hat\alpha_{1k})$.

When applied to the original augmented IPW estimator $\hat S_{1k} (\hat\gamma, \hat\rho_{1,1:k}, \hat\eta_{1,1:k})$, the calibration equations (Tan 2020; Ghosh \& Tan 2022) are
\begin{align}
& E \left\{\frac{\partial }{\partial\gamma} \hat S_{1k} (\gamma, \rho_{1,1:k}, \eta_{1,1:k }) \right\}
 = 0, \label{eq:original-CAL-a} \\
& E \left\{\frac{\partial }{\partial \rho_{1,1:k}} \hat S_{1k} (\gamma, \rho_{1,1:k}, \eta_{1,1:k }) \right\}
 = 0, \label{eq:original-CAL-b} \\
& E \left\{\frac{\partial }{\partial \eta_{1,1:k}} \hat S_{1k} (\gamma, \rho_{1,1:k}, \eta_{1,1:k }) \right\}
 = 0. \label{eq:original-CAL-c}
\end{align}
For simplicity, consider the case of $k=2$ and $\hat S_{12} (\hat\gamma, \hat\rho_{1,1:2}, \hat\eta_{1,1:2})$ in \eqref{eq:AIPW-2period-a} and \eqref{eq:AIPW-2period-b}.
For logistic PS model \eqref{eq:PS-model} and CCP and CSP models \eqref{eq:const-pi} and \eqref{eq:const-m}, equation \eqref{eq:original-CAL-a} can be calculated as
\begin{align*}
0 = - E \Bigg[ \frac{1}{n} \sum_{i=1}^n & \Bigg\{ A_i \frac{1-\pi(X_i;\gamma) }{\pi(X_i;\gamma) } \frac{R_{1i} R_{2i} }{\rho_{11}\rho_{12} } 1\{ U^{(1)}_i > u_1 \}  (1\{ U^{(1)}_i > u_2\} - \eta_{12}) f(X_i) \\
 & \quad + A_i \frac{1-\pi(X_i;\gamma) }{\pi(X_i;\gamma) }
  \frac{R_{1 i} }{\rho_{11} }\eta_{12} ( 1\{ U^{(1)}_i > u_1 \} -\eta_{11}) f(X_i) \Bigg\} \Bigg] ,
\end{align*}
equation \eqref{eq:original-CAL-b} with $\rho_{1,1:2}=(\rho_{11},\rho_{12})$ can be calculated as
\begin{align}   \label{eq:original-CAL-b-2period}
\begin{split}
0 = - E \Bigg[ \frac{1}{n} \sum_{i=1}^n & \Bigg\{  \frac{A_i}{\pi(X_i;\gamma) } \frac{R_{1i} R_{2i} }{\rho_{11}^2 \rho_{12} } 1\{ U^{(1)}_i > u_1 \}  (1\{ U^{(1)}_i > u_2\} - \eta_{12}) \\
 & \quad + \frac{A_i }{\pi(X_i;\gamma) }
  \frac{R_{1 i}}{\rho_{11}^2} \eta_{12} ( 1\{ U^{(1)}_i > u_1 \} -\eta_{11}) \Bigg\} \Bigg] ,\\
0 = - E \Bigg[ \frac{1}{n} \sum_{i=1}^n & \Bigg\{  \frac{A_i}{\pi(X_i;\gamma) } \frac{R_{1i} R_{2i} }{\rho_{11}\rho_{12}^2 } 1\{ U^{(1)}_i > u_1 \}  (1\{ U^{(1)}_i > u_2\} - \eta_{12})  \Bigg\} \Bigg],
\end{split}
\end{align}
and equation \eqref{eq:original-CAL-c} with $\eta_{1,1:2}=(\eta_{11},\eta_{12})$ can be calculated as
\begin{align*}
 0= -E \Bigg[ \frac{1}{n} \sum_{i=1}^n & \Bigg\{ \frac{A_i }{\pi(X_i;\gamma) } \left( \frac{R_{1 i} }{\rho_{11} } -1 \right) \eta_{12}
 + \left( \frac{A_i }{\pi(X_i;\gamma) } -1 \right) \eta_{12} \Bigg\} \Bigg], \\
 0= -E \Bigg[ \frac{1}{n} \sum_{i=1}^n & \Bigg\{ \frac{A_i }{\pi(X_i;\gamma) } \frac{R_{1 i}}{\rho_{11}}  \left( \frac{R_{2 i} }{\rho_{12} }-1 \right) 1\{ U^{(1)}_i > u_1 \} \\
 & \quad + \frac{A_i }{\pi(X_i;\gamma) } \left( \frac{R_{1 i} }{\rho_{11} } -1 \right) \eta_{11}
 + \left( \frac{A_i }{\pi(X_i;\gamma) } -1 \right) \eta_{11} \Bigg\} \Bigg] .
\end{align*}
There are $\dim(f)+4$ scalar unknowns, $(\gamma, \rho_{11},\rho_{12}, \eta_{11},\eta_{12})$, and
also $\dim(f)+4$ scalar equations. However, it seems difficult to manipulate this system of equations and obtain tractable solutions.

When applied to the new augmented IPW estimator $\hat S_{1k} (\hat\gamma, \hat\rho_{1,1:k}, \hat\eta_{1,1:k}, \hat\alpha_{1k})$, the calibration equations (Tan 2020; Ghosh \& Tan 2022) are
\begin{align}
& E \left\{\frac{\partial }{\partial\gamma} \hat S_{1k} (\gamma, \rho_{1,1:k}, \eta_{1,1:k }, \alpha_{1k}) \right\}
 = 0, \label{eq:new-CAL-a} \\
& E \left\{\frac{\partial }{\partial \rho_{1,1:k}} \hat S_{1k} (\gamma, \rho_{1,1:k}, \eta_{1,1:k }, \alpha_{1k}) \right\}
 = 0, \label{eq:new-CAL-b} \\
& E \left\{\frac{\partial }{\partial \eta_{1,1:k}} \hat S_{1k} (\gamma, \rho_{1,1:k}, \eta_{1,1:k }, \alpha_{1k}) \right\}
 = 0, \label{eq:new-CAL-c} \\
& E \left\{\frac{\partial }{\partial \alpha_{1k}} \hat S_{1k} (\gamma, \rho_{1,1:k}, \eta_{1,1:k }, \alpha_{1k}) \right\}
 = 0. \label{eq:new-CAL-d}
\end{align}
In the case of $k=2$, the new augmented IPW estimator is
\begin{align*}
& \quad \hat S_{12}  (\hat\gamma, \hat\rho_{1,1:2}, \hat\eta_{1,1:2}, \hat\alpha_{12}) \\
& =  \frac{1}{n} \sum_{i=1}^n \Bigg\{ \frac{A_i }{\hat\pi_i } \frac{R_{1i} R_{2i} }{\hat\pi_{11i} \hat\pi_{12i} } 1\{ U^{(1)}_i > u_2\}
 - \frac{A_i }{\hat\pi_i } \frac{R_{1 i} }{\hat\pi_{11 i} }  \left( \frac{R_{2 i} }{\hat\pi_{12 i} }-1 \right)\hat m_{12i}  1\{ U^{(1)}_i > u_1 \} \\
 & \quad - \frac{A_i }{\hat\pi_i } \left( \frac{R_{1 i} }{\hat\pi_{11 i} } -1 \right)\hat m_{12i}\hat m_{11i}
 - \left( \frac{A_i }{\hat\pi_i } -1 \right) \hat\mu_{12i} \Bigg\} .
\end{align*}
For logistic PS model \eqref{eq:PS-model} and CCP and CSP models \eqref{eq:const-pi} and \eqref{eq:const-m}, equation \eqref{eq:new-CAL-a} can be calculated as
\begin{align}\label{eq:new-CAL-a-2period}
\begin{split}
0 = - E \Bigg[ \frac{1}{n} \sum_{i=1}^n & \Bigg\{ A_i \frac{1-\pi(X_i;\gamma) }{\pi(X_i;\gamma) } \frac{R_{1i} R_{2i} }{\rho_{11}\rho_{12} } 1\{ U^{(1)}_i > u_1 \}  (1\{ U^{(1)}_i > u_2\} - \eta_{12}) f(X_i) \\
 & \quad + A_i \frac{1-\pi(X_i;\gamma) }{\pi(X_i;\gamma) }
  \frac{R_{1 i} }{\rho_{11} }\eta_{12} ( 1\{ U^{(1)}_i > u_1 \} -\eta_{11}) f(X_i) \\
 & \quad + A_i \frac{1-\pi(X_i;\gamma) }{\pi(X_i;\gamma) }
  (\eta_{12}\eta_{11} -\mu_{12}(X_i;\alpha_{12})) f(X_i) \Bigg\} \Bigg] ,
\end{split}
\end{align}
equation \eqref{eq:new-CAL-b} with $\rho_{1,1:2}=(\rho_{11},\rho_{12})$ can be calculated as \eqref{eq:original-CAL-b-2period},
equation \eqref{eq:new-CAL-c} with $\eta_{1,1:2}=(\eta_{11},\eta_{12})$ can be calculated as
\begin{align}\label{eq:new-CAL-c-2period}
\begin{split}
 0= -E \Bigg[ \frac{1}{n} \sum_{i=1}^n & \Bigg\{ \frac{A_i }{\pi(X_i;\gamma) } \left( \frac{R_{1 i} }{\rho_{11} } -1 \right) \eta_{12}
 \Bigg\} \Bigg], \\
 0= -E \Bigg[ \frac{1}{n} \sum_{i=1}^n & \Bigg\{ \frac{A_i }{\pi(X_i;\gamma) } \frac{R_{1 i}}{\rho_{11}}  \left( \frac{R_{2 i} }{\rho_{12} }-1 \right) 1\{ U^{(1)}_i > u_1 \} \\
 & \quad + \frac{A_i }{\pi(X_i;\gamma) } \left( \frac{R_{1 i} }{\rho_{11} } -1 \right) \eta_{11}
 \Bigg\} \Bigg] .
\end{split}
\end{align}
For logistic OR model \eqref{eq:OR-model}, equation \eqref{eq:new-CAL-d} can be calculated as
\begin{align}\label{eq:new-CAL-d-2period}
0 = -E \Bigg[ \frac{1}{n} \sum_{i=1}^n \Bigg\{ \left( \frac{A_i }{\pi(X_i;\gamma) } -1 \right) \frac{\partial\mu_{12}(X_i;\alpha_{12})}{\partial\alpha_{12}} \Bigg\} \Bigg] .
\end{align}
For linear OR model \eqref{eq:OR-model-lin}, equation \eqref{eq:new-CAL-d} can be calculated as
\begin{align}\label{eq:new-CAL-d-2period-linear}
0 = -E \Bigg[ \frac{1}{n} \sum_{i=1}^n \Bigg\{ \left( \frac{A_i }{\pi(X_i;\gamma) } -1 \right) f(X_i) \Bigg\} \Bigg] .
\end{align}
Equations \eqref{eq:new-CAL-a-2period}, \eqref{eq:original-CAL-b-2period}, \eqref{eq:new-CAL-c-2period}, \eqref{eq:new-CAL-d-2period-linear}
can be solved in a tractable manner via \eqref{eq:CAL-rho-pop-b}--\eqref{eq:CAL-alpha-pop} as in the proof of Proposition~\ref{pro:AIPW-CAL}.
For logistic OR model \eqref{eq:OR-model}, equation \eqref{eq:new-CAL-d-2period} is not directly solved, but
is still satisfied by the limit values of the calibrated estimators when PS model is correctly specified as assumed in Proposition~\ref{pro:AIPW-CAL}.

\section{Technical details} \label{sec:technical-detail}

\textbf{Derivation of variance estimator \eqref{eq:weighted-KM-Vr}.}\;
Suppose that $\hat\gamma$ were data-independent. By definition \eqref{eq:weighted-KM} for $\hat S_{1k,\mytext{wKM}}$,
\begin{align*}
\log \hat S_{1k,\mytext{wKM}} = \sum_{j=1}^k  \log ( 1- \hat q_{1j} ).
\end{align*}
By the delta method, it suffices to show that a model-robust variance estimator for $\log \hat S_{1k,\mytext{wKM}}$ is
$\frac{1}{n^2} \sum_{i=1}^n \hat\varphi_{1ki,\mytext{wKM}}^2$. The population (or limit) version of $\hat q_{1j}$ is
\begin{align*}
 \breve q_{1j} = \frac{E( \sum_{i \in J_j } A_i \hat w_{1i})  }{ E( \sum_{i\in I_j } A_i \hat w_{1i} ) } ,
\end{align*}
and hence the population (or limit) version of $\log \hat S_{1k,\mytext{wKM}}$ is
\begin{align*}
\log \breve S_{1k,\mytext{wKM}} = \sum_{j=1}^k  \log ( 1- \breve q_{1j} ) .
\end{align*}
Consider the Taylor expansion
\begin{align*}
& \quad \log \hat S_{1k,\mytext{wKM}}  - \log \breve S_{1k,\mytext{wKM}}
=  \sum_{j=1}^k  \left\{ \log ( 1- \hat q_{1j} ) - \log ( 1- \breve q_{1j} ) \right\}\\
& =  \sum_{j=1}^k  \frac{-1}{ 1- \breve q_{1j} } ( \hat q_{1j} - \breve q_{1j} ) + o_p(n^{-1/2}) ,
\end{align*}
where by standard manipulation,
\begin{align*}
 & \quad \hat q_{1j} - \breve q_{1j} =  \frac{\sum_i 1\{i\in I_j \}( 1\{i \in J_j \} - \breve q_{1j} ) A_i \hat w_{1i} }{ \sum_i 1\{i\in I_j \} A_i \hat w_{1i} } \\
 & = \frac{\sum_i ( 1\{i \in J_j \} - \breve q_{1j} ) A_i \hat w_{1i} }{ E ( \sum_i 1\{i\in I_j \} A_i \hat w_{1i} ) } + o_p(n^{-1/2}) .
\end{align*}
Combining the preceding two displays yields
\begin{align*}
& \quad \log \hat S_{1k,\mytext{wKM}}  - \log \breve S_{1k,\mytext{wKM}}
=  \sum_{j=1}^k  \left\{ \log ( 1- \hat q_{1j} ) - \log ( 1- \breve q_{1j} ) \right\}\\
& = \sum_{j=1}^k  \frac{-1}{ E ( \sum_{i\in I_j\setminus J_j} A_i \hat w_{1i} )  }
\sum_{i=1}^n 1\{i\in I_j \}( 1\{i \in J_j \} - \breve q_{1j} ) A_i \hat w_{1i}
+ o_p(n^{-1/2}) \\
& = \frac{1}{n} \sum_{i=1}^n \breve\varphi_{1ki,\mytext{wKM}} + o_p(n^{-1/2}) ,
\end{align*}
where
\begin{align*}
  \breve\varphi_{1ki,\mytext{wKM}} = n \sum_{j=1}^k  \frac{-A_i \hat w_{1i} }
  { E ( \sum_{i\in I_j\setminus J_j} A_i \hat w_{1i} )  }  1\{i\in I_j \}( 1\{i \in J_j \} - \breve q_{1j} ).
\end{align*}
With the approximation of $\breve\varphi_{1ki,\mytext{wKM}}$ by $\hat\varphi_{1ki,\mytext{wKM}}$, we see that
a consistent variance estimator for $ \log \hat S_{1k,\mytext{wKM}}$ is $\frac{1}{n^2} \sum_{i=1}^n \hat\varphi_{1ki,\mytext{wKM}}^2$.

The function \texttt{survfit()} in R package \texttt{survival}
also enables model-robust variance estimation for survival probabilities after fitting proportional hazards models,
but the formula appears to in general differ from that derived from a Taylor expansion as discussed in Tan (2023), Supplement Section I.1.

\vspace{.1in}
\textbf{Proof of Proposition~\ref{pro:weighted-KM}.}\;
First, we show that the augmented IPW estimator $\hat S_{1k} (\hat\gamma, \hat\rho_{1,1:k,\mytext{CAL}}, \allowbreak \hat\eta_{1,1:k,\mytext{CAL}})$ coincides with the ICE estimator $\hat S_{1k} (\hat\eta_{1,1:k,\mytext{CAL}})$. By definitions (\ref{eq:const-pi}) and (\ref{eq:const-m}), the fitted values $\hat\pi_{1ji,\mytext{CAL}} = \hat\rho_{1j,\mytext{CAL}}$
and $\hat m_{1ji,\mytext{CAL}} = \hat\eta_{1j,\mytext{CAL}}$ are constant in $i$.
Then the identity $\hat S_{1k} (\hat\gamma, \hat\rho_{1,1:k,\mytext{CAL}}, \hat\eta_{1,1:k,\mytext{CAL}})
= \hat S_{1k} (\hat\eta_{1,1:k,\mytext{CAL}}) $ follows by applying the equations \eqref{eq:CAL-eta}
to the expressions \eqref{eq:AIPW} and \eqref{eq:phi-or}.

Second, we show that the ICE estimator $\hat S_{1k} (\hat\eta_{1,1:k,\mytext{CAL}})$ coincides with the weighted Kaplan–Meier estimator $\hat S_{1k,\mytext{wKM}} $. Solving equation \eqref{eq:CAL-eta} for individual $0 \le \ell\le k-1$ yields
\begin{align*}
  \hat \eta_{1,\ell+1,\mytext{CAL}} = \frac{\sum_{i=1}^n  \frac{A_i }{\hat\pi_i } \overline R_{\ell+1, i}
 1\{ U^{(1)}_i > u_{\ell+1} \} }
{\sum_{i=1}^n  \frac{A_i }{\hat\pi_i } \overline R_{\ell+1, i} 1\{ U^{(1)}_i > u_\ell \} } ,
\end{align*}
which can be expressed in the notation of Section~\ref{sec:existing} as
\begin{align}
 \hat \eta_{1,\ell+1,\mytext{CAL}} = \frac{\sum_{i\in I_{\ell+1} \setminus J_{\ell+1} } A_i \hat w_{1i} }
{ \sum_{i\in I_{\ell+1} } A_i \hat w_{1i} } = 1- \hat q_{1,\ell+1}.  \label{eq:prf-eta-expression}
\end{align}
Then by definition \eqref{eq:ICE},
\begin{align*}
  \hat S_{1k} (\hat\eta_{1,1:k,\mytext{CAL}})
  & = \prod_{\ell=0}^{k-1}  \hat \eta_{1,\ell+1,\mytext{CAL}} =\prod_{\ell=0}^{k-1}   \frac{\sum_{i\in I_{\ell+1} \setminus J_{\ell+1} } A_i \hat w_{1i} }
{ \sum_{i\in I_{\ell+1} } A_i \hat w_{1i} }  = \hat S_{1k,\mytext{wKM}} .
\end{align*}

Third, we show that the augmented IPW estimator $\hat S_{1k} (\hat\gamma, \hat\rho_{1,1:k,\mytext{CAL}}, \hat\eta_{1,1:k,\mytext{CAL}})$ coincides with the IPW estimator $\hat S_{1k} (\hat\gamma, \hat\rho_{1,1:k,\mytext{CAL}} )$ under the normalization condition. 
The identity $\hat S_{1k} (\hat\gamma, \hat\rho_{1,1:k,\mytext{CAL}}, \allowbreak \hat\eta_{1,1:k,\mytext{CAL}})
= \hat S_{1k} (\hat\gamma, \hat\rho_{1,1:k,\mytext{CAL}} ) $ follow by applying the equations \eqref{eq:CAL-rho}
to the expressions \eqref{eq:AIPW} and \eqref{eq:phi-ps} and using $\sum_{i=1}^n (\frac{A_i}{\hat\pi_i}-1)=0$
by the assumption $\sum_{i=1}^n A_i \hat\pi_i^{-1}=n$. This completes the proof of Proposition~\ref{pro:weighted-KM}.

For completeness, we also verify directly that the weighted Kaplan–Meier estimator $\hat S_{1k,\mytext{wKM}}$ coincides with the IPW estimator $\hat S_{1k} (\hat\gamma, \hat\rho_{1,1:k,\mytext{CAL}} )$ under the normalization condition. 
Solving equation \eqref{eq:CAL-rho} for individual $0 \le \ell\le k-1$ yields
\begin{align*}
 \hat \rho_{1,\ell+1,\mytext{CAL}} = \frac{\sum_{i=1}^n  \frac{A_i }{\hat\pi_i } \overline R_{\ell i}
 R_{\ell+1, i} 1\{ U^{(1)}_i > u_\ell \} }
{\sum_{i=1}^n  \frac{A_i }{\hat\pi_i } \overline R_{\ell i} 1\{ U^{(1)}_i > u_\ell \} } ,
\end{align*}
which can be expressed in the notation of Section~\ref{sec:existing} ($\hat w_{1i} = \hat \pi_i^{-1}$) as
\begin{align}
 \hat \rho_{1,\ell+1,\mytext{CAL}} = \frac{\sum_{i\in I_{\ell+1}} A_i \hat w_{1i} }
{ \sum_{i\in I_\ell \setminus J_\ell} A_i \hat w_{1i} } ,\label{eq:prf-rho-expression}
\end{align}
due to the following relationship
\begin{align} \label{eq:prf-IJ-relation}
\begin{split}
 \{ i \in I_\ell \setminus J_\ell, A_i=1\}
& =\{ Y_i \ge u_\ell, (Y_i,\Delta_i)\not=(u_\ell,1), A_i=1 \} \\
& = \{ Y_i \ge u_\ell, U^{(1)}_i > u_\ell , A_i=1 \} = \{ R_{\ell i}=1, U^{(1)}_i > u_\ell , A_i=1 \} .
\end{split}
\end{align}
Then by definition \eqref{eq:IPW},
\begin{align*}
  \hat S_{1k} (\hat\gamma, \hat\rho_{1,1:k,\mytext{CAL}} )
  &= \frac{\frac{1}{n} \sum_{i=1}^n  \frac{A_i }{\hat\pi_i } \overline R_{ki} 1\{ U^{(1)}_i > u_k\} }
  { \prod_{\ell=0}^{k-1} \hat\rho_{1,\ell+1,\mytext{CAL}} } \\
  & = \frac{\frac{1}{n} \sum_{i \in I_k \setminus J_k} A_i \hat w_{1i}}
  { \prod_{\ell=0}^{k-1}  \frac{\sum_{i\in I_{\ell+1}} A_i \hat w_{1i} }
{ \sum_{i\in I_\ell \setminus J_\ell} A_i \hat w_{1i} } }  \\
& = \frac{\frac{1}{n} } {  \sum_{i\in I_0 \setminus J_0} A_i \hat w_{1i} }
\prod_{\ell=0}^{k-1} \frac{\sum_{i \in I_{\ell+1} \setminus J_{\ell+1}} A_i \hat w_{1i}} {\sum_{i\in I_{\ell+1}} A_i \hat w_{1i} } = \hat S_{1k,\mytext{wKM}} ,
\end{align*}
Note that $I_0 = \{1\le i\le n\}$ and $J_0=\emptyset$ because $P(U^{(a)} >u_0)=1$ for $a=0,1$ in our setup (Section~\ref{sec:setup}).
The last step above follows because
$\sum_{i \in I_0 \setminus J_0} A_i \hat w_{1i} = \sum_{i=1}^n A_i \hat w_{1i} = n$ by assumption.

\vspace{.1in}
\textbf{Proof of properties (i) and (ii) in Section~\ref{sec:comparison-KM}.}\;
(i) Suppose that model \eqref{eq:const-pi} is correctly specified. Then for any $j=1,\ldots,k$,
\begin{align*}
 P( C^{(1)} \ge u_j | U^{(1)}, X, A=1)
 & = P( C^{(1)} \ge u_j | X, A=1) \\
 & = P( C^{(1)} \ge u_j | A=1),
\end{align*}
where the first equality follows from Assumption \eqref{eq:non-informative}, and the second follows from model \eqref{eq:const-pi}.
The above equation holds for any $X$, and hence implies that  for any $j=1,\ldots,k$,
\begin{align*}
 &  P( C^{(1)} \ge u_j | U^{(1)},  A=1)
 =P( C^{(1)} \ge u_j | A=1).
\end{align*}
This is Assumption \eqref{eq:non-informative-b} provided that the behavior of $C^{(1)}$ is restricted up to time $u_k$.

(ii) Suppose that model \eqref{eq:const-m} is correctly specified. Then for any $j=1,\ldots,k$,
\begin{align*}
 P( U^{(1)} > u_j | C^{(1)}, X, A=1)
 & = P( U^{(1)} > u_j | X, A=1) \\
 & = P( U^{(1)} > u_j | A=1),
\end{align*}
where the first equality follows from Assumption \eqref{eq:non-informative}, and the second follows from model \eqref{eq:const-m}.
The above equation holds for any $X$, and hence implies that  for any $j=1,\ldots,k$,
\begin{align*}
 &  P( U^{(1)} > u_j | C^{(1)},  A=1)
 =P( U^{(1)} > u_j | A=1).
\end{align*}
This is Assumption \eqref{eq:non-informative-b} provided that the behavior of $U^{(1)}$ is restricted up to time $u_{k+1}$.
Moreover, for any $j=1,\ldots,k$,
\begin{align*}
 P( U^{(1)} > u_j | X)
 &=  P( U^{(1)} > u_j | X, A=1) \\
 & = P( U^{(1)} > u_j | A=1),
\end{align*}
where the first equality follows from Assumption \eqref{eq:no-confounding}, and the second follows from model \eqref{eq:const-m} as indicated earlier.
The above equation holds for any $X$, and hence implies that  for any $j=1,\ldots,k$,
\begin{align*}
 &  P( U^{(1)} > u_j )
 =P( U^{(1)} > u_j | A=1).
\end{align*}
This is Assumption \eqref{eq:no-confounding-b} provided that the behavior of $U^{(1)}$ is restricted up to time $u_{k+1}$.

\vspace{.1in}
\textbf{Proof of Proposition~\ref{pro:weighted-KM-Vr}.}\;
The consistency of $\hat V_{\mytext{r}} ( \hat S_{1k,\mytext{wKM}} )$ follows from the derivation as a model-robust variance estimator.
To show the consistency of $\hat V_{\mytext{r}}( \hat S_{1k} (\hat\gamma, \hat\rho_{1,1:k,\mytext{CAL}}, \hat\eta_{1,1:k,\mytext{CAL}}) )$,
let $( \breve\rho_{1,1:k,\mytext{CAL}}, \breve\eta_{1,1:k,\mytext{CAL}})$ be the limit values of
$( \hat\rho_{1,1:k,\mytext{CAL}}, \hat\eta_{1,1:k,\mytext{CAL}})$, which are determined as the solutions to the population version
of the equations \eqref{eq:CAL-rho} and \eqref{eq:CAL-eta}:
\begin{align}
& 0= E \left[\frac{1}{n} \sum_{i=1}^n  \frac{A_i }{\hat\pi_i } \overline R_{\ell i}
\left( \frac{R_{\ell+1, i} }{\rho_{1,\ell+1} }-1 \right) 1\{ U^{(1)}_i > u_\ell \} \right], \quad \ell=0,\ldots,k-1,  \label{eq:CAL-rho-pop} \\
& 0 = E \left[ \frac{1}{n} \sum_{i=1}^n  \frac{A_i }{\hat\pi_i } \overline R_{\ell+1, i}
1\{ U^{(1)}_i > u_\ell \} ( 1\{U^{(1)}_i > u_{\ell+1}-  \eta_{1,\ell+1} ) \right], \quad \ell=0,\ldots,k-1,  \label{eq:CAL-eta-pop}
\end{align}
By a Taylor expansion, the estimator $ \hat S_{1k} (\hat\gamma, \hat\rho_{1,1:k,\mytext{CAL}}, \hat\eta_{1,1:k,\mytext{CAL}}) $
can be shown to satisfy the following asymptotic expansion:
\begin{align} \label{eq:AIPW-expan-fix-gamma}
\begin{split}
& \quad \hat S_{1k} (\hat\gamma, \hat\rho_{1,1:k,\mytext{CAL}}, \hat\eta_{1,1:k,\mytext{CAL}})
= \hat S_{1k} (\hat\gamma, \breve\rho_{1,1:k,\mytext{CAL}}, \breve\eta_{1,1:k,\mytext{CAL}}) \\
& \quad + E \left\{\frac{\partial }{\partial \rho_{1,1:k}} \hat S_{1k} (\hat\gamma, \rho_{1,1:k}, \eta_{1,1:k }) \right\}^\T
 \Big|_{\breve\rho_{1,1:k,\mytext{CAL}}, \breve\eta_{1,1:k,\mytext{CAL}}}  ( \hat\rho_{1,1:k,\mytext{CAL}} - \breve\rho_{1,1:k,\mytext{CAL}}) \\
& \quad + E \left\{\frac{\partial }{\partial \eta_{1,1:k}} \hat S_{1k} (\hat\gamma, \rho_{1,1:k}, \eta_{1,1:k }) \right\}^\T
 \Big|_{\breve\rho_{1,1:k,\mytext{CAL}}, \breve\eta_{1,1:k,\mytext{CAL}}}   ( \hat\eta_{1,1:k,\mytext{CAL}} - \breve\eta_{1,1:k,\mytext{CAL}}) + o_p(n^{-1/2}),
\end{split}
\end{align}
where $\hat S_{1k} (\hat\gamma, \rho_{1,1:k}, \eta_{1,1:k })$ is defined as
$S_{1k} (\hat\gamma, \hat\rho_{1,1:k}, \hat\eta_{1,1:k})$
with $(\hat\rho_{1,1:k}, \hat\eta_{1,1:k})$ replaced by $(\rho_{1,1:k},$ $ \eta_{1,1:k })$.
Because $\pi_{1j}(X_i; \rho_{1j} ) = \rho_{1j}$ and $m_{1j}(X_i; \eta_{1j} ) = \eta_{1j}$ are constant in $i$,
it is easily verified that
\begin{align*}
& 0 = E \left\{\frac{\partial }{\partial \rho_{1,1:k}} \hat S_{1k} (\hat\gamma, \rho_{1,1:k}, \eta_{1,1:k }) \right\}
 \Big|_{\breve\rho_{1,1:k,\mytext{CAL}}, \breve\eta_{1,1:k,\mytext{CAL}}} ,\\
& 0 = E \left\{\frac{\partial }{\partial \eta_{1,1:k}} \hat S_{1k} (\hat\gamma, \rho_{1,1:k}, \eta_{1,1:k }) \right\}
 \Big|_{\breve\rho_{1,1:k,\mytext{CAL}}, \breve\eta_{1,1:k,\mytext{CAL}}} ,
\end{align*}
using equations \eqref{eq:CAL-rho-pop}, \eqref{eq:CAL-eta-pop} and $E( \sum_{i=1}^n A_i \hat\pi_i^{-1}) =n$.
The derivatives involved are, for $j = 1, \ldots, k$,
\begin{align*}
& \quad \frac{\partial }{\partial \rho_{1j}} \hat S_{1k} (\hat\gamma, \rho_{1,1:k}, \eta_{1,1:k }) \\
& =  -\sum_{\ell = j -1}^{k - 1}\left\{\frac{\prod_{j^\prime = \ell + 2}^k \eta_{1j^\prime}}{\rho_{1j}\prod_{j^\prime = 1}^{\ell + 1} \rho_{1j^\prime} }\times \frac{1}{n}\sum_{i=1}^n \frac{A_i }{\hat\pi_i } \overline R_{\ell+1, i}1\{ U^{(1)}_i > u_\ell \} ( 1\{U^{(1)}_i > u_{\ell+1}-  \eta_{1,\ell+1} ) \right\}, \\
& \quad \frac{\partial }{\partial \eta_{1j}} \hat S_{1k} (\hat\gamma, \rho_{1,1:k}, \eta_{1,1:k }) \\
& = -\sum_{\ell = 0}^{j - 1}\left\{\frac{\prod_{j^\prime = \ell + 1, j^\prime\neq j}^k \eta_{1j^\prime}}{\prod_{j^\prime = 1}^\ell \rho_{1j^\prime} }\times \frac{1}{n}\sum_{i=1}^n  \frac{A_i }{\hat\pi_i } \overline R_{\ell i}\left( \frac{R_{\ell+1, i} }{\rho_{1,\ell+1} }-1 \right) 1\{ U^{(1)}_i > u_\ell \}\right\} \\
& \quad - \prod_{j^\prime = 1, j^\prime\neq j}^k \eta_{1j^\prime}\times \frac{1}{n}\sum_{i=1}^n \left(\frac{A_i}{\hat\pi_i} - 1\right),
\end{align*}
where a product over an empty index set (e.g., $\prod_{j^\prime = k+1}^k \eta_{1j^\prime}$) is taken to be 1. These expressions are obtained by straightforward differentiation of
$\hat{S}_{1k}(\hat\gamma, \hat{\rho}_{1, 1:k}, \hat{\eta}_{1, 1:k})$ in Equation \eqref{eq:AIPW}.
See Section~\ref{sec:original-CAL} for related expressions in the case of $k=2$.
Therefore,
the simple asymptotic expansion holds:
\begin{align*}
& \quad \hat S_{1k} (\hat\gamma, \hat\rho_{1,1:k,\mytext{CAL}}, \hat\eta_{1,1:k,\mytext{CAL}})
= \hat S_{1k} (\hat\gamma, \breve\rho_{1,1:k,\mytext{CAL}}, \breve\eta_{1,1:k,\mytext{CAL}}) + o_p(n^{-1/2}),
\end{align*}
which implies the consistency of the variance estimator $\hat V_{\mytext{r}}( \hat S_{1k} (\hat\gamma, \hat\rho_{1,1:k,\mytext{CAL}}, \hat\eta_{1,1:k,\mytext{CAL}}) )$
under standard regularity conditions.

Next, we show that
if $\sum_{i=1}^n A_i \hat\pi_i^{-1}=n$, then
 $\hat V_{\mytext{r}}( \hat S_{1k} (\hat\gamma, \hat\rho_{1,1:k,\mytext{CAL}}, \hat\eta_{1,1:k,\mytext{CAL}}) )$ algebraically coincides with
$\hat V_{\mytext{r}} ( \hat S_{1k,\mytext{wKM}} )$.
Note that
$\hat S_{1k} (\hat\gamma, \hat\rho_{1,1:k,\mytext{CAL}}, \hat\eta_{1,1:k,\mytext{CAL}}) = \hat{\overline \eta}_{1k,\mytext{CAL}}$.
By comparing the expressions in \eqref{eq:weighted-KM-Vr} and \eqref{eq:AIPW-var}, it suffices to show
\begin{align*}
\varphi_{1ki} (\hat\gamma, \hat\rho_{1,1:k,\mytext{CAL}}, \hat\eta_{1,1:k,\mytext{CAL}})  - \hat{\overline \eta}_{1k,\mytext{CAL}}
 =  \hat S_{1k,\mytext{wKM}}\, \hat\varphi_{1ki,\mytext{wKM}},
\end{align*}
or to show that for $0 \le \ell \le k-1$,
\begin{align*}
& \quad \frac{A_i }{\hat\pi_i } \frac{\overline R_{\ell+1, i} }{\hat{\overline \rho}_{1,\ell+1,\mytext{CAL}} } \hat{\underline \eta}_{1,\ell+2,\mytext{CAL}}
1\{ U^{(1)}_i > u_\ell \} ( 1\{U^{(1)}_i > u_{\ell+1} \}- \hat \eta_{1,\ell+1,\mytext{CAL}} )\\
& = n \hat S_{1k,\mytext{wKM}}  \frac{ - A_i \hat w_{1i} } { \sum_{\tau \in I_{\ell+1} \setminus J_{\ell+1} } A_\tau \hat w_{1\tau} }
 1\{i\in I_{\ell+1} \} ( 1\{i\in J_{\ell+1} \} - \hat q_{1,\ell+1}) .
\end{align*}
In fact, the relationship \eqref{eq:prf-IJ-relation} and the expression \eqref{eq:prf-eta-expression} indicate
\begin{align*}
&  A_i \overline R_{\ell+1, i} 1\{U^{(1)}_i > u_{\ell+1} \}  =  A_i  1\{i\in I_{\ell+1}\setminus J_{\ell+1} \}  =  A_i  1\{i\in I_{\ell+1} \}
(1- 1\{ i\in J_{\ell+1} \}) , \\
&  A_i  \overline R_{\ell+1, i} 1\{U^{(1)}_i > u_\ell \}  = A_i  1\{i\in I_{\ell+1} \}, \\
&  \hat \eta_{1,\ell+1,\mytext{CAL}}  = 1- \hat q_{1,\ell+1},
\end{align*}
and hence
\begin{align} \label{eq:prf1-weighted-KM-Vr}
\begin{split}
& \quad \frac{A_i }{\hat\pi_i } 1\{ U^{(1)}_i > u_\ell \} ( 1\{U^{(1)}_i > u_{\ell+1} \}- \hat \eta_{1,\ell+1,\mytext{CAL}} ) \\
& =  - A_i \hat w_{1i} 1\{i\in I_{\ell+1} \} ( 1\{i\in J_{\ell+1} \} - \hat q_{1,\ell+1}) .
\end{split}
\end{align}
Moreover, direct calculation using the expressions \eqref{eq:prf-rho-expression} and \eqref{eq:prf-eta-expression} yields
\begin{align}\label{eq:prf2-weighted-KM-Vr}
\begin{split}
& \quad \frac{1}{n} \frac{1}{\hat{\overline \rho}_{1,\ell+1,\mytext{CAL}} } \hat{\underline \eta}_{1,\ell+2,\mytext{CAL}}
 = \frac{1}{n} \frac{1}{\prod_{j=0}^\ell \frac{\sum_{i\in I_{j+1}} A_i \hat w_{1i} }{\sum_{i\in I_j\setminus J_j } A_i \hat w_{1i} } }
 \hat{\underline \eta}_{1,\ell+2,\mytext{CAL}}\\
& = \left( \prod_{j=0}^\ell \frac{\sum_{i\in I_{j+1} \setminus J_{j+1} } A_i \hat w_{1i} } {\sum_{i\in I_{j+1}} A_i \hat w_{1i} } \right)
 \frac{1}{\sum_{i\in I_{\ell+1} \setminus J_{\ell+1} } A_i \hat w_{1i}  }
 \hat{\underline \eta}_{1,\ell+2,\mytext{CAL}}\\
& = \hat{\overline \eta}_{1,\ell+1,\mytext{CAL}}
 \frac{1}{\sum_{i\in I_{\ell+1} \setminus J_{\ell+1} } A_i \hat w_{1i}  }
 \hat{\underline \eta}_{1,\ell+2,\mytext{CAL}}\\
& = \hat S_{1k,\mytext{wKM}}   \frac{1} { \sum_{\tau \in I_{\ell+1} \setminus J_{\ell+1} } A_\tau \hat w_{1\tau} }
\end{split}
\end{align}
where the second equality uses the assumption $\sum_{i=1}^n A_i \hat\pi_i^{-1}=n$.
Combining the two displays \eqref{eq:prf1-weighted-KM-Vr} and \eqref{eq:prf2-weighted-KM-Vr} gives the desired result.

\vspace{.1in}
\textbf{Proof of Proposition~\ref{pro:weighted-KM-Vb}.}\;
To account for the sampling variation of $\hat\gamma$, an asymptotic expansion for
$\hat S_{1k} (\hat\gamma, \hat\rho_{1,1:k,\mytext{CAL}}, \hat\eta_{1,1:k,\mytext{CAL}})$
can be obtained with an additional term from \eqref{eq:AIPW-expan-fix-gamma}:
\begin{align}\label{eq:AIPW-expan-random-gamma}
\begin{split}
& \quad \hat S_{1k} (\hat\gamma, \hat\rho_{1,1:k,\mytext{CAL}}, \hat\eta_{1,1:k,\mytext{CAL}})
= \hat S_{1k} (\breve\gamma, \breve\rho_{1,1:k,\mytext{CAL}}, \breve\eta_{1,1:k,\mytext{CAL}}) \\
& \quad + E \left\{\frac{\partial }{\partial\gamma} \hat S_{1k} (\gamma, \rho_{1,1:k}, \eta_{1,1:k }) \right\}^\T
 \Big|_{\breve\gamma,\breve\rho_{1,1:k,\mytext{CAL}}, \breve\eta_{1,1:k,\mytext{CAL}}}  ( \hat\gamma-\breve\gamma) \\
& \quad + E \left\{\frac{\partial }{\partial \rho_{1,1:k}} \hat S_{1k} (\gamma, \rho_{1,1:k}, \eta_{1,1:k }) \right\}^\T
 \Big|_{\breve\gamma,\breve\rho_{1,1:k,\mytext{CAL}}, \breve\eta_{1,1:k,\mytext{CAL}}}  ( \hat\rho_{1,1:k,\mytext{CAL}} - \breve\rho_{1,1:k,\mytext{CAL}}) \\
& \quad + E \left\{\frac{\partial }{\partial \eta_{1,1:k}} \hat S_{1k} (\gamma, \rho_{1,1:k}, \eta_{1,1:k }) \right\}^\T
 \Big|_{\breve\gamma,\breve\rho_{1,1:k,\mytext{CAL}}, \breve\eta_{1,1:k,\mytext{CAL}}}   ( \hat\eta_{1,1:k,\mytext{CAL}} - \breve\eta_{1,1:k,\mytext{CAL}}) + o_p(n^{-1/2}),
\end{split}
\end{align}
where $(\breve\gamma, \breve\rho_{1,1:k,\mytext{CAL}}, \breve\eta_{1,1:k,\mytext{CAL}})$ are the limit values of
$(\hat\gamma, \hat\rho_{1,1:k,\mytext{CAL}}, \hat\eta_{1,1:k,\mytext{CAL}})$.
Then \eqref{eq:CAL-rho-pop} and \eqref{eq:CAL-eta-pop} are satisfied with $\hat\pi_i$ replaced by $\breve\pi_i = \pi(X_i;\breve\gamma)$.
Similarly as in the proof of Proposition~\ref{pro:weighted-KM-Vr}, it can be shown that
\begin{align*}
& 0 = E \left\{\frac{\partial }{\partial \rho_{1,1:k}} \hat S_{1k} (\gamma, \rho_{1,1:k}, \eta_{1,1:k }) \right\}
 \Big|_{\breve\gamma,\breve\rho_{1,1:k,\mytext{CAL}}, \breve\eta_{1,1:k,\mytext{CAL}}} ,\\
& 0 = E \left\{\frac{\partial }{\partial \eta_{1,1:k}} \hat S_{1k} (\gamma, \rho_{1,1:k}, \eta_{1,1:k }) \right\}
 \Big|_{\breve\gamma,\breve\rho_{1,1:k,\mytext{CAL}}, \breve\eta_{1,1:k,\mytext{CAL}}} ,
\end{align*}
using equations \eqref{eq:CAL-rho-pop}, \eqref{eq:CAL-eta-pop} and $E( A \pi(X;\breve\gamma)^{-1}) =1$.
Moreover, if model \eqref{eq:const-m} is correctly specified, then it can be shown using \eqref{eq:phi-or} that
\begin{align*}
 &\quad E \left\{\frac{\partial }{\partial\gamma} \hat S_{1k} (\gamma, \rho_{1,1:k}, \eta_{1,1:k }) \right\}
 \Big|_{\breve\gamma,\breve\rho_{1,1:k,\mytext{CAL}}, \breve\eta_{1,1:k,\mytext{CAL}}} \\
 & = E \Big[ \frac{1}{n} \sum_{i=1}^n  \sum_{\ell=0}^{k-1} \frac{\partial\pi(X_i;\gamma)}{\partial\gamma}
 \frac{-A_i}{\pi(X_i;\gamma)^2 } \frac{\overline R_{\ell+1,i} }{\overline \rho_{1,\ell+1} } \underline \eta_{1,\ell+2} \\
 & \qquad \times 1\{ U^{(1)}_i > u_\ell \} ( 1\{U^{(1)}_i > u_{\ell+1}\}- \eta_{1,\ell+1} ) \Big]
 \Big|_{\breve\gamma,\breve\rho_{1,1:k,\mytext{CAL}}, \breve\eta_{1,1:k,\mytext{CAL}}} \\
 & = 0.
\end{align*}
By combining the preceding two displays, if model \eqref{eq:const-m} is correctly specified, then the asymptotic expansion \eqref{eq:AIPW-expan-random-gamma} reduces to
\begin{align*}
& \quad \hat S_{1k} (\hat\gamma, \hat\rho_{1,1:k,\mytext{CAL}}, \hat\eta_{1,1:k,\mytext{CAL}})
= \hat S_{1k} (\breve\gamma, \breve\rho_{1,1:k,\mytext{CAL}}, \breve\eta_{1,1:k,\mytext{CAL}}) + o_p(n^{-1/2}).
\end{align*}

Next, if model \eqref{eq:const-m} is correctly specified, then from the preceding expansion, the asymptotic variance of $\hat S_{1k} (\hat\gamma, \hat\rho_{1,1:k,\mytext{CAL}}, \hat\eta_{1,1:k,\mytext{CAL}})$
is determined as
\begin{align*}
& \quad E \left[ \left\{ \hat S_{1k} (\breve\gamma, \breve\rho_{1,1:k,\mytext{CAL}}, \breve\eta_{1,1:k,\mytext{CAL}}) - S_{1k} \right\}^2 \right] \\
& = \frac{1}{n^2} E \left[ \sum_{i=1}^n
\left\{ \sum_{\ell=0}^{k-1} \frac{A_i }{\breve\pi_i } \frac{\overline R_{\ell+1, i} }{\breve{\overline \rho}_{1,\ell+1,\mytext{CAL}} } \breve{\underline \eta}_{1,\ell+2,\mytext{CAL}}
1\{ U^{(1)}_i > u_\ell \} ( 1\{U^{(1)}_i > u_{\ell+1}\}- \breve \eta_{1,\ell+1,\mytext{CAL}} ) \right\}^2 \right] \\
& = \frac{1}{n^2} E \left[ \sum_{i=1}^n
 \sum_{\ell=0}^{k-1} \frac{A_i }{ ( \breve\pi_i)^2 } \frac{\overline R_{\ell+1, i} }{ ( \breve{\overline \rho}_{1,\ell+1,\mytext{CAL}} )^2 } (\breve{\underline \eta}_{1,\ell+2,\mytext{CAL}})^2
1\{ U^{(1)}_i > u_\ell \} ( 1\{U^{(1)}_i > u_{\ell+1}\}- \breve \eta_{1,\ell+1,\mytext{CAL}} )^2 \right] \\
& = \frac{1}{n^2} E \left[ \sum_{i=1}^n
 \sum_{\ell=0}^{k-1} \frac{A_i }{ ( \breve\pi_i)^2 } \frac{\overline R_{\ell+1, i} }{ ( \breve{\overline \rho}_{1,\ell+1,\mytext{CAL}} )^2 } (\breve{\underline \eta}_{1,\ell+2,\mytext{CAL}})^2
 \breve \eta_{1,\ell+1,\mytext{CAL}} (1-\breve \eta_{1,\ell+1,\mytext{CAL}})  \right] .
\end{align*}
With model \eqref{eq:const-m} correctly specified,
the first equality above follows because $S_{1k} = \breve{\overline\rho}_{1k,\mytext{CAL}}$ and the individual observations are independent across $i$,
the second equality follows because the individual terms
\begin{align*}
\frac{A_i }{\breve\pi_i } \frac{\overline R_{\ell+1, i} }{\breve{\overline \rho}_{1,\ell+1,\mytext{CAL}} } \breve{\underline \eta}_{1,\ell+2,\mytext{CAL}}
1\{ U^{(1)}_i > u_\ell \} ( 1\{U^{(1)}_i > u_{\ell+1}\}- \breve \eta_{1,\ell+1,\mytext{CAL}} )
\end{align*}
are uncorrelated across $\ell$, and the third equality follows because
\begin{align*}
& \quad E \{ ( 1\{U^{(1)}_i > u_{\ell+1}\}- \breve \eta_{1,\ell+1,\mytext{CAL}} )^2 | R_{\ell+1,i}=1, X_i, A_i=1) \\
& =   \breve \eta_{1,\ell+1,\mytext{CAL}}  (1-\breve \eta_{1,\ell+1,\mytext{CAL}}).
\end{align*}
Then the asymptotic variance of $\hat S_{1k} (\hat\gamma, \hat\rho_{1,1:k,\mytext{CAL}}, \hat\eta_{1,1:k,\mytext{CAL}})$
can be consistently estimated by
\begin{align*}
\frac{1}{n^2}  \sum_{i=1}^n
 \sum_{\ell=0}^{k-1} \frac{A_i }{ ( \hat\pi_i)^2 } \frac{\overline R_{\ell+1, i} }{ ( \hat{\overline \rho}_{1,\ell+1,\mytext{CAL}} )^2 } (\hat{\underline \eta}_{1,\ell+2,\mytext{CAL}})^2
 \hat \eta_{1,\ell+1,\mytext{CAL}} (1-\hat \eta_{1,\ell+1,\mytext{CAL}}) ,
\end{align*}
which we show is algebraically identical to $\hat V_{\mytext{b}} ( \hat S_{1k,\mytext{wKM}} )$.
By comparing the preceding display and the expression in \eqref{eq:weighted-KM-Vb}, it suffices to show that for $0\le \ell\le k-1$,
\begin{align*}
 & \quad \frac{1}{n^2}  \sum_{i=1}^n
  \frac{A_i }{ ( \hat\pi_i)^2 } \frac{\overline R_{\ell+1, i} }{ ( \hat{\overline \rho}_{1,\ell+1,\mytext{CAL}} )^2 } (\hat{\underline \eta}_{1,\ell+2,\mytext{CAL}})^2
 \hat \eta_{1,\ell+1,\mytext{CAL}} (1-\hat \eta_{1,\ell+1,\mytext{CAL}}) \\
 & = \hat S_{1k,\mytext{wKM}} ^2
\frac{\sum_{i \in I_{\ell+1} } A_i \hat w_{1i}^2 }{ (\sum_{i\in I_{\ell+1} } A_i \hat w_{1i})^2 } \frac{\hat q_{1,\ell+1} }{1- \hat q_{1,\ell+1}} .
\end{align*}
This can be verified by combining the following equations:
\begin{align*}
& \sum_{i=1}^n
  \frac{A_i }{ ( \hat\pi_i)^2 }\overline R_{\ell+1, i} = \sum_{i \in I_{\ell+1} } A_i \hat w_{1i}^2 ,\\
& \quad \frac{1}{n} \frac{1}{\hat{\overline \rho}_{1,\ell+1,\mytext{CAL}} } \hat{\underline \eta}_{1,\ell+2,\mytext{CAL}}
 = \hat S_{1k,\mytext{wKM}}   \frac{1} { \sum_{\tau \in I_{\ell+1} \setminus J_{\ell+1} } A_\tau \hat w_{1\tau} },\\
&  \hat \eta_{1,\ell+1,\mytext{CAL}} (1-\hat \eta_{1,\ell+1,\mytext{CAL}}) =
 \frac{(\sum_{i\in I_{\ell+1}\setminus J_{\ell+1} } A_i \hat w_{1i})^2}{ (\sum_{i\in I_{\ell+1} } A_i \hat w_{1i})^2 } \frac{\hat q_{1,\ell+1} }{1- \hat q_{1,\ell+1}} ,
\end{align*}
where the first equality follows by definition, the second equality follows from \eqref{eq:prf2-weighted-KM-Vr},
and the third equality follows from the expression \eqref{eq:prf-eta-expression}.

\vspace{.1in}
\textbf{Proof of Proposition~\ref{pro:AIPW-CAL}.}\;
To account for the sampling variations of
$\hat\xi_{\mytext{CAL}} =(\hat\gamma_{1,\mytext{CAL}}, $ $\hat\rho_{1,1:k,\mytext{CAL}}, \allowbreak \hat\eta_{1,1:k,\mytext{CAL}}, \hat\alpha_{1k,\mytext{CAL}})$
about the limit values
$\breve\xi_{\mytext{CAL}}=(\breve\gamma_{1,\mytext{CAL}}, \breve\rho_{1,1:k,\mytext{CAL}}, \breve\eta_{1,1:k,\mytext{CAL}}, \breve\alpha_{1k,\mytext{CAL}})$,
an asymptotic expansion for
$\hat S_{1k,\mytext{CAL}} = \hat S_{1k} (\hat\xi_{\mytext{CAL}} )$
can be obtained as
\begin{align}\label{eq:AIPW-expan-full}
\begin{split}
& \quad \hat S_{1k} (\hat\xi_{\mytext{CAL}})
= \hat S_{1k} (\breve\xi_{\mytext{CAL}}) \\
& \quad + E \left\{\frac{\partial }{\partial\gamma} \hat S_{1k} (\xi) \right\}^\T
 \Big|_{\breve\xi_{\mytext{CAL}}}  (\hat\gamma_{1,\mytext{CAL}} - \breve\gamma_{1,\mytext{CAL}} ) \\
& \quad + E \left\{\frac{\partial }{\partial \rho_{1,1:k}} \hat S_{1k} (\xi) \right\}^\T
 \Big|_{\breve\xi_{\mytext{CAL}}}  ( \hat\rho_{1,1:k,\mytext{CAL}} - \breve\rho_{1,1:k,\mytext{CAL}}) \\
& \quad + E \left\{\frac{\partial }{\partial \eta_{1,1:k}} \hat S_{1k} (\xi) \right\}^\T
 \Big|_{\breve\xi_{\mytext{CAL}}}   ( \hat\eta_{1,1:k,\mytext{CAL}} - \breve\eta_{1,1:k,\mytext{CAL}}) \\
& \quad + E \left\{\frac{\partial }{\partial\alpha_{1k}} \hat S_{1k} (\xi) \right\}^\T
 \Big|_{\breve\xi_{\mytext{CAL}}}  (\hat\alpha_{1k,\mytext{CAL}} - \breve\alpha_{1k,\mytext{CAL}} ) + o_p(n^{-1/2}),
\end{split}
\end{align}
where $\xi= (\gamma, \rho_{1,1:k}, \eta_{1,1:k}, \alpha_{1k})$.
The limit values
$\breve\xi_{\mytext{CAL}}$
are jointly determined as the solutions to the population version
of the equations \eqref{eq:CAL-rho}, \eqref{eq:CAL-eta}, \eqref{eq:CAL-gamma}, and \eqref{eq:CAL-alpha} for $\hat\xi_{\mytext{CAL}}$:
\begin{align}
& 0= E \left[\frac{1}{n} \sum_{i=1}^n  \frac{A_i }{\pi(X_i;\gamma) } \overline R_{\ell i}
\left( \frac{R_{\ell+1, i} }{\rho_{1,\ell+1} }-1 \right) 1\{ U^{(1)}_i > u_\ell \} \right], \quad \ell=0,\ldots,k-1,  \label{eq:CAL-rho-pop-b} \\
& 0 = E \left[ \frac{1}{n} \sum_{i=1}^n  \frac{A_i }{\pi(X_i;\gamma) } \overline R_{\ell+1, i}
1\{ U^{(1)}_i > u_\ell \} ( 1\{U^{(1)}_i > u_{\ell+1}-  \eta_{1,\ell+1} ) \right], \quad \ell=0,\ldots,k-1,  \label{eq:CAL-eta-pop-b} \\
& 0 = E \left[ \frac{1}{n} \sum_{i=1}^n \left\{ \frac{A_i}{\pi(X_i;\gamma)} -1 \right\} f(X_i)  \right],  \label{eq:CAL-gamma-pop} \\
& 0 = E \left[ \frac{1}{n} \sum_{i=1}^n A_i \frac{1-\pi(X_i;\gamma)} {\pi(X_i;\gamma) }
  \left\{ \hat U^{(1)}_{ki} (\rho_{1,1:k},\eta_{1,1:k}) - \mu_{1k} (X_i;\alpha_{1k}) \right\} f(X_i) \right],  \label{eq:CAL-alpha-pop}
\end{align}
where $\hat U^{(1)}_{ki} (\rho_{1,1:k},\eta_{1,1:k})$ is defined as $\hat U^{(1)}_{ki} (\hat\rho_{1,1:k},\hat\eta_{1,1:k})$
with $(\hat\rho_{1,1:k},\hat\eta_{1,1:k})$ replaced by $(\rho_{1,1:k},$ $\eta_{1,1:k})$.
First, similarly as in the proof of Proposition~\ref{pro:weighted-KM-Vr}, it can be shown that
\begin{align*}
& 0 = E \left\{\frac{\partial }{\partial \rho_{1,1:k}} \hat S_{1k} (\xi) \right\}
 \Big|_{\breve\xi_{\mytext{CAL}}} ,\\
& 0 = E \left\{\frac{\partial }{\partial \eta_{1,1:k}} \hat S_{1k} (\xi) \right\}
 \Big|_{\breve\xi_{\mytext{CAL}}} ,
\end{align*}
using equations \eqref{eq:CAL-rho-pop-b}, \eqref{eq:CAL-eta-pop-b} and $E( A \pi(X;\breve\gamma)^{-1}) =1$ from \eqref{eq:CAL-gamma-pop} with
$f$ including an intercept.
Second, it can be shown using the expression of $\varphi_{1ki}(\xi)$ and \eqref{eq:CAL-alpha-pop} that
\begin{align*}
 &\quad E \left\{\frac{\partial }{\partial\gamma} \hat S_{1k} (\xi) \right\}
 \Big|_{\breve\xi_{\mytext{CAL}}} \\
 & = E \Big[ \frac{1}{n} \sum_{i=1}^n  \frac{\partial\pi(X_i;\gamma)}{\partial\gamma}
 \frac{-A_i}{\pi(X_i;\gamma)^2} \left\{ \hat U^{(1)}_{ki} (\rho_{1,1:k},\eta_{1,1:k}) - \mu_{1k}(X_i;\alpha_{1k}) \right\} \Big]
 \Big|_{\breve\xi_{\mytext{CAL}}} \\
 & = 0 ,
\end{align*}
where $\frac{\partial\pi(X_i;\gamma)}{\partial\gamma}
 \frac{1}{\pi(X_i;\gamma)^2} =\frac{1-\pi(X_i;\gamma)} {\pi(X_i;\gamma) } f(X_i)$ for logistic regression \eqref{eq:PS-model}.
Second, if model \eqref{eq:PS-model} is correctly specified, then it can be shown using the expression of $\varphi_{1ki}(\xi)$ that
\begin{align} \label{eq:derv-alpha}
\begin{split}
 &\quad E \left\{\frac{\partial }{\partial\alpha_{1k}} \hat S_{1k} (\xi) \right\}
 \Big|_{\breve\xi_{\mytext{CAL}}} \\
 & = E \Big[ \frac{-1}{n} \sum_{i=1}^n
 \left\{ \frac{A_i}{\pi(X_i;\gamma) } -1 \right\}  \frac{\partial\mu_{1k}(X_i;\alpha_{1k})}{\partial\alpha_{1k}} \Big]
 \Big|_{\breve\xi_{\mytext{CAL}}} \\
 & = 0.
\end{split}
\end{align}
By combining the preceding three displays, if model \eqref{eq:PS-model} is correctly specified, then the asymptotic expansion \eqref{eq:AIPW-expan-full} reduces to \eqref{eq:CAL-expan}:
\begin{align*}
& \quad \hat S_{1k} (\hat\xi_{\mytext{CAL}})
= \hat S_{1k} (\breve\xi_{\mytext{CAL}}) + o_p(n^{-1/2}).
\end{align*}
From this asymptotic expansion, the consistency of the variance estimator \eqref{eq:AIPW-CAL-Vr} can be verified by standard arguments.

\vspace{.1in}
\textbf{Proof of Proposition~\ref{pro:AIPW-CAL-lin}.}\;
The result can be shown similarly as Proposition~\ref{pro:AIPW-CAL}, except for the following modification.
Equation \eqref{eq:derv-alpha} can be obtained using \eqref{eq:CAL-gamma-pop} but without requiring model \eqref{eq:PS-model} to be correctly specified:
\begin{align*}
 &\quad E \left\{\frac{\partial }{\partial\alpha_{1k}} \hat S_{1k} (\xi) \right\}
 \Big|_{\breve\xi_{\mytext{CAL}}} \\
 & = E \Big[ \frac{-1}{n} \sum_{i=1}^n
 \left\{ \frac{A_i}{\pi(X_i;\gamma) } -1 \right\}  \frac{\partial\mu_{1k}(X_i;\alpha_{1k})}{\partial\alpha_{1k}} \Big]
 \Big|_{\breve\xi_{\mytext{CAL}}} \\
 & = E \Big[ \frac{-1}{n} \sum_{i=1}^n
 \left\{ \frac{A_i}{\pi(X_i;\gamma) } -1 \right\} f(X_i) \Big]
 \Big|_{\breve\xi_{\mytext{CAL}}} \\
 & = 0,
\end{align*}
because $\frac{\partial\mu_{1k}(X_i;\alpha_{1k})}{\partial\alpha_{1k}} = f(X_i)$ for linear regression \eqref{eq:OR-model-lin}.

\vspace{.1in}
\textbf{Proof of equation \eqref{eq:AIPW-reduction}.}\;
By the definitions of $\hat{S}_{1k, \mytext{CAL, lin}}$ and $\hat{S}_{1k}(\hat{\gamma}_{1, \mytext{CAL}}, \hat{\rho}_{1, 1:k, \mytext{CAL}}, \hat{\eta}_{1, 1:k, \mytext{CAL}})$ in \eqref{eq:AIPW-new} and \eqref{eq:AIPW}, the only difference between these two estimators lies in their final terms, namely,
\begin{align*}
\frac{1}{n}\sum_{i=1}^n (\frac{A_i}{\hat{\pi}_i} - 1)\hat{\mu}_{1ki}\quad \text{and}\quad \frac{1}{n}\sum_{i=1}^n \{ (\frac{A_i}{\hat{\pi}_i} - 1)\prod_{j=1}^k\hat{m}_{1ji} \}.
\end{align*}
Under linear OR model \eqref{eq:OR-model-lin} and constant CSP model \eqref{eq:const-m}, these two terms can be rewritten as
\begin{align*}
\frac{1}{n}\hat{\alpha}^\T_{1k, \mytext{CAL}}\sum_{i=1}^n (\frac{A_i}{\hat{\pi}_i} - 1)f(X) \quad \text{and}\quad \frac{1}{n}\prod_{j=1}^k\hat{\eta}_{1j, \mytext{CAL}}\sum_{i=1}^n (\frac{A_i}{\hat{\pi}_i} - 1).
\end{align*}
By the definition of $\hat{\gamma}_{1, \mytext{CAL}}$ in \eqref{eq:CAL-gamma} and noting that $f(X)$ includes a constant term, both expressions equals 0. Therefore,
$\hat{S}_{1k, \mytext{CAL, lin}} = \hat{S}_{1k}(\hat{\gamma}_{1, \mytext{CAL}}, \hat{\rho}_{1, 1:k, \mytext{CAL}}, \hat{\eta}_{1, 1:k, \mytext{CAL}})$.

The equality between $\hat{S}_{1k}(\hat{\gamma}_{1, \mytext{CAL}}, \hat{\rho}_{1, 1:k, \mytext{CAL}}, \hat{\eta}_{1, 1:k, \mytext{CAL}})$ and $\hat{S}_{\mytext{wKM}}(\hat{\gamma}_{1, \mytext{CAL}})$ follows directly from Proposition~\ref{pro:weighted-KM}.
This completes the proof.

\vspace{.1in}
\textbf{Proof of Proposition~\ref{pro:CAL-theta}.}\;
Suppose that model \eqref{eq:PS-model} is correctly specified, with the true value $\gamma^*$. From Proposition~\ref{pro:AIPW-CAL},
the calibrated estimator $\hat S_{ak,\mytext{CAL}}$ satisfies the asymptotic expansion
\begin{align}
 \hat S_{ak,\mytext{CAL}} & = \hat S_{ak} (\breve\gamma_{a,\mytext{CAL}}, \breve\rho_{a,1:k,\mytext{CAL}}, \breve\eta_{a,1:k,\mytext{CAL}}, \breve\alpha_{ak,\mytext{CAL}})  + o_p (n^{-1/2}) \nonumber \\
 & =  \frac{1}{n} \sum_{i=1}^n \underbrace{ \varphi_{aki} (\breve\gamma_{a,\mytext{CAL}}, \breve\rho_{a,1:k,\mytext{CAL}}, \breve\eta_{a,1:k,\mytext{CAL}}, \breve\alpha_{ak,\mytext{CAL}}) }_{\breve\varphi_{aki}} + o_p (n^{-1/2}) . \label{eq:CAL-theta-expan-S}
\end{align}
The limit value of $\hat S_{ak,\mytext{CAL}} $, denoted as $\breve S_{ak,\mytext{CAL}} $, coincides with $S_{ak}$
if models \eqref{eq:PS-model} and \eqref{eq:const-pi} are correctly specified
but may differ from $S_{ak}$ otherwise.
From Tan (2020), the calibrated estimator $\hat W_{ak,\mytext{CAL}}$ is consistent for
$\breve W_{ak,\mytext{CAL}} = W_{ak}$ (because model \eqref{eq:PS-model} is correctly specified) and satisfies the asymptotic expansion
\begin{align}
 \hat W_{ak,\mytext{CAL}} &= \hat W_{ak} ( \breve\gamma_{a,\mytext{CAL}}, \breve\beta_{ak,\mytext{CAL}}) + o_p(n^{-1/2}) \nonumber \\
 & = \frac{1}{n} \sum_{i=1}^n \underbrace{ \psi_{aki} ( \breve\gamma_{a,\mytext{CAL}}, \breve\beta_{ak,\mytext{CAL}}) }_{ \breve\psi_{aki}} + o_p(n^{-1/2}),
 \label{eq:CAL-theta-expan-W}
\end{align}
where $( \breve\gamma_{a,\mytext{CAL}}, \breve\beta_{ak,\mytext{CAL}})$ are the limit values of
$( \hat\gamma_{a,\mytext{CAL}}, \hat\beta_{ak,\mytext{CAL}})$, and $\breve\gamma_{a,\mytext{CAL}}=\gamma^*$.

By standard arguments in theory of M-estimation, it can be shown that $\hat\theta_{\mytext{CAL}}$
converges in probability to a limit value $\breve\theta_{\mytext{CAL}}$, which is determined as a solution to
\begin{align}
 0 = \sum_{k=1}^K \frac{ W_{1k} W_{0k} }
 { W_{1k} \me^\theta + W_{0k} } (\breve q_{1k,\mytext{CAL}} - \breve q_{0k,\mytext{CAL}} \me^\theta )  ,  \label{eq:CAL-theta-expan-prf0}
\end{align}
where $\breve q_{ak,\mytext{CAL}} = (\breve S_{a,k-1,\mytext{CAL}} - \breve S_{ak,\mytext{CAL}} )/\breve S_{a,k-1,\mytext{CAL}} $.
The limit value $\breve\theta_{\mytext{CAL}} $ coincides with $\breve\theta_{\mytext{wBP}}$
if models \eqref{eq:PS-model} and \eqref{eq:const-pi} are correctly specified and hence $\breve S_{ak,\mytext{CAL}} = S_{ak}$
but may differ from $\breve\theta_{\mytext{wBP}}$ otherwise.
Moreover, $\hat\theta_{\mytext{CAL}}$ satisfies the asymptotic expansion
\begin{align}
 & \quad \hat\theta_{\mytext{CAL}} - \breve\theta_{\mytext{CAL}} \nonumber \\
 & = H_{\mytext{CAL}} (\theta)^{-1}\; \sum_{k=1}^K \left\{  \frac{ \hat W_{1k,\mytext{CAL}} \hat W_{0k,\mytext{CAL}} }
 { \hat W_{1k,\mytext{CAL}} \me^\theta + \hat W_{0k,\mytext{CAL}} } ( \hat q_{1k,\mytext{CAL}} - \hat q_{0k,\mytext{CAL}} \me^\theta ) \right\}
 \Big|_{\breve\theta_{\mytext{CAL}}} + o_p(n^{-1/2}), \label{eq:CAL-theta-expan-prf1}
\end{align}
where
\begin{align*}
 H_{\mytext{CAL}} (\theta) &=  (-1)\frac{\dif}{\dif \theta} \sum_{k=1}^K \frac{ W_{1k} W_{0k} }
 { W_{1k} \me^\theta + W_{0k} } (\breve q_{1k,\mytext{CAL}} - \breve q_{0k,\mytext{CAL}} \me^\theta ) \\
 &= \sum_{k=1}^K \frac{W_{1k} W_{0k} \me^\theta}
 { ( W_{1k} \me^\theta + W_{0k} )^2 }
 ( W_{1k} \breve q_{1k,\mytext{CAL}}+ W_{0k} \breve q_{0k,\mytext{CAL}} ).
\end{align*}
We show that
\begin{align}
 & \quad \sum_{k=1}^K \frac{ \hat W_{1k,\mytext{CAL}} \hat W_{0k,\mytext{CAL}} }
 { \hat W_{1k,\mytext{CAL}} \me^\theta + \hat W_{0k,\mytext{CAL}} } ( \hat q_{1k,\mytext{CAL}} - \hat q_{0k,\mytext{CAL}} \me^\theta )
 \Big|_{\breve\theta_{\mytext{CAL}}} \nonumber \\
 & =  \frac{1}{n} \sum_{i=1}^n  \sum_{k=1}^K \breve\varphi_{ki} (\breve\theta_{\mytext{CAL}}) + o_p(n^{-1/2}),\label{eq:CAL-theta-expan-prf2}
\end{align}
where $\breve\varphi_{ki} (\theta)$ is defined as $\hat\varphi_{ki} (\theta)$ prior to Proposition~\ref{pro:CAL-theta},
with $\hat W_{ak,\mytext{CAL}}$ replaced by $W_{ak}$, $\hat q_{ak,\mytext{CAL}}$ replaced by $\breve q_{ak,\mytext{CAL}}$,
$\hat\varphi_{aki}$ replaced by $\breve\varphi_{aki}$,
and $\hat\psi_{aki}$ replaced by $\breve\psi_{aki}$, i.e.,
\begin{align*}
 & \breve \varphi_{ki} (\theta) =
  \frac{ W_{0k} W_{1k}} { W_{1k} \me^\theta + W_{0k} }
  \left(\frac{ -\breve\varphi_{1ki} + (1-\breve q_{1k,\mytext{CAL}}) \breve\varphi_{1,k-1,i} }{ \breve S_{1,k-1,\mytext{CAL}} } \right) \\
  & \quad - \frac{ W_{0k} W_{1k} \me^\theta  } { W_{1k} \me^\theta + W_{0k} }
  \left( \frac{ - \breve\varphi_{0ki} + (1-\breve q_{0k,\mytext{CAL}}) \breve\varphi_{0,k-1,i} }{ \breve S_{0,k-1,\mytext{CAL}} } \right)  \\
  & \quad + \frac{\breve q_{1k,\mytext{CAL}} -  \breve q_{0k,\mytext{CAL}} \me^\theta }{(W_{1k} \me^\theta + W_{0k})^2 }
 \left\{ W_{0k}^2 (\breve\psi_{1ki} - W_{1k} ) +
 \me^\theta W_{1k}^2 (\breve\psi_{0ki} - W_{0k} ) \right\} .
\end{align*}
Combining \eqref{eq:CAL-theta-expan-prf1} and \eqref{eq:CAL-theta-expan-prf2} then yields
\begin{align*}
  \hat\theta_{\mytext{CAL}} - \breve\theta_{\mytext{CAL}}
  = H_{\mytext{CAL}} (\theta)^{-1}\;  \sum_{k=1}^K \breve\varphi_{ki} (\breve\theta_{\mytext{CAL}})  \Big|_{\breve\theta_{\mytext{CAL}}} + o_p(n^{-1/2}),
\end{align*}
from which the consistency of the variance estimator \eqref{eq:CAL-theta-Vr} for the asymptotic variance of $\hat\theta_{\mytext{CAL}}$ can be deduced by standard arguments.

To show \eqref{eq:CAL-theta-expan-prf2}, consider the decomposition
\begin{align*}
 & \quad \sum_{k=1}^K  \frac{ \hat W_{1k,\mytext{CAL}} \hat W_{0k,\mytext{CAL}} }
 { \hat W_{1k,\mytext{CAL}} \me^\theta + \hat W_{0k,\mytext{CAL}} } ( \hat q_{1k,\mytext{CAL}} - \hat q_{0k,\mytext{CAL}} \me^\theta )  \\
 & =\sum_{k=1}^K  \Bigg[   \frac{ \hat W_{1k,\mytext{CAL}} \hat W_{0k,\mytext{CAL}} }
 { \hat W_{1k,\mytext{CAL}} \me^\theta + \hat W_{0k,\mytext{CAL}} } (\hat q_{1k,\mytext{CAL}} - \breve q_{1k,\mytext{CAL}} ) \\
 & \qquad\qquad -  \frac{ \hat W_{1k,\mytext{CAL}} \hat W_{0k,\mytext{CAL}} \me^\theta }
 { \hat W_{1k,\mytext{CAL}} \me^\theta + \hat W_{0k,\mytext{CAL}} } (\hat q_{0k,\mytext{CAL}} - \breve q_{0k,\mytext{CAL}} ) \\
 & \qquad\qquad + (\breve q_{1k,\mytext{CAL}} -  \breve q_{0k,\mytext{CAL}} \me^\theta )
  \left\{ \frac{ \hat W_{1k,\mytext{CAL}} \hat W_{0k,\mytext{CAL}} }
 { \hat W_{1k,\mytext{CAL}} \me^\theta + \hat W_{0k,\mytext{CAL}} }  -
 \frac{ W_{1k} W_{0k} }
 { W_{1k} \me^\theta + W_{0k} } \right\} \Bigg],
\end{align*}
using the fact that $\breve\theta_{\mytext{CAL}}$ satisfies \eqref{eq:CAL-theta-expan-prf0}.
Here and subsequently, $\theta$ is evaluated at $\breve\theta_{\mytext{CAL}}$.
First, by standard manipulation and asymptotic expansion \eqref{eq:CAL-theta-expan-S} for $\hat S_{ak,\mytext{CAL}}$,
\begin{align*}
& \quad \sum_{k=1}^K  \frac{ \hat W_{1k,\mytext{CAL}} \hat W_{0k,\mytext{CAL}} }
 { \hat W_{1k,\mytext{CAL}} \me^\theta + \hat W_{0k,\mytext{CAL}} } ( \hat q_{1k,\mytext{CAL}} - \breve q_{1k,\mytext{CAL}} )\\
 &= \sum_{k=1}^K  \frac{ W_{0k} W_{1k}} { W_{1k} \me^\theta + W_{0k} }
 \left( \frac{-\hat S_{1k,\mytext{CAL}} + (1- \breve q_{1k,\mytext{CAL}}) \hat S_{1,k-1,\mytext{CAL}} }{\hat S_{1,k-1,\mytext{CAL}} } \right)  + o_p(n^{-1/2}) \\
& = \frac{1}{n}\sum_{i=1}^n \sum_{k=1}^K   \frac{ W_{0k} W_{1k}} { W_{1k} \me^\theta + W_{0k} }
 \left( \frac{ -\breve\varphi_{1ki} + (1-\breve q_{1k,\mytext{CAL}}) \breve\varphi_{1,k-1,i} }{ \breve S_{1,k-1,\mytext{CAL}} } \right)
 + o_p(n^{-1/2}),
\end{align*}
and similarly
\begin{align*}
& \quad \sum_{k=1}^K  \frac{ \hat W_{0k,\mytext{CAL}} \hat W_{0k,\mytext{CAL}} \me^\theta }
 { \hat W_{0k,\mytext{CAL}} \me^\theta + \hat W_{0k,\mytext{CAL}} } ( \hat q_{0k,\mytext{CAL}} - \breve q_{0k,\mytext{CAL}} )\\
& = \frac{1}{n}\sum_{i=1}^n \sum_{k=1}^K   \frac{ W_{0k} W_{0k} \me^\theta} { W_{0k} \me^\theta + W_{0k} }
 \left( \frac{- \breve\varphi_{0ki} + (1-\breve q_{0k,\mytext{CAL}}) \breve\varphi_{0,k-1,i} }{ \breve S_{0,k-1,\mytext{CAL}} } \right)
 + o_p(n^{-1/2}) .
\end{align*}
Second, by a Taylor expansion and asymptotic expansion \eqref{eq:CAL-theta-expan-W} for $\hat W_{ak,\mytext{CAL}}$,
\begin{align*}
& \quad \sum_{k=1}^K
 (\breve q_{1k,\mytext{CAL}} -  \breve q_{0k,\mytext{CAL}} \me^\theta )
  \left\{ \frac{ \hat W_{1k,\mytext{CAL}} \hat W_{0k,\mytext{CAL}} }
 { \hat W_{1k,\mytext{CAL}} \me^\theta + \hat W_{0k,\mytext{CAL}} }  -
 \frac{ W_{1k} W_{0k} }
 { W_{1k} \me^\theta + W_{0k} } \right\}\\
 & =\sum_{k=1}^K
 \frac{\breve q_{1k,\mytext{CAL}} -  \breve q_{0k,\mytext{CAL}} \me^\theta} { ( W_{1k}\me^\theta + W_{0k} )^2}
 \left\{ W_{0k}^2 (\hat W_{1k,\mytext{CAL}}- W_{1k}) + \me^\theta W_{1k}^2 (\hat W_{0k,\mytext{CAL}}- W_{0k})  \right\}
 + o_p(n^{-1/2}) \\
 & = \frac{1}{n}\sum_{i=1}^n \sum_{k=1}^K
\frac{\breve q_{1k,\mytext{CAL}} -  \breve q_{0k,\mytext{CAL}} \me^\theta} { ( W_{1k}\me^\theta + W_{0k} )^2}
 \left\{ W_{0k}^2 (\breve \psi_{1ki} - W_{1k}) + \me^\theta W_{1k}^2 (\breve \psi_{0ki}- W_{0k})  \right\}
 + o_p(n^{-1/2}).
\end{align*}
Combining the preceding four displays gives the expansion \eqref{eq:CAL-theta-expan-prf2}.

\vspace{.1in}
\textbf{Proof of asymptotic expansion \eqref{eq:expand}.}\;
The asymptotic expansion \eqref{eq:expand} is similar to \eqref{eq:CAL-theta-expan-prf2} in the proof of Proposition~\ref{pro:CAL-theta},
where $\theta$ is set to $\hat\theta_{\mytext{CAL}}$  (satisfying \eqref{eq:CAL-theta-expan-prf0}).
In fact, for any $\theta$, following similar steps as in the proof of \eqref{eq:CAL-theta-expan-prf2} gives
\begin{align*}
& \quad \sum_{k=1}^K \Big[\frac{\hat{W}_{1k, \mytext{IPW}}\hat{W}_{0k, \mytext{IPW}}}{\hat{W}_{1k, \mytext{IPW}}\me^{\theta} + \hat{W}_{0k, \mytext{IPW}}}(\hat{q}_{1k, \mytext{IPW}} - \hat{q}_{0k, \mytext{IPW}}\me^{\theta})
 - \frac{\breve{W}_{1k, \mytext{IPW}}\breve{W}_{0k, \mytext{IPW}}}{\breve{W}_{1k, \mytext{IPW}}\me^{\theta} + \breve{W}_{0k, \mytext{IPW}}}(\breve{q}_{1k, \mytext{IPW}} - \breve{q}_{0k, \mytext{IPW}}\me^{\theta}) \Big] \\
& = \sum_{k=1}^K \Big[ \frac{\breve{W}_{1k, \mytext{IPW}}\breve{W}_{0k, \mytext{IPW}}}{\breve{W}_{1k, \mytext{IPW}}\me^{\theta} + \breve{W}_{0k, \mytext{IPW}}}\left(\frac{- \hat{S}_{1k, \mytext{IPW}} + (1 - \breve{q}_{1k, \mytext{IPW}})\hat{S}_{1, k-1, \mytext{IPW}} }{\breve{S}_{1, k-1, \mytext{IPW}}}\right) \\
& \quad - \frac{\breve{W}_{1k, \mytext{IPW}}\breve{W}_{0k, \mytext{IPW}} \me^{\theta}}{\breve{W}_{1k, \mytext{IPW}}\me^{\theta} + \breve{W}_{0k, \mytext{IPW}}}\left(\frac{- \hat{S}_{0k, \mytext{IPW}} + (1 - \breve{q}_{0k, \mytext{IPW}})\hat{S}_{0, k-1, \mytext{IPW}}}{\breve{S}_{0, k-1, \mytext{IPW}}}\right) \nonumber \\
& \quad + \frac{\breve{q}_{1k, \mytext{IPW}} - \breve{q}_{0k, \mytext{IPW}}\me^{\theta}}{(\breve{W}_{1k, \mytext{IPW}}\me^{\theta} + \breve{W}_{0k, \mytext{IPW}})^2}
\{ \breve{W}_{0k, \mytext{IPW}}^2 (\hat{W}_{1k, \mytext{IPW}} -\breve{W}_{1k, \mytext{IPW}} ) + \me^{\theta} \breve{W}_{1k, \mytext{IPW}}^2
 (\hat{W}_{0k, \mytext{IPW}} -\breve{W}_{0k, \mytext{IPW}} ) \} \Big] \\
& \quad + o_P(n^{-1/2}).
\end{align*}
By direct calculation,
\begin{align*}
& \quad \sum_{k=1}^K \frac{\breve{W}_{1k, \mytext{IPW}}\breve{W}_{0k, \mytext{IPW}}}{\breve{W}_{1k, \mytext{IPW}}\me^{\theta} + \breve{W}_{0k, \mytext{IPW}}}(\breve{q}_{1k, \mytext{IPW}} - \breve{q}_{0k, \mytext{IPW}}\me^{\theta}) \\
& = \sum_{k=1}^K \Big[ \frac{\breve{q}_{1k, \mytext{IPW}} - \breve{q}_{0k, \mytext{IPW}}\me^{\theta}}{(\breve{W}_{1k, \mytext{IPW}}\me^{\theta} + \breve{W}_{0k, \mytext{IPW}})^2}
\{ \breve{W}_{0k, \mytext{IPW}}^2 \breve{W}_{1k, \mytext{IPW}} + \me^{\theta} \breve{W}_{1k, \mytext{IPW}}^2 \breve{W}_{0k, \mytext{IPW}} \} \Big] .
\end{align*}
Combining the preceding two displays gives
\begin{align*}
& \quad \sum_{k=1}^K \frac{\hat{W}_{1k, \mytext{IPW}}\hat{W}_{0k, \mytext{IPW}}}{\hat{W}_{1k, \mytext{IPW}}\me^{\theta} + \hat{W}_{0k, \mytext{IPW}}}(\hat{q}_{1k, \mytext{IPW}} - \hat{q}_{0k, \mytext{IPW}}\me^{\theta}) \\
& = \sum_{k=1}^K \Big[ \frac{\breve{W}_{1k, \mytext{IPW}}\breve{W}_{0k, \mytext{IPW}}}{\breve{W}_{1k, \mytext{IPW}}\me^{\theta} + \breve{W}_{0k, \mytext{IPW}}}\frac{(1 - \breve{q}_{1k, \mytext{IPW}})\hat{S}_{1, k-1, \mytext{IPW}} - \hat{S}_{1k, \mytext{IPW}} }{\breve{S}_{1, k-1, \mytext{IPW}}} \\
& \quad - \frac{\breve{W}_{1k, \mytext{IPW}}\breve{W}_{0k, \mytext{IPW}} \me^{\theta}}{\breve{W}_{1k, \mytext{IPW}}\me^{\theta} + \breve{W}_{0k, \mytext{IPW}}}\frac{(1 - \breve{q}_{0k, \mytext{IPW}})\hat{S}_{0, k-1, \mytext{IPW}} - \hat{S}_{0k, \mytext{IPW}}}{\breve{S}_{0, k-1, \mytext{IPW}}} \nonumber \\
& \quad + \frac{\breve{q}_{1k, \mytext{IPW}} - \breve{q}_{0k, \mytext{IPW}}\me^{\theta}}{(\breve{W}_{1k, \mytext{IPW}}\me^{\theta} + \breve{W}_{0k, \mytext{IPW}})^2}
\{ \breve{W}_{0k, \mytext{IPW}}^2 \hat{W}_{1k, \mytext{IPW}}  + \me^{\theta} \breve{W}_{1k, \mytext{IPW}}^2 \hat{W}_{0k, \mytext{IPW}} \} \Big] \\
& \quad + o_P(n^{-1/2}).
\end{align*}
Substituting into the above equation
the definitions of $\hat S_{ak,\mytext{IPW}}$ and $\hat W_{ak,\mytext{IPW}}$ as in \eqref{eq:IPW-RCAL} for $a=1$ yields
the asymptotic expansion \eqref{eq:expand}.

\vspace{.1in}
\textbf{Proof of Proposition~\ref{pro:AIPW}.}\;
For conciseness, we provide the reasoning leading to the asymptotic expansion \eqref{eq:RCAL-theta-expan} only.
To make explicit the dependency on $(\hat\rho_{a,1:k},\hat\eta_{a,1:k})$ through $\hat U^{(a)}_{ki}$, we rewrite
$B_{ai}(\theta, \hat{S}_{1,\mytext{IPW}}, \hat{W}_{\mytext{IPW}})=
B_{ai}(\theta, \hat{S}_{1,\mytext{IPW}}, \hat{W}_{\mytext{IPW}}, \hat\rho_{a,1:k},\hat\eta_{a,1:k})$ and
\begin{align*}
\hat\delta_i(\theta) = & \left\{\frac{A_i}{\hat{\pi}_{1i}}B_{1i}(\theta, \hat{S}_{1,\mytext{IPW}}, \hat{W}_{\mytext{IPW}}, \hat\rho_{1,1:k},\hat\eta_{1,1:k}) - \left(\frac{A_i}{\hat{\pi}_{1i}} - 1\right)\hat{g}_{1i}\right\}  \\
& - \left\{\frac{1 - A_i}{1 - \hat{\pi}_{0i}}B_{0i}(\theta,\hat{S}_{0,\mytext{IPW}}, \hat{W}_{\mytext{IPW}}, \hat\rho_{0,1:k},\hat\eta_{0,1:k})  - \left(\frac{1 - A_i}{1 - \hat{\pi}_{0i}} - 1\right)\hat{g}_{0i}\right\}.
\end{align*}
Then as noted in the main paper,
\begin{align*}
\frac{1}{n} \sum_{i=1}^n \hat\delta_i(\theta)
 = & \sum_{k=1}^K \left\{  \frac{ \hat W_{1k,\mytext{IPW}} \hat W_{0k,\mytext{IPW}} }
 { \hat W_{1k,\mytext{IPW}} \me^\theta + \hat W_{0k,\mytext{IPW}} } ( \hat q_{1k,\mytext{IPW}} - \hat q_{0k,\mytext{IPW}} \me^\theta ) \right\}  \\
& - \frac{1}{n}\sum_{i=1}^n \left\{\left(\frac{A_i}{\hat{\pi}_{1i}} - 1\right)\hat{g}_{1i} - \left(\frac{1 - A_i}{1 - \hat{\pi}_{0i}} - 1\right)\hat{g}_{0i}\right\} .
\end{align*}
The limit version of $\hat\delta_i(\theta)$ is
\begin{align*}
\breve\delta_i(\theta) = & \left\{\frac{A_i}{\breve{\pi}_{1i}}B_{1i}(\theta, \breve{S}_{1,\mytext{IPW}}, \breve{W}_{\mytext{IPW}}, \breve\rho_{1,1:k},\breve\eta_{1,1:k}) - \left(\frac{A_i}{\breve{\pi}_{1i}} - 1\right)\breve{g}_{1i}\right\}  \\
& - \left\{\frac{1 - A_i}{1 - \breve{\pi}_{0i}}B_{0i}(\theta,\breve{S}_{0,\mytext{IPW}}, \breve{W}_{\mytext{IPW}}, \breve\rho_{0,1:k},\breve\eta_{0,1:k})  - \left(\frac{1 - A_i}{1 - \breve{\pi}_{0i}} - 1\right)\breve{g}_{0i}\right\}.
\end{align*}
Note that $B_{ai}(\theta, \breve{S}_{a,\mytext{IPW}}, \breve{W}_{\mytext{IPW}}, \breve\rho_{a,1:k},\breve\eta_{a,1:k})$
differs from $B_{ai}(\theta, \breve{S}_{a,\mytext{IPW}}, \breve{W}_{\mytext{IPW}})$ in the main paper, which corresponds to
$B_{ai}(\theta, \breve{S}_{a,\mytext{IPW}}, \breve{W}_{\mytext{IPW}}, \hat\rho_{a,1:k},\hat\eta_{a,1:k})$. Then
\begin{align}
\frac{1}{n} \sum_{i=1}^n \breve\delta_i(\theta)
 = & \sum_{k=1}^K \left\{  \frac{ \breve W_{1k,\mytext{IPW}} \breve W_{0k,\mytext{IPW}} }
 { \breve W_{1k,\mytext{IPW}} \me^\theta + \breve W_{0k,\mytext{IPW}} } ( \breve q_{1k,\mytext{IPW}} - \breve q_{0k,\mytext{IPW}} \me^\theta ) \right\} \nonumber \\
& - \frac{1}{n}\sum_{i=1}^n \left\{\left(\frac{A_i}{\breve{\pi}_{1i}} - 1\right)\breve{g}_{1i} - \left(\frac{1 - A_i}{1 - \breve{\pi}_{0i}} - 1\right)\breve{g}_{0i}\right\} , \label{eq:breve-delta}
\end{align}
where $\breve \pi_{ai} = \pi(X_i;\breve\gamma_a)$, $\breve g_{ai} = g_a (X_i;\breve\zeta_a)$.

By similar reasoning as in Ghosh and Tan(2022) as well as that underlying Proposition~\ref{pro:CAL-theta}, it can be shown that
$\hat\theta_{\mytext{AIPW}}$ satisfies the asymptotic expansion
\begin{align*}
 \hat\theta_{\mytext{AIPW}} - \breve \theta_{\mytext{IPW}} =
  H (\theta)^{-1} \cdot \frac{1}{n} \sum_{i=1}^n \breve\delta_i(\theta) \Big|_{\breve\theta_{\mytext{IPW}}} + o_p(n^{-1/2}) ,
\end{align*}
where
$ H(\theta) = (-1) E \left\{ \frac{\dif}{\dif \theta}  \frac{1}{n} \sum_{i=1}^n \breve\delta_i(\theta)  \right\}$.
Then the asymptotic expansion \eqref{eq:RCAL-theta-expan} follows because differentiating the expression~\eqref{eq:breve-delta} gives
$ H(\theta) = H_{\mytext{IPW}}(\theta)$, where
\begin{align*}
 H_{\mytext{IPW}} (\theta) = \sum_{k=1}^K \frac{\breve W_{1k,\mytext{IPW}} \breve W_{0k,\mytext{IPW}} \me^\theta}
 { ( \breve W_{1k,\mytext{IPW}} \me^\theta + \breve W_{0k,\mytext{IPW}} )^2 }
 ( \breve W_{1k,\mytext{IPW}} \breve q_{1k,\mytext{IPW}}+ \breve W_{0k,\mytext{IPW}}\breve q_{0k,\mytext{IPW}} ).
\end{align*}

\section{Additional results for simulation study} \label{sec:addtional-simulation}

\subsection{Results for simulation setup} \label{sec:simulation-setup}
We first give the formula for computing the true values of the survival probabilities $S_{ak}$ and $\breve{\theta}_{\mytext{wBP}}$, and then show that the true PS model is sparse, and  model \eqref{eq:PS-model} is correctly specified under (C1) and misspecified under (C2) with $f(X) =(1, X^\T)^\T$.

Applying the law of total expectation and no unmeasured confounding assumption, we have
\begin{align*}
S_{ak} & = P(\tilde{U}^{(a)} > u_k) = E\{P(\tilde{U}^{(a)} > u_k | X)\} = E\{P(\tilde{U} > u_k | A = a , X)\}  \\
& = P(A = 1)E_{X^{1:10} | A = 1}P(\tilde{U} > u_k | A = a , X^{1:10}) + P(A = 0)E_{X^{1:10} | A = 0}P(\tilde{U} > u_k | A = a , X^{1:10}),
\end{align*}
where $E_{X^{1:10} | A = a}$ denotes the conditional expectation of $X^{1:10} | A = a$.  We employ a Monte Carlo approach to obtain $S_{ak}$. Specifically, we generate independent samples of $X^{1:10} | A = a$ from the truncated multivariate normal distribution using the R package \texttt{tmvtnorm}. For each sampled $X^{1:10} = x$, we compute the probability $P(\tilde{U} > u_k | A = a , X^{1:10})$ based on the Weibull distribution of $\tilde{U} | A = a, X^{1:10}$. 
After generating samples of size $n = 10^6$ from $X^{1:10} | A=a$ for each $a$, we estimate the expectation $E_{X^{1:10} | A = a}P(\tilde{U} > u_k | A = a , X^{1:10})$ by taking the sample average of these probabilities and then we obtain an estimate of $P(\tilde{U}^{(a)} > u_k)$. 
To improve accuracy, we repeat the entire procedure 20 times and use the average of these 20 estimates as our final approximation of $P(\tilde{U}^{(a)} > u_k)$.

The parameter $\breve{\theta}_{\mytext{wBP}}$ is defined as a solution to \eqref{eq:weighted-BP-pop}. To compute $\breve{\theta}_{\mytext{wBP}}$, we
need to obtain values of $S_{ak}$ and $W_{ak}$ for $k = 1, \ldots, K$.
Under the non-informative censoring assumption and the independence of $\tilde{C}$ from $X$, we have
\begin{align*}
W_{ak} & = E\{ P( \tilde{Y} \ge u_k | A=a,X )\} = E[ P \{ \min(\tilde{U}, \tilde{C}) \geq u_k | A = a, X\}] \\
& = E\{ P(\tilde{U} \geq u_k | A = a, X) \times P(\tilde{C} \geq u_k | A = a, X) \} \\
& = E\{ P(\tilde{U} \geq u_k | A = a, X)\} \times P(\tilde{C} \geq u_k | A = a) = S_{ak}\times P(\tilde{C} \geq u_k | A = a),
\end{align*}
which can be used to compute $W_{ak}$ once $S_{ak}$ is computed.  We still employ a Monte Carlo approach to obtain $\breve{\theta}_{\mytext{wBP}}$. Specifically, we generate a dataset of $(Y, \Delta, A, X)$ with sample size $n = 2\times 10^6$. 
For each $k = 1, \ldots, K$, we compute an estimate of $S_{ak}$ as described above, then compute $W_{ak}$ using the relationship above, and finally solve equation \eqref{eq:weighted-BP-pop} using the R package \texttt{trust} to obtain an estimate of $\breve{\theta}_{\mytext{wBP}}$.  To improve accuracy, we repeat the entire procedure 20 times and use the average of these 20 estimates as our final approximation of $\breve{\theta}_{\mytext{wBP}}$. The value of $K$ is determined as described in Section \ref{sec:addtional-result}.

Under the data generating setting, the propensity score is
\begin{align*}
\log \frac{\pi^*(X)}{1 - \pi^*(X)} = - \frac{1}{2}X^{{1:10}^T}(\Sigma_1 - \Sigma_0)X^{1:10} + (\mu_1^T\Sigma_1 - \mu_0^T\Sigma_0)^T X^{1:10}.
\end{align*}
Then it is easily seen that model \eqref{eq:PS-model} with $f(X) = (1, X)^\T$ is correctly specified under (C1) and misspecified under (C2).
Next, we show that $\breve{\gamma}_{1, \mytext{CAL}}$ are sparse, with only the first 11 components nonzero.
In fact, consider a coefficient vector $\breve\gamma$ with the first 11 components defined as a solution to \eqref{eq:limit-gamma} below for $j=0$ with $X^0\equiv 1$ and for $j=1,\ldots,10$,
and the remaining components set to 0. Then it suffices to verify that
\begin{align}
E\left\{\left(\frac{A}{\pi(X; \breve{\gamma})} - 1\right)X^j\right\} = 0  \label{eq:limit-gamma}
\end{align}
for $j = 11, \ldots, p$. By the definition of $\breve\gamma$, \eqref{eq:limit-gamma} holds for $j = 0, 1, \ldots, 10$. For $j\geq 11$,
\begin{align*}
E\left\{\left(\frac{A}{\pi(X; \breve{\gamma} )} - 1\right)X^j \right\} & =
E\left\{\left(\frac{A}{\pi(X; \breve{\gamma} )} - 1\right) (X^j-E(X^j)) \right\} \\
& = E\left\{\left(\frac{\pi^*(X)}{\pi(X; \breve{\gamma} )} - 1\right) (X^j-E(X^j)) \right\} \\
& = E\left\{ \frac{\pi^*(X)}{\pi(X; \breve{\gamma} )} - 1 \right\} E(X^j-E(X^j)) = 0,
\end{align*}
where the first equality follows from \eqref{eq:limit-gamma} for $j=0$ (i.e., $X^0\equiv 1$),
the second equality follows from the law of iterated expectations by conditioning on $X$, and
the third equality follows from the independence of $X^j$ and $X^{1:10}$ for $j \ge 11$.

\subsection{Additional simulation results} \label{sec:addtional-result}
The number of time periods $K$ is determined by the maximum observed time $Y$ in the relevant treatment groups, provided that the observation is uncensored and at least one event has occurred. For estimating survival probabilities, $K$ is determined as the maximum of $Y$  in the treatment-1 or treatment-0 group, provided that the observation is uncensored and at least one event has occurred. For estimating the hazard ratio parameter \( \breve{\theta}_{\mytext{wBP}} \), $K$ is determined as the maximum of $Y$ satisfying this criterion in both treatment groups.

Tables \ref{tb:target-S} and \ref{tb:target-theta} present the true and limit (or target) values of estimators for survival probabilities at $u_k = 60, 90, 120$ and the hazard-ratio parameter $\breve{\theta}_{\mytext{wBP}}$.
The limit values are taken to be the Monte Carlo means of the point estimates from 20 repeated simulations each of size $10^6$ 
for the survival probabilities and $2\times 10^6$ for the hazard-ratio parameter.
From Table \ref{tb:target-S}, we observe that under (C1), the limit values of $\hat{S}_{1k, \mytext{KM}}$ differ significantly from the true values,  whereas the limit values of all other estimators closely approximate the true values. Under (C2), the limit values of all estimators deviate from the true values.
Table \ref{tb:target-theta} shows that under (C1), the limit values of all estimators are similar, except for  $\hat{\theta}_{\mytext{BP}}$. Under (C2), the limit values of $\hat{\theta}_{\mytext{CAL, lin}}$,  $\hat{\theta}_{\mytext{RCw}}$ and $\hat{\theta}_{\mytext{RCa}}$ remain the same due to the relationship
explained in Section \ref{sec:high-dimension}, and these values differ from those of  $\hat{\theta}_{\mytext{BP}}$ and  $\hat{\theta}_{\mytext{wBP}}$.

Tables \ref{tb:s1-n1000-c2} and \ref{tb:s1-n400} summarize the simulation results for estimating $S_{1k}$ with $n = 1000$ under (C2) and $n = 400$ under both (C1) and (C2), respectively.
Tables \ref{tb:s0-n400} and \ref{tb:s0-n1000} summarize the simulation results for estimating $S_{0k}$ with $n = 400$ and $n = 1000$, respectively. Tables \ref{tb:theta-n1000-c2} and \ref{tb:theta-n400} report the simulation results for the estimating hazard-ratio parameter $\breve{\theta}_{\mytext{wBP}}$. The conclusions drawn from these tables agree with the findings discussed in the main paper.

\begin{table} 
\caption{True and limit values of  estimators for survival probabilities at $u_k = 60, 90, 120$.} \label{tb:target-S} \vspace{-.15in}
\begin{center}
 \small
  \renewcommand{\arraystretch}{0.8}
\resizebox{0.9\textwidth}{!}{\begin{tabular}{lccccccc}
\hline
& \multicolumn{3}{c}{(C1)} & & \multicolumn{3}{c}{(C2)}  \\
\cline{2-4} \cline{6-8}
& $u_k = 60$ & $u_k = 90$ & $u_k = 120$ & & $u_k = 60$ & $u_k = 90$ & $u_k = 120$  \\
\hline

& \multicolumn{7}{c}{$p = 10, a = 1$} \\
$P(U^{(1)} > u_k)$  & 0.525 ( 0.000 )&  0.405 ( 0.000 )&  0.322 ( 0.000 )&  & 0.531 ( 0.000 )&  0.411 ( 0.000 )&  0.326 ( 0.000 ) \\
$\hat{S}_{1k, \mytext{KM}}$ &  0.555 ( 0.001 )&  0.437 ( 0.001 )&  0.352 ( 0.001 ) & &  0.564 ( 0.001 )&  0.442 ( 0.001 )&  0.353 ( 0.001 )\\
$\hat{S}_{1k, \mytext{wKM}}$ &  0.525 ( 0.001 )&  0.406 ( 0.001 )&  0.322 ( 0.001 ) & &  0.540 ( 0.001 )&  0.416 ( 0.001 )&  0.328 ( 0.001 )\\
$\hat{S}_{1k, \mytext{CAL}}$ &  0.525 ( 0.001 )&  0.406 ( 0.001 )&  0.322 ( 0.001 ) & &  0.534 ( 0.001 )&  0.414 ( 0.001 )&  0.329 ( 0.001 )\\
$\hat{S}_{1k, \mytext{CAL,lin}}$ &  0.525 ( 0.001 )&  0.406 ( 0.001 )&  0.322 ( 0.001 ) & &  0.534 ( 0.001 )&  0.411 ( 0.001 )&  0.323 ( 0.001 )\\

& \multicolumn{7}{c}{$p = 10, a = 0$} \\
$P(U^{(0)} > u_k)$ & 0.608 ( 0.000 )&  0.441 ( 0.000 )&  0.324 ( 0.000 )&  & 0.620 ( 0.000 )&  0.450 ( 0.000 )&  0.329 ( 0.000 ) \\
$\hat{S}_{0k, \mytext{KM}}$ &  0.567 ( 0.001 )&  0.397 ( 0.001 )&  0.283 ( 0.001 ) & &  0.570 ( 0.001 )&  0.407 ( 0.001 )&  0.296 ( 0.001 )\\
$\hat{S}_{0k, \mytext{wKM}}$ &  0.608 ( 0.001 )&  0.441 ( 0.001 )&  0.324 ( 0.001 ) & &  0.617 ( 0.000 )&  0.457 ( 0.001 )&  0.344 ( 0.001 )\\
$\hat{S}_{0k, \mytext{CAL}}$ &  0.608 ( 0.001 )&  0.441 ( 0.001 )&  0.324 ( 0.001 ) & &  0.618 ( 0.001 )&  0.448 ( 0.001 )&  0.328 ( 0.001 )\\
$\hat{S}_{0k, \mytext{CAL,lin}}$ &  0.608 ( 0.001 )&  0.441 ( 0.001 )&  0.324 ( 0.001 ) & &  0.610 ( 0.000 )&  0.449 ( 0.001 )&  0.336 ( 0.001 )\\

& \multicolumn{7}{c}{$p=200, a = 1$} \\
$P(U^{(1)} > u_k)$ & 0.525 ( 0.000 )&  0.405 ( 0.000 )&  0.322 ( 0.000 )&  & 0.531 ( 0.000 )&  0.411 ( 0.000 )&  0.326 ( 0.000 ) \\
$\hat{S}_{1k, \mytext{KM}}$ &  0.554 ( 0.001 )&  0.436 ( 0.001 )&  0.351 ( 0.001 ) & &  0.564 ( 0.001 )&  0.442 ( 0.001 )&  0.353 ( 0.001 )\\
$\hat{S}_{1k, \mytext{wKM}}$ &  0.525 ( 0.001 )&  0.406 ( 0.001 )&  0.322 ( 0.001 ) & &  0.539 ( 0.001 )&  0.416 ( 0.001 )&  0.328 ( 0.001 )\\
$\hat{S}_{1k, \mytext{RCAL}}$ &  0.525 ( 0.001 )&  0.406 ( 0.001 )&  0.322 ( 0.001 ) & &  0.533 ( 0.001 )&  0.414 ( 0.001 )&  0.329 ( 0.001 )\\
$\hat{S}_{1k, \mytext{RCAL,lin}}$ &  0.525 ( 0.001 )&  0.406 ( 0.001 )&  0.322 ( 0.001 ) & &  0.533 ( 0.001 )&  0.410 ( 0.001 )&  0.322 ( 0.001 )\\

& \multicolumn{7}{c}{$p=200, a = 0$} \\
$P(U^{(0)} > u_k)$ & 0.608 ( 0.000 )&  0.441 ( 0.000 )&  0.324 ( 0.000 )&  & 0.620 ( 0.000 )&  0.450 ( 0.000 )&  0.329 ( 0.000 ) \\
$\hat{S}_{1k, \mytext{KM}}$ &  0.567 ( 0.001 )&  0.397 ( 0.001 )&  0.284 ( 0.000 ) & &  0.570 ( 0.001 )&  0.406 ( 0.001 )&  0.296 ( 0.001 )\\
$\hat{S}_{1k, \mytext{wKM}}$ &  0.608 ( 0.001 )&  0.441 ( 0.001 )&  0.324 ( 0.000 ) & &  0.617 ( 0.000 )&  0.457 ( 0.000 )&  0.344 ( 0.000 )\\
$\hat{S}_{1k, \mytext{CAL}}$ &  0.608 ( 0.001 )&  0.441 ( 0.001 )&  0.324 ( 0.000 ) & &  0.618 ( 0.000 )&  0.448 ( 0.000 )&  0.328 ( 0.000 )\\
$\hat{S}_{1k, \mytext{CAL,lin}}$ &  0.608 ( 0.001 )&  0.441 ( 0.001 )&  0.324 ( 0.000 ) & &  0.610 ( 0.000 )&  0.449 ( 0.000 )&  0.336 ( 0.000 )\\

\hline
\end{tabular}}
\end{center}
\setlength{\baselineskip}{0.2\baselineskip}
\vspace{-.15in}\noindent{\tiny
\textbf{Note}: Limit values are taken to be the Monte Carlo means of the point estimates from 20 repeated simulations each of size $10^6$. Numbers inside the parentheses are standard deviations of the point estimates from the 20 repeated simulations.}
\vspace{-.2in}
\end{table}

\begin{table} [H]
\caption{True and limit values of estimators for $\breve{\theta}_{\mytext{wBP}}$.} \label{tb:target-theta} \vspace{-.15in}
\begin{center}
 \small
  \renewcommand{\arraystretch}{0.8}
\resizebox{0.9\textwidth}{!}{\begin{tabular}{lccccccc}
\hline
Case & $\breve{\theta}_{\mytext{wBP}}$ & $\hat{\theta}_{\mytext{BP}}$ & $\hat{\theta}_{\mytext{wBP}}$ & $\hat{\theta}_{\mytext{CAL}}$ & $\hat{\theta}_{\mytext{CAL,lin}}$ & $\hat{\theta}_{\mytext{RCw}}$ & $\hat{\theta}_{\mytext{RCa}}$ \\
\hline
& \multicolumn{7}{c}{$p=10$} \\
(C1) & 0.068 ( 0.000 ) & -0.122 ( 0.003 ) &  0.066 ( 0.002 ) &   0.066 ( 0.002 ) &  0.066 ( 0.002 ) &  0.066 ( 0.002 ) &  0.066 ( 0.002 ) \\
(C2) & 0.065 ( 0.000 ) & -0.099 ( 0.002 ) &  0.094 ( 0.002 ) &   0.053 ( 0.002 ) &  0.091 ( 0.002 ) &  0.091 ( 0.002 ) &  0.091 ( 0.002 ) \\
& \multicolumn{7}{c}{$p=200$} \\
(C1) & 0.068 ( 0.000 ) & -0.121 ( 0.003 ) &  0.067 ( 0.002 ) & -- & -- & 0.067 ( 0.002 ) &  0.067 ( 0.002 )\\
(C2) & 0.065 ( 0.000 ) & -0.098 ( 0.002 ) &  0.094 ( 0.002 ) & -- & --  & 0.091 ( 0.002 ) &  0.091 ( 0.002 ) \\
\hline
\end{tabular}}
\end{center}
\setlength{\baselineskip}{0.2\baselineskip}
\vspace{-.15in}\noindent{\tiny
\textbf{Note}: Limit values are taken to be the Monte Carlo means of the point estimates from 20 repeated simulations each of size $2\times 10^6$. Numbers inside the parentheses are standard deviations of the point estimates from the 20 repeated simulations.
}
\vspace{-.2in}
\end{table}

\begin{table} [H]
\caption{Summary for estimation of $S_{1k}$ at $u_k = 60, 90, 120$ with $n = 1000$ under (C2).} \vspace{-.15in} \label{tb:s1-n1000-c2}
\begin{center}
 \small
  \renewcommand{\arraystretch}{0.8}
\resizebox{0.9\textwidth}{!}{\begin{tabular}{lccccccccccccccccc}
\hline
& \multicolumn{5}{c}{$u_k = 60$} && \multicolumn{5}{c}{$u_k = 90$} && \multicolumn{5}{c}{$u_k = 120$}   \\
\cline{2-6}\cline{8-12}\cline{14-18}
& Bias & $\sqrt{\text{Var}}$ & $\sqrt{\text{EVar}}$ & Cov90(L90) & Cov95(L95) & & Bias & $\sqrt{\text{Var}}$ & $\sqrt{\text{EVar}}$ & Cov90(L90) & Cov95(L95) & & Bias & $\sqrt{\text{Var}}$ & $\sqrt{\text{EVar}}$ & Cov90(L90) & Cov95(L95) \\
\hline
& \multicolumn{17}{c}{$p = 10$} \\
$\hat{S}_{1k, \mytext{KM}}$ &  -0.001  &  0.023  &  0.023  &  0.895 ( 0.076 )  &  0.946 ( 0.091 ) &  &  -0.002  &  0.024  &  0.024  &  0.899 ( 0.078 )  &  0.949 ( 0.093 ) &  &  -0.003  &  0.023  &  0.023  &  0.895 ( 0.077 )  &  0.942 ( 0.092 )\\
$\hat{S}_{1k, \mytext{wKM}}$ &  -0.001  &  0.023  &  0.024  &  0.902 ( 0.078 )  &  0.955 ( 0.093 ) &  &  -0.002  &  0.023  &  0.024  &  0.897 ( 0.078 )  &  0.952 ( 0.093 ) &  &  -0.003  &  0.022  &  0.023  &  0.903 ( 0.075 )  &  0.951 ( 0.090 )\\
$\hat{S}_{1k, \mytext{CAL}}$ &  -0.001  &  0.023  &  0.023  &  0.901 ( 0.075 )  &  0.952 ( 0.090 ) &  &  -0.001  &  0.022  &  0.023  &  0.895 ( 0.075 )  &  0.951 ( 0.089 ) &  &  -0.001  &  0.021  &  0.022  &  0.901 ( 0.072 )  &  0.951 ( 0.086 )\\
$\hat{S}_{1k, \mytext{CAL,lin}}$ &  -0.001  &  0.023  &  0.023  &  0.889 ( 0.075 )  &  0.946 ( 0.090 ) &  &  -0.002  &  0.023  &  0.023  &  0.881 ( 0.075 )  &  0.942 ( 0.089 ) &  &  -0.002  &  0.022  &  0.022  &  0.895 ( 0.073 )  &  0.944 ( 0.087 )\\

& \multicolumn{17}{c}{$p = 200$} \\
$\hat{S}_{1k, \mytext{KM}}$ &  -0.001  &  0.023  &  0.023  &  0.902 ( 0.076 )  &  0.947 ( 0.091 ) &  &  -0.001  &  0.024  &  0.024  &  0.899 ( 0.078 )  &  0.949 ( 0.093 ) &  &  -0.002  &  0.024  &  0.023  &  0.900 ( 0.077 )  &  0.945 ( 0.092 )\\
$\hat{S}_{1k, \mytext{wKM}}$ &   0.019  &  0.023  &  0.023  &  0.803 ( 0.076 )  &  0.873 ( 0.091 ) &  &   0.019  &  0.024  &  0.024  &  0.798 ( 0.078 )  &  0.879 ( 0.093 ) &  &   0.018  &  0.023  &  0.023  &  0.819 ( 0.077 )  &  0.886 ( 0.092 )\\
$\hat{S}_{1k, \mytext{RCAL}}$ &   0.012  &  0.023  &  0.022  &  0.844 ( 0.074 )  &  0.908 ( 0.088 ) &  &   0.009  &  0.023  &  0.023  &  0.875 ( 0.075 )  &  0.930 ( 0.089 ) &  &   0.006  &  0.022  &  0.022  &  0.889 ( 0.074 )  &  0.947 ( 0.088 )\\
$\hat{S}_{1k, \mytext{RCAL,lin}}$ &   0.012  &  0.024  &  0.022  &  0.839 ( 0.074 )  &  0.903 ( 0.088 ) &  &   0.011  &  0.023  &  0.023  &  0.854 ( 0.075 )  &  0.923 ( 0.089 ) &  &   0.010  &  0.023  &  0.022  &  0.869 ( 0.074 )  &  0.926 ( 0.088 )\\

\hline
\end{tabular}}
\end{center}
\setlength{\baselineskip}{0.2\baselineskip}
\vspace{-.15in}\noindent{\tiny
\textbf{Note}: Bias is the Monte Carlo bias of the point estimates against the limit values. Var are the Monte Carlo variance of the point estimates. EVar is the mean of the variance estimates. Cov90(L90) and Cov95(L95) denote coverage proportion and average length of the 90\% and 95\% confidence intervals.}
\vspace{-.2in}
\end{table}

\begin{table} [H]
\caption{Summary for estimation of $S_{1k}$ at $u_k = 60, 90, 120$ with $n = 400$.} \vspace{-.15in} \label{tb:s1-n400}
\begin{center}
 \small
  \renewcommand{\arraystretch}{0.8}
\resizebox{0.9\textwidth}{!}{\begin{tabular}{lccccccccccccccccc}
\hline
& \multicolumn{5}{c}{$u_k = 60$} && \multicolumn{5}{c}{$u_k = 90$} && \multicolumn{5}{c}{$u_k = 120$}   \\
\cline{2-6}\cline{8-12}\cline{14-18}
& Bias & $\sqrt{\text{Var}}$ & $\sqrt{\text{EVar}}$ & Cov90(L90) & Cov95(L95) & & Bias & $\sqrt{\text{Var}}$ & $\sqrt{\text{EVar}}$ & Cov90(L90) & Cov95(L95) & & Bias & $\sqrt{\text{Var}}$ & $\sqrt{\text{EVar}}$ & Cov90(L90) & Cov95(L95) \\
\hline
& \multicolumn{17}{c}{(C1), $p = 10$} \\
$\hat{S}_{1k, \mytext{KM}}$ &  0.027  &  0.036  &  0.037  &  0.817 ( 0.120 )  &  0.896 ( 0.143 ) &  &   0.027  &  0.036  &  0.037  &  0.824 ( 0.123 )  &  0.903 ( 0.147 ) &  &   0.025  &  0.036  &  0.037  &  0.834 ( 0.122 )  &  0.897 ( 0.145 )\\
$\hat{S}_{1k, \mytext{wKM}}$ &  -0.003  &  0.035  &  0.038  &  0.922 ( 0.124 )  &  0.961 ( 0.148 ) &  &  -0.004  &  0.035  &  0.037  &  0.917 ( 0.123 )  &  0.958 ( 0.147 ) &  &  -0.005  &  0.034  &  0.036  &  0.908 ( 0.119 )  &  0.955 ( 0.142 )\\
$\hat{S}_{1k, \mytext{CAL}}$ & -0.003  &  0.035  &  0.035  &  0.902 ( 0.116 )  &  0.950 ( 0.138 ) &  &  -0.004  &  0.035  &  0.035  &  0.899 ( 0.115 )  &  0.947 ( 0.137 ) &  &  -0.005  &  0.034  &  0.034  &  0.886 ( 0.111 )  &  0.943 ( 0.133 )\\
$\hat{S}_{1k, \mytext{CAL,lin}}$ &  -0.003  &  0.035  &  0.035  &  0.898 ( 0.116 )  &  0.947 ( 0.138 ) &  &  -0.004  &  0.035  &  0.035  &  0.896 ( 0.115 )  &  0.946 ( 0.138 ) &  &  -0.005  &  0.034  &  0.034  &  0.891 ( 0.112 )  &  0.944 ( 0.133 )\\

& \multicolumn{17}{c}{(C2), $p = 10$} \\
$\hat{S}_{1k, \mytext{KM}}$ &  -0.005  &  0.037  &  0.037  &  0.899 ( 0.120 )  &  0.947 ( 0.143 ) &  &  -0.006  &  0.037  &  0.038  &  0.898 ( 0.123 )  &  0.951 ( 0.147 ) &  &  -0.006  &  0.037  &  0.037  &  0.900 ( 0.122 )  &  0.946 ( 0.145 )\\
$\hat{S}_{1k, \mytext{wKM}}$ &  -0.005  &  0.037  &  0.038  &  0.907 ( 0.124 )  &  0.957 ( 0.147 ) &  &  -0.006  &  0.037  &  0.038  &  0.903 ( 0.124 )  &  0.951 ( 0.147 ) &  &  -0.005  &  0.035  &  0.036  &  0.904 ( 0.120 )  &  0.950 ( 0.143 )\\
$\hat{S}_{1k, \mytext{CAL}}$ &  -0.005  &  0.035  &  0.036  &  0.910 ( 0.119 )  &  0.955 ( 0.142 ) &  &  -0.003  &  0.035  &  0.036  &  0.908 ( 0.118 )  &  0.948 ( 0.141 ) &  &  -0.002  &  0.034  &  0.035  &  0.902 ( 0.114 )  &  0.954 ( 0.136 )\\
$\hat{S}_{1k, \mytext{CAL,lin}}$ &  -0.005  &  0.037  &  0.036  &  0.898 ( 0.120 )  &  0.948 ( 0.143 ) &  &  -0.006  &  0.037  &  0.036  &  0.888 ( 0.119 )  &  0.938 ( 0.142 ) &  &  -0.005  &  0.035  &  0.035  &  0.887 ( 0.115 )  &  0.939 ( 0.137 )\\

& \multicolumn{17}{c}{(C1), $p = 200$} \\
$\hat{S}_{1k, \mytext{KM}}$ & 0.026  &  0.037  &  0.037  &  0.814 ( 0.120 )  &  0.881 ( 0.143 ) &  &  0.026  &  0.038  &  0.037  &  0.822 ( 0.123 )  &  0.903 ( 0.146 ) &  &  0.024  &  0.037  &  0.037  &  0.829 ( 0.121 )  &  0.908 ( 0.145 )\\
$\hat{S}_{1k, \mytext{wKM}}$ &   0.024  &  0.037  &  0.037  &  0.834 ( 0.121 )  &  0.894 ( 0.144 ) &  &  0.023  &  0.037  &  0.037  &  0.835 ( 0.123 )  &  0.909 ( 0.147 ) &  &  0.022  &  0.037  &  0.037  &  0.846 ( 0.121 )  &  0.916 ( 0.145 )\\
$\hat{S}_{1k, \mytext{RCAL}}$ &   0.016  &  0.037  &  0.035  &  0.852 ( 0.115 )  &  0.906 ( 0.137 ) &  &  0.014  &  0.037  &  0.035  &  0.857 ( 0.116 )  &  0.924 ( 0.139 ) &  &  0.013  &  0.036  &  0.035  &  0.874 ( 0.115 )  &  0.927 ( 0.137 )\\
$\hat{S}_{1k, \mytext{RCAL,lin}}$ &   0.015  &  0.037  &  0.035  &  0.857 ( 0.114 )  &  0.907 ( 0.136 ) &  &  0.014  &  0.037  &  0.035  &  0.861 ( 0.116 )  &  0.926 ( 0.139 ) &  &  0.013  &  0.036  &  0.035  &  0.875 ( 0.115 )  &  0.928 ( 0.137 )\\

& \multicolumn{17}{c}{(C2), $p = 200$} \\
$\hat{S}_{1k, \mytext{KM}}$ &  -0.002  &  0.036  &  0.037  &  0.904 ( 0.120 )  &  0.953 ( 0.143 ) &  &  -0.004  &  0.038  &  0.038  &  0.900 ( 0.123 )  &  0.953 ( 0.147 ) &  &  -0.006  &  0.036  &  0.037  &  0.895 ( 0.122 )  &  0.952 ( 0.145 )\\
$\hat{S}_{1k, \mytext{wKM}}$ &   0.020  &  0.036  &  0.037  &  0.853 ( 0.121 )  &  0.916 ( 0.144 ) &  &   0.020  &  0.037  &  0.038  &  0.854 ( 0.124 )  &  0.921 ( 0.147 ) &  &   0.017  &  0.036  &  0.037  &  0.879 ( 0.122 )  &  0.933 ( 0.145 )\\
$\hat{S}_{1k, \mytext{RCAL}}$ &   0.020  &  0.036  &  0.035  &  0.828 ( 0.116 )  &  0.906 ( 0.138 ) &  &   0.016  &  0.037  &  0.036  &  0.861 ( 0.119 )  &  0.922 ( 0.141 ) &  &   0.010  &  0.036  &  0.036  &  0.892 ( 0.117 )  &  0.945 ( 0.140 )\\
$\hat{S}_{1k, \mytext{RCAL,lin}}$ &   0.020  &  0.036  &  0.035  &  0.829 ( 0.116 )  &  0.904 ( 0.138 ) &  &   0.019  &  0.037  &  0.036  &  0.842 ( 0.118 )  &  0.913 ( 0.141 ) &  &   0.016  &  0.036  &  0.036  &  0.874 ( 0.117 )  &  0.927 ( 0.139 )\\

\hline
\end{tabular}}
\end{center}
\setlength{\baselineskip}{0.2\baselineskip}
\vspace{-.15in}\noindent{\tiny
\textbf{Note}: Bias is the Monte Carlo bias of the point estimates against the true values under (C1), or against the limit values under (C2). Var are the Monte Carlo variance of the point estimates. EVar is the mean of the variance estimates. Cov90(L90) and Cov95(L95) denote coverage proportion and average length of the 90\% and 95\% confidence intervals.}
\vspace{-.2in}
\end{table}

\begin{table} [H]
\caption{Summary for estimation of $S_{0k}$ with $n = 400$.} \vspace{-.15in} \label{tb:s0-n400}
\begin{center}
 \small
  \renewcommand{\arraystretch}{0.8}
\resizebox{0.9\textwidth}{!}{\begin{tabular}{lccccccccccccccccc}
\hline
& \multicolumn{5}{c}{$u_k = 60$} && \multicolumn{5}{c}{$u_k = 90$} && \multicolumn{5}{c}{$u_k = 120$}   \\
\cline{2-6}\cline{8-12}\cline{14-18}
& Bias & $\sqrt{\text{Var}}$ & $\sqrt{\text{EVar}}$ & Cov90(L90) & Cov95(L95) & & Bias & $\sqrt{\text{Var}}$ & $\sqrt{\text{EVar}}$ & Cov90(L90) & Cov95(L95) & & Bias & $\sqrt{\text{Var}}$ & $\sqrt{\text{EVar}}$ & Cov90(L90) & Cov95(L95) \\
\hline
& \multicolumn{17}{c}{(C1), $p = 10$} \\
$\hat{S}_{0k, \mytext{KM}}$ &  -0.042  &  0.036  &  0.036  &  0.673 ( 0.117 )  &  0.782 ( 0.139 ) &  &  -0.047  &  0.036  &  0.036  &  0.638 ( 0.117 )  &  0.733 ( 0.140 ) &  &  -0.045  &  0.035  &  0.034  &  0.610 ( 0.111 )  &  0.709 ( 0.132 )\\
$\hat{S}_{0k, \mytext{wKM}}$ &  -0.001  &  0.032  &  0.035  &  0.926 ( 0.116 )  &  0.970 ( 0.138 ) &  &  -0.003  &  0.034  &  0.038  &  0.917 ( 0.124 )  &  0.967 ( 0.147 ) &  &  -0.005  &  0.035  &  0.037  &  0.908 ( 0.122 )  &  0.949 ( 0.146 )\\
$\hat{S}_{0k, \mytext{CAL}}$ &  -0.001  &  0.032  &  0.031  &  0.896 ( 0.102 )  &  0.944 ( 0.122 ) &  &  -0.004  &  0.034  &  0.033  &  0.875 ( 0.107 )  &  0.934 ( 0.128 ) &  &  -0.005  &  0.034  &  0.032  &  0.866 ( 0.105 )  &  0.922 ( 0.126 )\\
$\hat{S}_{0k, \mytext{CAL,lin}}$ &   -0.001  &  0.033  &  0.032  &  0.892 ( 0.104 )  &  0.946 ( 0.124 ) &  &  -0.003  &  0.034  &  0.033  &  0.875 ( 0.109 )  &  0.934 ( 0.130 ) &  &  -0.005  &  0.034  &  0.033  &  0.874 ( 0.107 )  &  0.924 ( 0.128 )\\

& \multicolumn{17}{c}{(C2), $p = 10$} \\
$\hat{S}_{0k, \mytext{KM}}$ &  -0.002  &  0.035  &  0.035  &  0.905 ( 0.117 )  &  0.952 ( 0.139 ) &  &  -0.003  &  0.035  &  0.036  &  0.908 ( 0.118 )  &  0.954 ( 0.140 ) &  &  -0.004  &  0.033  &  0.034  &  0.898 ( 0.112 )  &  0.947 ( 0.133 )\\
$\hat{S}_{0k, \mytext{wKM}}$ &  -0.002  &  0.030  &  0.035  &  0.941 ( 0.116 )  &  0.980 ( 0.138 ) &  &  -0.003  &  0.032  &  0.038  &  0.942 ( 0.125 )  &  0.975 ( 0.149 ) &  &  -0.005  &  0.033  &  0.038  &  0.943 ( 0.125 )  &  0.973 ( 0.149 )\\
$\hat{S}_{0k, \mytext{CAL}}$ &  -0.001  &  0.032  &  0.031  &  0.881 ( 0.100 )  &  0.936 ( 0.119 ) &  &  -0.004  &  0.034  &  0.032  &  0.879 ( 0.105 ) &  0.930 ( 0.125 ) &  &  -0.007  &  0.034  &  0.032  &  0.865 ( 0.104 )  &  0.925 ( 0.124 )\\
$\hat{S}_{0k, \mytext{CAL,lin}}$ &  -0.002  &  0.031  &  0.031  &  0.899 ( 0.101 )  &  0.951 ( 0.121 ) &  &  -0.003  &  0.032  &  0.032  &  0.899 ( 0.106 )  &  0.945 ( 0.126 ) &  &  -0.005  &  0.032  &  0.032  &  0.892 ( 0.105 )  &  0.949 ( 0.125 )\\

& \multicolumn{17}{c}{(C1), $p = 200$} \\
$\hat{S}_{0k, \mytext{KM}}$ &  -0.043  &  0.036  &  0.036  &  0.675 ( 0.117 )  &  0.780 ( 0.139 ) &  &  -0.046  &  0.037  &  0.036  &  0.633 ( 0.117 )  &  0.733 ( 0.140 ) &  &  -0.043  &  0.034  &  0.034  &  0.628 ( 0.111 )  &  0.724 ( 0.132 )\\
$\hat{S}_{0k, \mytext{wKM}}$ &  -0.040  &  0.035  &  0.036  &  0.712 ( 0.117 )  &  0.811 ( 0.139 ) &  &  -0.042  &  0.036  &  0.036  &  0.672 ( 0.118 )  &  0.772 ( 0.141 ) &  &  -0.040  &  0.034  &  0.034  &  0.663 ( 0.112 )  &  0.757 ( 0.133 )\\
$\hat{S}_{0k, \mytext{RCAL}}$ &  -0.016  &  0.033  &  0.031  &  0.842 ( 0.102 )  &  0.900 ( 0.121 ) &  &  -0.019  &  0.035  &  0.031  &  0.792 ( 0.102 )  &  0.870 ( 0.122 ) &  &  -0.021  &  0.034  &  0.030  &  0.763 ( 0.097 )  &  0.838 ( 0.116 )\\
$\hat{S}_{0k, \mytext{RCAL,lin}}$ &  -0.016  &  0.033  &  0.031  &  0.845 ( 0.103 )  &  0.906 ( 0.123 ) &  &  -0.019  &  0.035  &  0.032  &  0.797 ( 0.104 )  &  0.873 ( 0.124 ) &  &  -0.020  &  0.034  &  0.030  &  0.771 ( 0.099 )  &  0.850 ( 0.118 )\\

& \multicolumn{17}{c}{(C2), $p = 200$} \\
$\hat{S}_{0k, \mytext{KM}}$ &  -0.002  &  0.036  &  0.035  &  0.900 ( 0.117 )  &  0.953 ( 0.139 ) &  &  -0.003  &  0.035  &  0.036  &  0.899 ( 0.118 )  &  0.952 ( 0.140 ) &  &  -0.004  &  0.034  &  0.034  &  0.887 ( 0.112 )  &  0.941 ( 0.133 )\\
$\hat{S}_{0k, \mytext{wKM}}$ &  -0.045  &  0.035  &  0.035  &  0.646 ( 0.117 )  &  0.770 ( 0.139 ) &  &  -0.049  &  0.035  &  0.036  &  0.600 ( 0.118 )  &  0.723 ( 0.141 ) &  &  -0.048  &  0.034  &  0.034  &  0.585 ( 0.113 )  &  0.696 ( 0.135 )\\
$\hat{S}_{0k, \mytext{RCAL}}$ &  -0.018  &  0.033  &  0.031  &  0.819 ( 0.100 )  &  0.887 ( 0.120 ) &  &  -0.018  &  0.034  &  0.031  &  0.810 ( 0.101 )  &  0.886 ( 0.121 ) &  &  -0.018  &  0.034  &  0.030  &  0.787 ( 0.097 )  &  0.864 ( 0.116 )\\
$\hat{S}_{0k, \mytext{RCAL,lin}}$ &  -0.015  &  0.032  &  0.031  &  0.854 ( 0.102 )  &  0.915 ( 0.122 ) &  &  -0.018  &  0.033  &  0.031  &  0.827 ( 0.103 )  &  0.895 ( 0.123 ) &  &  -0.020  &  0.034  &  0.030  &  0.781 ( 0.099 )  &  0.859 ( 0.118 )\\

\hline
\end{tabular}}
\end{center}
\setlength{\baselineskip}{0.2\baselineskip}
\vspace{-.15in}\noindent{\tiny
\textbf{Note}: Bias is the Monte Carlo bias of the point estimates against the true values under (C1), or against the limit values under (C2). Var are the Monte Carlo variance of the point estimates. EVar is the mean of the variance estimates. Cov90(L90) and Cov95(L95) denote coverage proportion and average length of the 90\% and 95\% confidence intervals.}
\vspace{-.2in}
\end{table}

\begin{table} [H]
\caption{Summary for estimation of $S_{0k}$ with $n = 1000$.} \vspace{-.15in} \label{tb:s0-n1000}
\begin{center}
 \small
  \renewcommand{\arraystretch}{0.8}
\resizebox{0.9\textwidth}{!}{\begin{tabular}{lccccccccccccccccc}
\hline
& \multicolumn{5}{c}{$u_k = 60$} && \multicolumn{5}{c}{$u_k = 90$} && \multicolumn{5}{c}{$u_k = 120$}   \\
\cline{2-6}\cline{8-12}\cline{14-18}
& Bias & $\sqrt{\text{Var}}$ & $\sqrt{\text{EVar}}$ & Cov90(L90) & Cov95(L95) & & Bias & $\sqrt{\text{Var}}$ & $\sqrt{\text{EVar}}$ & Cov90(L90) & Cov95(L95) & & Bias & $\sqrt{\text{Var}}$ & $\sqrt{\text{EVar}}$ & Cov90(L90) & Cov95(L95) \\
\hline
& \multicolumn{17}{c}{(C1), $p = 10$} \\
$\hat{S}_{0k, \mytext{KM}}$ &   -0.041  &  0.023  &  0.022  &  0.414 ( 0.074 )  &  0.549 ( 0.088 ) &  &  -0.045  &  0.023  &  0.023  &  0.375 ( 0.074 )  &  0.490 ( 0.089 ) &  &  -0.042  &  0.021  &  0.021  &  0.382 ( 0.070 )  &  0.494 ( 0.084 )\\
$\hat{S}_{0k, \mytext{wKM}}$ &   0.000  &  0.020  &  0.022  &  0.931 ( 0.073 )  &  0.972 ( 0.087 ) &  &  -0.001  &  0.021  &  0.024  &  0.932 ( 0.078 )  &  0.972 ( 0.092 ) &  &  -0.001  &  0.021  &  0.023  &  0.926 ( 0.077 )  &  0.969 ( 0.091 )\\
$\hat{S}_{0k, \mytext{CAL}}$ &   0.000  &  0.020  &  0.020  &  0.896 ( 0.065 )  &  0.951 ( 0.078 ) &  &  -0.001  &  0.021  &  0.021  &  0.897 ( 0.068 )  &  0.947 ( 0.081 ) &  &  -0.001  &  0.021  &  0.020  &  0.892 ( 0.067 )  &  0.943 ( 0.080 )\\
$\hat{S}_{0k, \mytext{CAL,lin}}$ &   0.000  &  0.020  &  0.020  &  0.897 ( 0.066 )  &  0.950 ( 0.079 ) &  &  -0.001  &  0.021  &  0.021  &  0.897 ( 0.069 )  &  0.948 ( 0.082 ) &  &  -0.001  &  0.021  &  0.021  &  0.893 ( 0.068 )  &  0.942 ( 0.081 )\\

& \multicolumn{17}{c}{(C2), $p = 10$} \\
$\hat{S}_{0k, \mytext{KM}}$ &   0.000  &  0.023  &  0.022  &  0.885 ( 0.074 )  &  0.944 ( 0.088 ) &  &   0.000  &  0.023  &  0.023  &  0.892 ( 0.075 )  &  0.944 ( 0.089 ) &  &  -0.001  &  0.022  &  0.022  &  0.894 ( 0.071 )  &  0.945 ( 0.085 )\\
$\hat{S}_{0k, \mytext{wKM}}$ &   0.000  &  0.020  &  0.022  &  0.940 ( 0.072 )  &  0.974 ( 0.086 ) &  &   0.000  &  0.021  &  0.024  &  0.939 ( 0.078 )  &  0.971 ( 0.093 ) &  &  -0.001  &  0.021  &  0.024  &  0.935 ( 0.078 )  &  0.972 ( 0.093 )\\
$\hat{S}_{0k, \mytext{CAL}}$ &   0.000  &  0.021  &  0.019  &  0.874 ( 0.064 )  &  0.938 ( 0.076 ) &  &  -0.001  &  0.022  &  0.020  &  0.871 ( 0.067 )  &  0.929 ( 0.080 ) &  &  -0.002  &  0.021  &  0.020  &  0.875 ( 0.066 )  &  0.931 ( 0.079 )\\
$\hat{S}_{0k, \mytext{CAL,lin}}$ &   0.000  &  0.020  &  0.020  &  0.893 ( 0.064 )  &  0.951 ( 0.077 ) &  &   0.000  &  0.021  &  0.020  &  0.892 ( 0.067 )  &  0.947 ( 0.080 ) &  &  -0.001  &  0.021  &  0.020  &  0.894 ( 0.067 )  &  0.944 ( 0.080 )\\

& \multicolumn{17}{c}{(C1), $p = 200$} \\
$\hat{S}_{0k, \mytext{KM}}$ &   -0.041  &  0.023  &  0.022  &  0.437 ( 0.074 )  &  0.556 ( 0.088 ) &  &  -0.043  &  0.023  &  0.023  &  0.405 ( 0.074 )  &  0.528 ( 0.089 ) &  &  -0.041  &  0.021  &  0.021  &  0.402 ( 0.070 )  &  0.519 ( 0.084 )\\
$\hat{S}_{0k, \mytext{wKM}}$ &  -0.032  &  0.022  &  0.022  &  0.587 ( 0.074 )  &  0.693 ( 0.088 ) &  &  -0.035  &  0.022  &  0.023  &  0.545 ( 0.075 )  &  0.671 ( 0.089 ) &  &  -0.033  &  0.021  &  0.022  &  0.542 ( 0.072 )  &  0.666 ( 0.085 )\\
$\hat{S}_{0k, \mytext{RCAL}}$ &  -0.007  &  0.021  &  0.020  &  0.857 ( 0.065 )  &  0.927 ( 0.077 ) &  &  -0.008  &  0.021  &  0.020  &  0.845 ( 0.066 )  &  0.909 ( 0.078 ) &  &  -0.009  &  0.021  &  0.019  &  0.818 ( 0.063 )  &  0.890 ( 0.075 )\\
$\hat{S}_{0k, \mytext{RCAL,lin}}$ &  -0.008  &  0.021  &  0.020  &  0.853 ( 0.066 )  &  0.918 ( 0.079 ) &  &  -0.009  &  0.022  &  0.020  &  0.837 ( 0.067 )  &  0.903 ( 0.080 ) &  &  -0.011  &  0.021  &  0.020  &  0.812 ( 0.064 )  &  0.887 ( 0.077 )\\

& \multicolumn{17}{c}{(C2), $p = 200$} \\
$\hat{S}_{0k, \mytext{KM}}$ &   0.000  &  0.023  &  0.022  &  0.898 ( 0.074 )  &  0.944 ( 0.088 ) &  &   0.001  &  0.023  &  0.023  &  0.896 ( 0.075 )  &  0.944 ( 0.089 ) &  &   0.000  &  0.021  &  0.022  &  0.901 ( 0.071 )  &  0.947 ( 0.085 )\\
$\hat{S}_{0k, \mytext{wKM}}$ &  -0.037  &  0.021  &  0.022  &  0.513 ( 0.074 )  &  0.632 ( 0.088 ) &  &  -0.039  &  0.022  &  0.023  &  0.466 ( 0.075 )  &  0.596 ( 0.090 ) &  &  -0.038  &  0.021  &  0.022  &  0.471 ( 0.073 )  &  0.601 ( 0.086 )\\
$\hat{S}_{0k, \mytext{RCAL}}$ &  -0.009  &  0.021  &  0.019  &  0.848 ( 0.064 )  &  0.913 ( 0.076 ) &  &  -0.007  &  0.022  &  0.020  &  0.848 ( 0.065 )  &  0.918 ( 0.077 ) &  &  -0.006  &  0.021  &  0.019  &  0.844 ( 0.063 )  &  0.906 ( 0.075 )\\
$\hat{S}_{0k, \mytext{RCAL,lin}}$ &  -0.007  &  0.020  &  0.020  &  0.876 ( 0.065 )  &  0.934 ( 0.077 ) &  &  -0.007  &  0.021  &  0.020  &  0.860 ( 0.066 )  &  0.923 ( 0.078 ) &  &  -0.009  &  0.021  &  0.019  &  0.843 ( 0.064 )  &  0.897 ( 0.076 )\\

\hline
\end{tabular}}
\end{center}
\setlength{\baselineskip}{0.2\baselineskip}
\vspace{-.15in}\noindent{\tiny
\textbf{Note}: Bias is the Monte Carlo bias of the point estimates against the true values under (C1), or against the limit values under (C2). Var are the Monte Carlo variance of the point estimates. EVar is the mean of the variance estimates. Cov90(L90) and Cov95(L95) denote coverage proportion and average length of the 90\% and 95\% confidence intervals.}
\vspace{-.2in}
\end{table}

\begin{table}[H]
\caption{Summary for estimation of $\breve{\theta}_{\mytext{wBP}}$ with $n = 1000$ under (C2).} \vspace{.15in} \label{tb:theta-n1000-c2}
  \vspace{-.3in}
  \begin{center}
   \small
  \renewcommand{\arraystretch}{0.8}
  \resizebox{0.9\textwidth}{!}{
  \begin{tabular}{lccccccccccc}
  \hline
  & \multicolumn{6}{c}{$p = 10$} &&  \multicolumn{4}{c}{$p = 200$} \\
  \cline{2-7}\cline{9-12}
  & $\hat{\theta}_{\mytext{BP}}$
  & $\hat{\theta}_{\mytext{wBP}}$
  & $\hat{\theta}_{\mytext{CAL}}$
  & $\hat{\theta}_{\mytext{CAL,lin}}$
  & $\hat{\theta}_{\mytext{RCw}}$
  & $\hat{\theta}_{\mytext{RCa}}$ & 
  & $\hat{\theta}_{\mytext{BP}}$
  & $\hat{\theta}_{\mytext{wBP}}$
  & $\hat{\theta}_{\mytext{RCw}}$
  & $\hat{\theta}_{\mytext{RCa}}$ \\
  \hline
  Bias
  &  0.001 & 0.001 &  -0.007 & 0.001 & 0.001 & 0.001 & 
  & 0.004 & -0.149 & -0.153 & -0.039 \\
  $\sqrt{\text{Var}}$
  & 0.074 & 0.062 & 0.060 & 0.061 & 0.061 & 0.061 & 
  & 0.074 & 0.069 & 0.068 & 0.062 \\
  $\sqrt{\text{EVar}}$
  & 0.074 & 0.076 & 0.060 & 0.060 & 0.075 & 0.060 & 
  & 0.074 & 0.074 & 0.074 & 0.060 \\
  Cov90(L90)
  & 0.895 (0.242) & 0.957 (0.248) & 0.896 (0.199) & 0.894 (0.198) & 0.958 (0.246) & 0.894 (0.198) & 
  & 0.900 (0.242) & 0.350 (0.243) & 0.328 (0.242) & 0.815 (0.197) \\
  Cov95(L95)
  & 0.949 (0.288) & 0.981 (0.296) & 0.950 (0.237) & 0.946 (0.236) & 0.983 (0.293) & 0.946 (0.236) & 
  & 0.946 (0.288) & 0.480 (0.289) & 0.456 (0.289) & 0.894 (0.234) \\
  \hline
  \end{tabular}}
  \end{center}
\setlength{\baselineskip}{0.2\baselineskip}
\vspace{-.15in}\noindent{\tiny
\textbf{Note}:  Bias is the Monte Carlo bias of the point estimates against the limit values. Var is the Monte Carlo variance of the point estimates. EVar is the mean of the variance estimates. Cov90(L90) and Cov95(L95) denote coverage proportion and average length of the 90\% and 95\% confidence intervals.}
  \vspace{-.2in}
  \end{table}

\begin{table} [H]
\caption{Summary for estimation of $\breve{\theta}_{\mytext{wBP}}$ with $n = 400$.}  \vspace{-.15in} \label{tb:theta-n400}
\begin{center}
 \small
  \renewcommand{\arraystretch}{0.8}
\resizebox{0.9\textwidth}{!}{\begin{tabular}{lccccccccccc}
\hline
& Bias & $\sqrt{\text{Var}}$ & $\sqrt{\text{EVar}}$ & Cov90(L90) & Cov95(L95) & & Bias & $\sqrt{\text{Var}}$ & $\sqrt{\text{EVar}}$ & Cov90(L90) & Cov95(L95) \\
\hline
& \multicolumn{11}{c}{$p = 10$} \\
& \multicolumn{5}{c}{(C1)} && \multicolumn{5}{c}{(C2)} \\
\cline{2-6}\cline{8-12}
$\hat{\theta}_{\mytext{BP}}$ &  -0.191  &  0.117  &  0.116  &  0.499 ( 0.383 )  &  0.624 ( 0.456 ) &  & 0.005  &  0.112  &  0.116  &  0.913 ( 0.383 )  &  0.965 ( 0.456 )\\
$\hat{\theta}_{\mytext{wBP}}$ &  -0.002  &  0.096  &  0.120  &  0.964 ( 0.393 )  &  0.985 ( 0.469 ) &  &   0.004  &  0.093  &  0.121  &  0.966 ( 0.397 )  &  0.989 ( 0.473 )\\
$\hat{\theta}_{\mytext{CAL}}$ &   0.001  &  0.095  &  0.092  &  0.899 ( 0.304 )  &  0.944 ( 0.362 ) & &  -0.015  &  0.090  &  0.095  &  0.914 ( 0.312 )  &  0.959 ( 0.372 )\\
$\hat{\theta}_{\mytext{CAL,lin}}$  &  -0.002  &  0.096  &  0.094  &  0.897 ( 0.311 )  &  0.952 ( 0.370 ) & &   0.003  &  0.092  &  0.095  &  0.913 ( 0.314 )  &  0.959 ( 0.374 )\\
$\hat{\theta}_{\mytext{RCw}}$ &  -0.002  &  0.096  &  0.120  &  0.963 ( 0.393 )  &  0.984 ( 0.469 ) & &   0.003  &  0.092  &  0.120  &  0.970 ( 0.394 )  &  0.989 ( 0.469 ) \\
$\hat{\theta}_{\mytext{RCa}}$ &  -0.002  &  0.096  &  0.094  &  0.897 ( 0.311 )  &  0.952 ( 0.370 ) & &   0.003  &  0.092  &  0.095  &  0.913 ( 0.314 )  &  0.959 ( 0.374 ) \\
& \multicolumn{11}{c}{$p = 200$} \\
& \multicolumn{5}{c}{(C1)} && \multicolumn{5}{c}{(C2)} \\
\cline{2-6}\cline{8-12}
$\hat{\theta}_{\mytext{BP}}$ &  -0.185  &  0.118  &  0.116  &  0.512 ( 0.383 )  &  0.643 ( 0.456 ) & &   0.000  &  0.115  &  0.116  &  0.905 ( 0.383 )  &  0.953 ( 0.456 ) \\
$\hat{\theta}_{\mytext{wBP}}$ &  -0.170  &  0.114  &  0.117  &  0.558 ( 0.384 )  &  0.701 ( 0.457 ) & &  -0.176  &  0.113  &  0.117  &  0.549 ( 0.384 )  &  0.672 ( 0.458 ) \\
$\hat{\theta}_{\mytext{RCw}}$  &  -0.180  &  0.117  &  0.116  &  0.526 ( 0.383 )  &  0.659 ( 0.456 ) & &  -0.185  &  0.114  &  0.116  &  0.511 ( 0.383 )  &  0.637 ( 0.456 ) \\
$\hat{\theta}_{\mytext{RCa}}$ &  -0.075  &  0.102  &  0.094  &  0.779 ( 0.310 )  &  0.864 ( 0.370 ) & &  -0.084  &  0.103  &  0.096  &  0.750 ( 0.316 )  &  0.840 ( 0.376 ) \\
\hline
\end{tabular}}
\end{center}
\setlength{\baselineskip}{0.2\baselineskip}
\vspace{-.15in}\noindent{\tiny
\textbf{Note}:   Bias is the Monte Carlo bias of the point estimates against the true value under (C1), or against the limit values under (C2). Var is the Monte Carlo variance of the point estimates. EVar is the mean of the variance estimates. Cov90(L90) and Cov95(L95) denote coverage proportion and average length of the 90\% and 95\% confidence intervals.}
\vspace{-.2in}
\end{table}

\section{Additional results for empirical application} \label{sec:addtional-empirical}

\subsection{Effects of covariate adjustment}  \label{sec:coveffect}
Figure \ref{fig:xeffect4svprob} presents the standardized absolute difference between $\hat{S}_{ak, \mytext{KM}}$ and $\hat{S}_{ak, \mytext{wKM}}$ under low-dimensional setting. For two estimators $\hat{\mu}_1$ and $\hat{\mu}_2$, the standardized absolute difference is defined as $|\hat{\mu}_1 - \hat{\mu}_2| / \text{SE}(\hat{\mu}_1)$ or $|\hat{\mu}_1 - \hat{\mu}_2| / \text{SE}(\hat{\mu}_2)$, being standardized by 
$\text{SE}(\hat{\mu}_1)$ or $\text{SE}(\hat{\mu}_2)$ respectively, where $\text{SE}(\cdot)$ denotes the standard error of an estimator.
For convenience, a standardized absolute difference greater than $1.5$ is considered non-negligible.

For the Hospitalization and Mortality outcomes, $\hat{S}_{1k, \mytext{wKM}}$ differs noticeably from $\hat{S}_{1k, \mytext{KM}}$ relative to their SEs, while $\hat{S}_{0k, \mytext{wKM}}$ is similar to $\hat{S}_{0k, \mytext{KM}}$ relative to their SEs, indicating that covariate adjustment plays a non-negligible role for treatment 1 but is less influential for treatment 0. This asymmetry helps explain the corresponding $\breve{\theta}_{\mytext{wBP}}$ estimates: the larger effect of weighting for treatment 1 leads to substantial differences between $\hat{\theta}_{\mytext{wBP}}$ and $\hat{\theta}_{\mytext{BP}}$ for these two outcomes, as shown in Table \ref{tb:empirical-theta}. For example, for the Mortality outcome, the absolute difference equals $0.495 - 0.195 = 0.300$, which is approximately $3.36$ times $\text{SE}(\hat{\theta}_{\mytext{BP}})$ and $2.75$ times $\text{SE}(\hat{\theta}_{\mytext{wBP}})$.

In contrast, for the ED visits outcome, $\hat{S}_{ak, \mytext{KM}}$ and $\hat{S}_{ak, \mytext{wKM}}$ differ noticeably relative to their SEs for both $a=1$ and $a=0$. Because the differences are in the same direction for both treatment groups, their effects partially cancel out when aggregated into the estimation of $\breve{\theta}_{\mytext{wBP}}$, which explains why $\hat{\theta}_{\mytext{wBP}}$ remains similar to $\hat{\theta}_{\mytext{BP}}$ for this outcome, as seen in Table \ref{tb:empirical-theta}. The absolute difference equals $0.219 - 0.192 = 0.028$, which is approximately $1.20$ times $\text{SE}(\hat{\theta}_{\mytext{BP}})$ and $0.98$ times $\text{SE}(\hat{\theta}_{\mytext{wBP}})$.

\begin{figure}[H]
\centering
\includegraphics[scale=0.52]{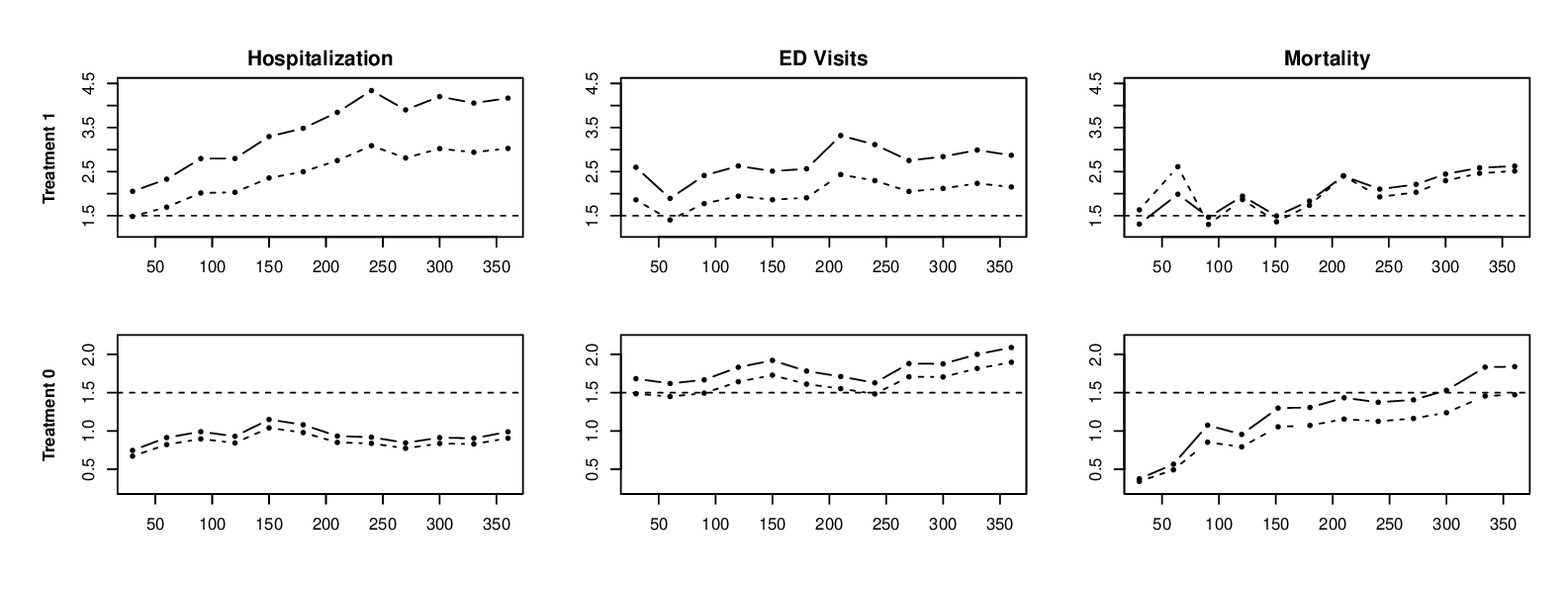} \vspace{-.48in}
\caption{Standardized difference between $\hat{S}_{ak, \mytext{KM}}$ and $\hat{S}_{ak, \mytext{wKM}}$ under low-dimensional setting. Solid line is standardized by $\text{SE}(\hat{S}_{ak, \mytext{KM}})$ and dashed line is standardized by $\text{SE}(\hat{S}_{ak, \mytext{wKM}})$. A horizontal line is placed at 1.5.} \label{fig:xeffect4svprob}
\end{figure}

\subsection{Additional tables and figures} \label{sec:additional-tables}
Table \ref{tb:covbal} shows the covariate balance results. The covariate balance results show that weighting based on low-dimensional PS models is effective. Both ML and CAL weighting substantially improve balance relative to the unweighted sample, with CAL achieving superior balance for main effects. Differences between CAL and ML mainly reflect their different objectives, explicit balancing versus likelihood-based model fitting. Moreover, extending PS model to include high-dimensional interaction terms does not yield additional gains in covariate balance for the interaction terms, suggesting that a main-effect only PS model may be correctly specified, because interaction terms appear to be approximately balanced once main effects are controlled for.

Figure \ref{fig:mlesvprob} presents the standardized absolute difference between $\hat{S}_{ak, \mytext{wKM}}$ in the low-dimensional setting and $\hat{S}_{ak, \mytext{wKM}}$ in the high-dimensional setting.
Figure \ref{fig:mlecalsvprob-high} presents the standardized absolute difference between $\hat{S}_{ak, \mytext{wKM}}$ (high-dimensional setting) and $\hat{S}_{ak, \mytext{RCAL,lin}}$.
Tables \ref{tb:hosp-tr1}--\ref{tb:mort-tr0} show numerical estimates and standard errors for the estimation of survival probabilities.
The conclusions drawn from these figures and tables agree with the findings discussed in the main paper.

\begin{table}[H]
\caption{Summary of covariate balance} \label{tb:covbal} \vspace{-.15in}
\begin{center}
 \small
  \renewcommand{\arraystretch}{0.8}
\begin{tabular}{lccc}
\hline
 & main & interaction & \# ImbanCov \\
\hline
Raw & (2.53e-4, 6.05e-1) & (5.68e-6, 5.51e-1) & (9, 148)\\
ML & (4.05e-5, 3.46e-2) & (1.49e-5, 7.55e-2) & (0, 0)\\
CAL & (2.86e-13, 3.49e-10) & (5.05e-6, 7.03e-2) & (0, 0) \\
RML & (3.46e-4, 4.25e-2) & (4.75e-6, 5.80e-2) & (0, 0)\\
RCAL & (5.14e-4, 6.88e-2) & (2.99e-5, 7.13e-2) & (0, 0) \\
\hline
\end{tabular}
\end{center}
\setlength{\baselineskip}{0.2\baselineskip}
\vspace{-.15in}\noindent{\tiny
\textbf{Note}: Each entry reports the range of relative differences across all main effects or interaction terms. 
For each covariate, the relative difference is the difference in (weighted) means between treatment groups divided by the pooled weighted standard deviation (see Section~\ref{sec:rel-dif} for a detailed definition).
Raw corresponds to unweighted differences. ML, CAL, RML, and RCAL correspond to differences weighted by propensity scores estimated using ML, CAL, RML, and RCAL methods, respectively. \# ImbanCov represents number of imbalance covariates of main effects and interaction terms, where we say it is imbalance if the difference exceeds 0.1.
}
\vspace{-.2in}
\end{table}

\begin{figure}
\centering
\includegraphics[scale=.52]{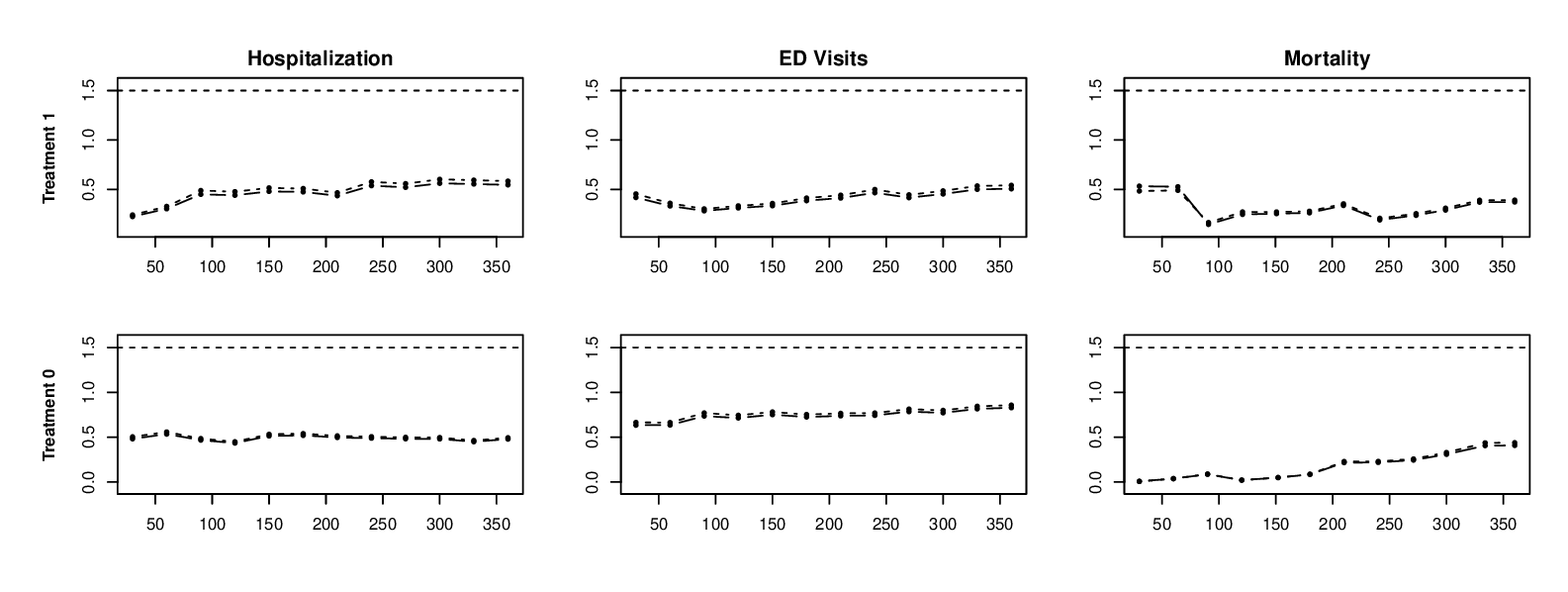} \vspace{-.48in}
\caption{Standardized absolute difference between $\hat{S}_{ak, \mytext{wKM}}$ in the low-dimensional setting and $\hat{S}_{ak, \mytext{wKM}}$ in the high-dimensional setting. Solid line is standardized by $\text{SE}(\hat{S}_{ak, \mytext{wKM}})$ under low-dimensional setting and dashed line is standardized by $\text{SE}(\hat{S}_{ak, \mytext{wKM}})$ under high-dimensional setting. A horizontal line is placed at 1.5.} \label{fig:mlesvprob}
\end{figure}

\begin{figure}
\centering
\includegraphics[scale=.52]{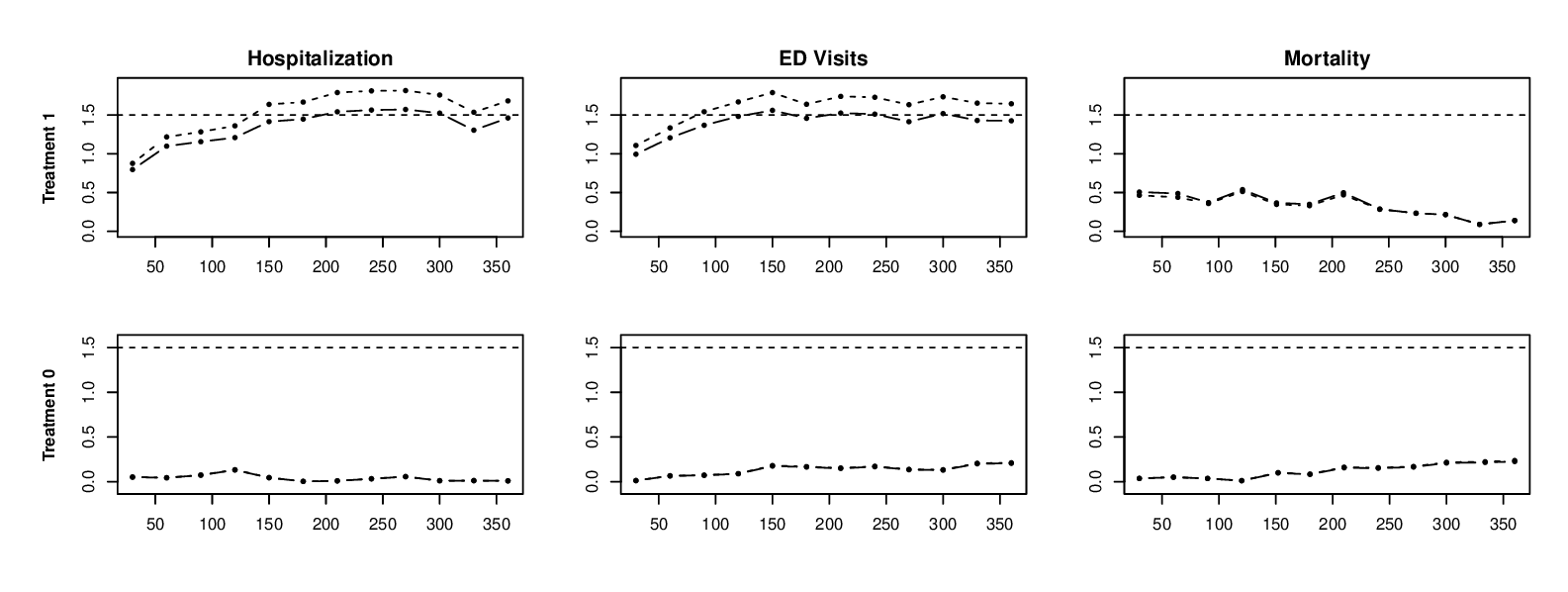} \vspace{-.48in}
\caption{Standardized absolute difference between $\hat{S}_{ak, \mytext{wKM}}$ (high-dimensional setting) and $\hat{S}_{ak, \mytext{RCAL,lin}}$. Solid line is standardized by $\text{SE}(\hat{S}_{ak, \mytext{wKM}})$ and dashed line is standardized by $\text{SE}(\hat{S}_{ak, \mytext{RCAL,lin}})$. A horizontal line is placed at 1.5.} \label{fig:mlecalsvprob-high}
\end{figure}

\begin{table} [H]
\caption{Summary for estimation of  $S_{1k}$ for Hospitalization outcome} \label{tb:hosp-tr1} \vspace{-.15in}
\begin{center}
 \small
  \renewcommand{\arraystretch}{0.8}
\resizebox{0.9\textwidth}{!}{\begin{tabular}{lccccccccccccc}
\hline
  \multirow{3}{*}{$u_k$} &  \multicolumn{6}{c}{Est} & &  \multicolumn{6}{c}{SE} \\
   \cline{2-7}\cline{9-14}
  & \multicolumn{3}{c}{Low} & & \multicolumn{2}{c}{High} & & \multicolumn{3}{c}{Low} & &  \multicolumn{2}{c}{High} \\
  \cline{2-4}\cline{6-7}\cline{9-11}\cline{13-14}
  & $\hat{S}_{1k, \mytext{KM}}$ & $\hat{S}_{1k, \mytext{wKM}}$ & $\hat{S}_{1k, \mytext{CAL,lin}}$ & & $\hat{S}_{1k, \mytext{wKM}}$ & $\hat{S}_{1k, \mytext{RCAL,lin}}$ &
  & $\hat{S}_{1k, \mytext{KM}}$ & $\hat{S}_{1k, \mytext{wKM}}$ & $\hat{S}_{1k, \mytext{CAL,lin}}$ & & $\hat{S}_{1k, \mytext{wKM}}$ & $\hat{S}_{1k, \mytext{RCAL,lin}}$ \\
  \hline
30  & 0.9476  & 0.9434  & 0.9461  & & 0.9441  & 0.9462  & & 0.0020  & 0.0028  & 0.0026  & & 0.0026  & 0.0024  \\
60  & 0.9188  & 0.9130  & 0.9171  & & 0.9140  & 0.9175  & & 0.0025  & 0.0035  & 0.0031  & & 0.0032  & 0.0029  \\
90  & 0.8981  & 0.8902  & 0.8953  & & 0.8920  & 0.8961  & & 0.0028  & 0.0039  & 0.0035  & & 0.0036  & 0.0032  \\
120  & 0.8796  & 0.8712  & 0.8765  & & 0.8730  & 0.8777  & & 0.0030  & 0.0042  & 0.0037  & & 0.0038  & 0.0034  \\
150  & 0.8655  & 0.8551  & 0.8606  & & 0.8572  & 0.8631  & & 0.0032  & 0.0044  & 0.0039  & & 0.0041  & 0.0036  \\
180  & 0.8527  & 0.8412  & 0.8467  & & 0.8433  & 0.8496  & & 0.0033  & 0.0046  & 0.0041  & & 0.0043  & 0.0037  \\
210  & 0.8403  & 0.8271  & 0.8329  & & 0.8292  & 0.8361  & & 0.0034  & 0.0048  & 0.0042  & & 0.0045  & 0.0039  \\
240  & 0.8288  & 0.8135  & 0.8192  & & 0.8162  & 0.8234  & & 0.0035  & 0.0050  & 0.0044  & & 0.0046  & 0.0040  \\
270  & 0.8178  & 0.8036  & 0.8094  & & 0.8062  & 0.8136  & & 0.0036  & 0.0050  & 0.0045  & & 0.0047  & 0.0041  \\
300  & 0.8077  & 0.7921  & 0.7983  & & 0.7950  & 0.8023  & & 0.0037  & 0.0052  & 0.0046  & & 0.0048  & 0.0042  \\
330  & 0.7974  & 0.7820  & 0.7885  & & 0.7849  & 0.7913  & & 0.0038  & 0.0053  & 0.0047  & & 0.0049  & 0.0042  \\
360  & 0.7884  & 0.7723  & 0.7790  & & 0.7752  & 0.7825  & & 0.0039  & 0.0053  & 0.0047  & & 0.0050  & 0.0043  \\
\hline
\end{tabular}}
\end{center}
\setlength{\baselineskip}{0.2\baselineskip}
\vspace{-.15in}\noindent{\tiny
\textbf{Note}: Est means point estimates, and SE means standard deviations. Low means that only main-effects are included in working models, and High means that both main-effects and interactions are included in working models.
}
\vspace{-.2in}
\end{table}

\begin{table} [H]
\caption{Summary for estimation of  $S_{0k}$ for Hospitalization outcome} \label{tb:hosp-tr0} \vspace{-.15in}
\begin{center}
 \small
  \renewcommand{\arraystretch}{0.8}
\resizebox{0.9\textwidth}{!}{\begin{tabular}{lccccccccccccc}
\hline
  \multirow{3}{*}{$u_k$} &  \multicolumn{6}{c}{Est} & &  \multicolumn{6}{c}{SE} \\
   \cline{2-7}\cline{9-14}
  & \multicolumn{3}{c}{Low} & & \multicolumn{2}{c}{High} & & \multicolumn{3}{c}{Low} & &  \multicolumn{2}{c}{High} \\
  \cline{2-4}\cline{6-7}\cline{9-11}\cline{13-14}
  & $\hat{S}_{0k, \mytext{KM}}$ & $\hat{S}_{0k, \mytext{wKM}}$ & $\hat{S}_{0k, \mytext{CAL,lin}}$ & & $\hat{S}_{0k, \mytext{wKM}}$ & $\hat{S}_{0k, \mytext{RCAL,lin}}$ &
  & $\hat{S}_{0k, \mytext{KM}}$ & $\hat{S}_{0k, \mytext{wKM}}$ & $\hat{S}_{0k, \mytext{CAL,lin}}$ & & $\hat{S}_{0k, \mytext{wKM}}$ & $\hat{S}_{0k, \mytext{RCAL,lin}}$ \\
  \hline
30  & 0.9610  & 0.9602  & 0.9609  & & 0.9608  & 0.9607  & & 0.0011  & 0.0012  & 0.0012  & & 0.0012  & 0.0012  \\
60  & 0.9380  & 0.9367  & 0.9377  & & 0.9375  & 0.9375  & & 0.0014  & 0.0015  & 0.0015  & & 0.0015  & 0.0014  \\
90  & 0.9192  & 0.9177  & 0.9189  & & 0.9185  & 0.9184  & & 0.0016  & 0.0017  & 0.0017  & & 0.0017  & 0.0016  \\
120  & 0.9037  & 0.9021  & 0.9035  & & 0.9030  & 0.9027  & & 0.0017  & 0.0019  & 0.0018  & & 0.0018  & 0.0018  \\
150  & 0.8892  & 0.8871  & 0.8887  & & 0.8881  & 0.8881  & & 0.0018  & 0.0020  & 0.0019  & & 0.0019  & 0.0019  \\
180  & 0.8771  & 0.8751  & 0.8767  & & 0.8762  & 0.8762  & & 0.0019  & 0.0021  & 0.0020  & & 0.0020  & 0.0020  \\
210  & 0.8667  & 0.8648  & 0.8665  & & 0.8659  & 0.8659  & & 0.0020  & 0.0022  & 0.0021  & & 0.0021  & 0.0020  \\
240  & 0.8564  & 0.8545  & 0.8563  & & 0.8556  & 0.8557  & & 0.0020  & 0.0022  & 0.0021  & & 0.0022  & 0.0021  \\
270  & 0.8471  & 0.8453  & 0.8472  & & 0.8464  & 0.8465  & & 0.0021  & 0.0023  & 0.0022  & & 0.0022  & 0.0022  \\
300  & 0.8373  & 0.8353  & 0.8371  & & 0.8364  & 0.8364  & & 0.0022  & 0.0024  & 0.0023  & & 0.0023  & 0.0022  \\
330  & 0.8286  & 0.8266  & 0.8285  & & 0.8277  & 0.8277  & & 0.0022  & 0.0024  & 0.0023  & & 0.0023  & 0.0023  \\
360  & 0.8210  & 0.8187  & 0.8207  & & 0.8199  & 0.8199  & & 0.0023  & 0.0025  & 0.0024  & & 0.0024  & 0.0023  \\
\hline
\end{tabular}}
\end{center}
\setlength{\baselineskip}{0.2\baselineskip}
\vspace{-.15in}\noindent{\tiny
\textbf{Note}: Est means point estimates, and SE means standard deviations. Low means that only main-effects are included in working models, and High means that both main-effects and interactions are included in working models.
}
\vspace{-.2in}
\end{table}

\begin{table} [H]
\caption{Summary for estimation of  $S_{1k}$ for ED visits outcome} \label{tb:ed-tr1} \vspace{-.15in}
\begin{center}
 \small
  \renewcommand{\arraystretch}{0.8}
\resizebox{0.9\textwidth}{!}{\begin{tabular}{lccccccccccccc}
\hline
  \multirow{3}{*}{$u_k$} &  \multicolumn{6}{c}{Est} & &  \multicolumn{6}{c}{SE} \\
   \cline{2-7}\cline{9-14}
  & \multicolumn{3}{c}{Low} & & \multicolumn{2}{c}{High} & & \multicolumn{3}{c}{Low} & &  \multicolumn{2}{c}{High} \\
  \cline{2-4}\cline{6-7}\cline{9-11}\cline{13-14}
  & $\hat{S}_{1k, \mytext{KM}}$ & $\hat{S}_{1k, \mytext{wKM}}$ & $\hat{S}_{1k, \mytext{CAL,lin}}$ & & $\hat{S}_{1k, \mytext{wKM}}$ & $\hat{S}_{1k, \mytext{RCAL,lin}}$ &
  & $\hat{S}_{1k, \mytext{KM}}$ & $\hat{S}_{1k, \mytext{wKM}}$ & $\hat{S}_{1k, \mytext{CAL,lin}}$ & & $\hat{S}_{1k, \mytext{wKM}}$ & $\hat{S}_{1k, \mytext{RCAL,lin}}$ \\
  \hline
30  & 0.9201  & 0.9137  & 0.9176  & & 0.9151  & 0.9183  & & 0.0025  & 0.0035  & 0.0031  & & 0.0032  & 0.0029  \\
60  & 0.8888  & 0.8833  & 0.8884  & & 0.8846  & 0.8890  & & 0.0029  & 0.0039  & 0.0035  & & 0.0036  & 0.0033  \\
90  & 0.8666  & 0.8590  & 0.8654  & & 0.8602  & 0.8657  & & 0.0031  & 0.0043  & 0.0038  & & 0.0040  & 0.0035  \\
120  & 0.8461  & 0.8374  & 0.8437  & & 0.8388  & 0.8450  & & 0.0033  & 0.0045  & 0.0041  & & 0.0042  & 0.0037  \\
150  & 0.8295  & 0.8207  & 0.8273  & & 0.8223  & 0.8292  & & 0.0035  & 0.0047  & 0.0043  & & 0.0044  & 0.0038  \\
180  & 0.8136  & 0.8043  & 0.8112  & & 0.8062  & 0.8129  & & 0.0036  & 0.0049  & 0.0044  & & 0.0046  & 0.0041  \\
210  & 0.7993  & 0.7869  & 0.7944  & & 0.7890  & 0.7962  & & 0.0037  & 0.0051  & 0.0045  & & 0.0048  & 0.0042  \\
240  & 0.7851  & 0.7732  & 0.7811  & & 0.7756  & 0.7829  & & 0.0038  & 0.0052  & 0.0046  & & 0.0049  & 0.0043  \\
270  & 0.7738  & 0.7630  & 0.7711  & & 0.7652  & 0.7722  & & 0.0039  & 0.0053  & 0.0047  & & 0.0049  & 0.0043  \\
300  & 0.7621  & 0.7507  & 0.7593  & & 0.7531  & 0.7608  & & 0.0040  & 0.0054  & 0.0048  & & 0.0050  & 0.0044  \\
330  & 0.7522  & 0.7400  & 0.7490  & & 0.7428  & 0.7501  & & 0.0041  & 0.0055  & 0.0049  & & 0.0051  & 0.0044  \\
360  & 0.7413  & 0.7294  & 0.7385  & & 0.7322  & 0.7396  & & 0.0042  & 0.0055  & 0.0050  & & 0.0052  & 0.0045  \\
\hline
\end{tabular}}
\end{center}
\setlength{\baselineskip}{0.2\baselineskip}
\vspace{-.15in}\noindent{\tiny
\textbf{Note}: Est means point estimates, and SE means standard deviations. Low means that only main-effects are included in working models, and High means that both main-effects and interactions are included in working models.
}
\vspace{-.2in}
\end{table}

\begin{table} [H]
\caption{Summary for estimation of  $S_{0k}$ for ED visits outcome} \label{tb:ed-tr0} \vspace{-.15in}
\begin{center}
 \small
  \renewcommand{\arraystretch}{0.8}
\resizebox{0.9\textwidth}{!}{\begin{tabular}{lccccccccccccc}
\hline
  \multirow{3}{*}{$u_k$} &  \multicolumn{6}{c}{Est} & &  \multicolumn{6}{c}{SE} \\
   \cline{2-7}\cline{9-14}
  & \multicolumn{3}{c}{Low} & & \multicolumn{2}{c}{High} & & \multicolumn{3}{c}{Low} & &  \multicolumn{2}{c}{High} \\
  \cline{2-4}\cline{6-7}\cline{9-11}\cline{13-14}
  & $\hat{S}_{0k, \mytext{KM}}$ & $\hat{S}_{0k, \mytext{wKM}}$ & $\hat{S}_{0k, \mytext{CAL,lin}}$ & & $\hat{S}_{0k, \mytext{wKM}}$ & $\hat{S}_{0k, \mytext{RCAL,lin}}$ &
  & $\hat{S}_{0k, \mytext{KM}}$ & $\hat{S}_{0k, \mytext{wKM}}$ & $\hat{S}_{0k, \mytext{CAL,lin}}$ & & $\hat{S}_{0k, \mytext{wKM}}$ & $\hat{S}_{0k, \mytext{RCAL,lin}}$ \\
  \hline
30  & 0.9475  & 0.9454  & 0.9464  & & 0.9463  & 0.9463  & & 0.0013  & 0.0014  & 0.0014  & & 0.0014  & 0.0013  \\
60  & 0.9187  & 0.9161  & 0.9175  & & 0.9173  & 0.9174  & & 0.0016  & 0.0017  & 0.0017  & & 0.0017  & 0.0016  \\
90  & 0.8967  & 0.8938  & 0.8955  & & 0.8952  & 0.8954  & & 0.0017  & 0.0019  & 0.0018  & & 0.0019  & 0.0018  \\
120  & 0.8775  & 0.8741  & 0.8759  & & 0.8756  & 0.8757  & & 0.0019  & 0.0021  & 0.0020  & & 0.0020  & 0.0020  \\
150  & 0.8606  & 0.8568  & 0.8589  & & 0.8585  & 0.8588  & & 0.0020  & 0.0022  & 0.0021  & & 0.0021  & 0.0021  \\
180  & 0.8448  & 0.8411  & 0.8432  & & 0.8428  & 0.8432  & & 0.0021  & 0.0023  & 0.0022  & & 0.0022  & 0.0022  \\
210  & 0.8318  & 0.8281  & 0.8303  & & 0.8298  & 0.8302  & & 0.0022  & 0.0024  & 0.0023  & & 0.0023  & 0.0022  \\
240  & 0.8193  & 0.8157  & 0.8180  & & 0.8175  & 0.8179  & & 0.0022  & 0.0025  & 0.0023  & & 0.0024  & 0.0023  \\
270  & 0.8085  & 0.8042  & 0.8066  & & 0.8062  & 0.8065  & & 0.0023  & 0.0025  & 0.0024  & & 0.0024  & 0.0024  \\
300  & 0.7976  & 0.7932  & 0.7957  & & 0.7952  & 0.7955  & & 0.0023  & 0.0026  & 0.0025  & & 0.0025  & 0.0024  \\
330  & 0.7873  & 0.7825  & 0.7851  & & 0.7847  & 0.7852  & & 0.0024  & 0.0026  & 0.0025  & & 0.0026  & 0.0025  \\
360  & 0.7784  & 0.7733  & 0.7760  & & 0.7756  & 0.7761  & & 0.0024  & 0.0027  & 0.0026  & & 0.0026  & 0.0025  \\
\hline
\end{tabular}}
\end{center}
\setlength{\baselineskip}{0.2\baselineskip}
\vspace{-.15in}\noindent{\tiny
\textbf{Note}: Est means point estimates, and SE means standard deviations. Low means that only main-effects are included in working models, and High means that both main-effects and interactions are included in working models.
}
\vspace{-.2in}
\end{table}

\begin{table} [H]
\caption{Summary for estimation of  $S_{1k}$ for Mortality outcome} \label{tb:mort-tr1} \vspace{-.15in}
\begin{center}
 \small
  \renewcommand{\arraystretch}{0.8}
\resizebox{0.9\textwidth}{!}{\begin{tabular}{lccccccccccccc}
\hline
  \multirow{3}{*}{$u_k$} &  \multicolumn{6}{c}{Est} & &  \multicolumn{6}{c}{SE} \\
   \cline{2-7}\cline{9-14}
  & \multicolumn{3}{c}{Low} & & \multicolumn{2}{c}{High} & & \multicolumn{3}{c}{Low} & &  \multicolumn{2}{c}{High} \\
  \cline{2-4}\cline{6-7}\cline{9-11}\cline{13-14}
  & $\hat{S}_{1k, \mytext{KM}}$ & $\hat{S}_{1k, \mytext{wKM}}$ & $\hat{S}_{1k, \mytext{CAL,lin}}$ & & $\hat{S}_{1k, \mytext{wKM}}$ & $\hat{S}_{1k, \mytext{RCAL,lin}}$ &
  & $\hat{S}_{1k, \mytext{KM}}$ & $\hat{S}_{1k, \mytext{wKM}}$ & $\hat{S}_{1k, \mytext{CAL,lin}}$ & & $\hat{S}_{1k, \mytext{wKM}}$ & $\hat{S}_{1k, \mytext{RCAL,lin}}$ \\
  \hline
30  & 0.9978  & 0.9984  & 0.9984  & & 0.9982  & 0.9980  & & 4e-04  & 3e-04  & 3e-04  & & 4e-04  & 4e-04  \\
64  & 0.9963  & 0.9974  & 0.9974  & & 0.9972  & 0.9970  & & 6e-04  & 4e-04  & 4e-04  & & 5e-04  & 5e-04  \\
91  & 0.9952  & 0.9961  & 0.9961  & & 0.9960  & 0.9958  & & 6e-04  & 7e-04  & 7e-04  & & 6e-04  & 7e-04  \\
121  & 0.9937  & 0.9951  & 0.9950  & & 0.9949  & 0.9946  & & 7e-04  & 8e-04  & 7e-04  & & 7e-04  & 7e-04  \\
151  & 0.9923  & 0.9935  & 0.9933  & & 0.9933  & 0.9930  & & 8e-04  & 9e-04  & 8e-04  & & 8e-04  & 9e-04  \\
180  & 0.9908  & 0.9924  & 0.9922  & & 0.9922  & 0.9919  & & 9e-04  & 9e-04  & 9e-04  & & 9e-04  & 9e-04  \\
210  & 0.9889  & 0.9912  & 0.9909  & & 0.9909  & 0.9904  & & 0.0010  & 0.0010  & 9e-04  & & 9e-04  & 0.0010  \\
242  & 0.9874  & 0.9896  & 0.9896  & & 0.9894  & 0.9891  & & 0.0010  & 0.0011  & 0.0010  & & 0.0010  & 0.0011  \\
274  & 0.9863  & 0.9887  & 0.9888  & & 0.9884  & 0.9882  & & 0.0011  & 0.0012  & 0.0011  & & 0.0011  & 0.0011  \\
300  & 0.9850  & 0.9878  & 0.9879  & & 0.9875  & 0.9872  & & 0.0011  & 0.0012  & 0.0011  & & 0.0011  & 0.0011  \\
330  & 0.9833  & 0.9864  & 0.9865  & & 0.9859  & 0.9858  & & 0.0012  & 0.0013  & 0.0011  & & 0.0012  & 0.0012  \\
361  & 0.9820  & 0.9852  & 0.9853  & & 0.9848  & 0.9846  & & 0.0013  & 0.0013  & 0.0012  & & 0.0012  & 0.0012  \\
\hline
\end{tabular}}
\end{center}
\setlength{\baselineskip}{0.2\baselineskip}
\vspace{-.15in}\noindent{\tiny
\textbf{Note}: Est means point estimates, and SE means standard deviations. Low means that only main-effects are included in working models, and High means that both main-effects and interactions are included in working models.
}
\vspace{-.2in}
\end{table}

\begin{table} [H]
\caption{Summary for estimation of  $S_{0k}$ for Mortality  outcome} \label{tb:mort-tr0} \vspace{-.15in}
\begin{center}
 \small
  \renewcommand{\arraystretch}{0.8}
\resizebox{0.9\textwidth}{!}{\begin{tabular}{lccccccccccccc}
\hline
  \multirow{3}{*}{$u_k$} &  \multicolumn{6}{c}{Est} & &  \multicolumn{6}{c}{SE} \\
   \cline{2-7}\cline{9-14}
  & \multicolumn{3}{c}{Low} & & \multicolumn{2}{c}{High} & & \multicolumn{3}{c}{Low} & &  \multicolumn{2}{c}{High} \\
  \cline{2-4}\cline{6-7}\cline{9-11}\cline{13-14}
  & $\hat{S}_{0k, \mytext{KM}}$ & $\hat{S}_{0k, \mytext{wKM}}$ & $\hat{S}_{0k, \mytext{CAL,lin}}$ & & $\hat{S}_{0k, \mytext{wKM}}$ & $\hat{S}_{0k, \mytext{RCAL,lin}}$ &
  & $\hat{S}_{0k, \mytext{KM}}$ & $\hat{S}_{0k, \mytext{wKM}}$ & $\hat{S}_{0k, \mytext{CAL,lin}}$ & & $\hat{S}_{0k, \mytext{wKM}}$ & $\hat{S}_{0k, \mytext{RCAL,lin}}$ \\
  \hline
30  & 0.9989  & 0.9988  & 0.9988  & & 0.9988  & 0.9988  & & 2e-04  & 2e-04  & 2e-04  & & 2e-04  & 2e-04  \\
60  & 0.9978  & 0.9976  & 0.9977  & & 0.9976  & 0.9977  & & 3e-04  & 3e-04  & 3e-04  & & 3e-04  & 3e-04  \\
90  & 0.9970  & 0.9967  & 0.9968  & & 0.9967  & 0.9967  & & 3e-04  & 4e-04  & 4e-04  & & 4e-04  & 4e-04  \\
120  & 0.9960  & 0.9957  & 0.9958  & & 0.9957  & 0.9957  & & 4e-04  & 4e-04  & 4e-04  & & 4e-04  & 4e-04  \\
152  & 0.9951  & 0.9945  & 0.9947  & & 0.9946  & 0.9946  & & 4e-04  & 5e-04  & 5e-04  & & 5e-04  & 5e-04  \\
180  & 0.9941  & 0.9935  & 0.9937  & & 0.9936  & 0.9936  & & 4e-04  & 5e-04  & 5e-04  & & 5e-04  & 5e-04  \\
210  & 0.9931  & 0.9924  & 0.9926  & & 0.9926  & 0.9927  & & 5e-04  & 6e-04  & 5e-04  & & 6e-04  & 5e-04  \\
240  & 0.9923  & 0.9916  & 0.9917  & & 0.9917  & 0.9918  & & 5e-04  & 6e-04  & 6e-04  & & 6e-04  & 6e-04  \\
271  & 0.9914  & 0.9907  & 0.9908  & & 0.9908  & 0.9909  & & 5e-04  & 6e-04  & 6e-04  & & 6e-04  & 6e-04  \\
300  & 0.9906  & 0.9898  & 0.9900  & & 0.9900  & 0.9901  & & 6e-04  & 7e-04  & 6e-04  & & 6e-04  & 6e-04  \\
334  & 0.9898  & 0.9888  & 0.9890  & & 0.9891  & 0.9892  & & 6e-04  & 7e-04  & 7e-04  & & 7e-04  & 7e-04  \\
360  & 0.9891  & 0.9880  & 0.9882  & & 0.9883  & 0.9885  & & 6e-04  & 8e-04  & 7e-04  & & 7e-04  & 7e-04  \\
\hline
\end{tabular}}
\end{center}
\setlength{\baselineskip}{0.2\baselineskip}
\vspace{-.15in}\noindent{\tiny
\textbf{Note}: Est means point estimates, and SE means standard deviations. Low means that only main-effects are included in working models, and High means that both main-effects and interactions are included in working models.
}
\vspace{-.2in}
\end{table}

\section{Definition of relative difference} \label{sec:rel-dif}
For each covariate $j$, the weighted mean in treatment group $a \in \{0,1\}$ is
\begin{align*}
\bar{X}^{(1)}_{j} = \frac{\sum_{i=1}^n A_i \hat{\omega}_{1i}X_{ij}}{\sum_{i=1}^n A_i\hat{\omega}_{1i}}, \quad \bar{X}^{(0)}_{j} = \frac{\sum_{i=1}^n (1 - A_i)\hat{\omega}_{0i}X_{ij}}{\sum_{i=1}^n (1 - A_i)\hat{\omega}_{0i}},
\end{align*}
where the weights are defined as $\hat{\omega}_{1i} = \pi(X_i; \hat{\gamma})^{-1}$ and  $\hat{\omega}_{0i} = (1 - \pi(X_i; \hat{\gamma}))^{-1}$, and $\hat{\omega}_{1i} = \hat{\omega}_{0i} = 1$ for unweighted differences.

The difference for covariate $j$ is $D_j = | \bar{X}^{(1)}_{j}   - \bar{X}^{(0)}_{j} |$ and the pooled weighted standard deviation is
\begin{align*}
S_j = \sqrt{ \frac{1}{2}\left(\frac{\sum_{i=1}^n A_i \hat{\omega}_{1i}(X_{ij} - \bar{X}^{(1)}_{j})^2}{\sum_{i=1}^n A_i \hat{\omega}_{1i} - \frac{\sum_{i=1}^n A_i \hat{\omega}_{1i}^2}{\sum_{i=1}^n A_i\hat{\omega}_{1i}}} + \frac{\sum_{i=1}^n (1 - A_i)\hat{\omega}_{0i}(X_{ij} - \bar{X}^{(0)}_{j})^2}{\sum_{i=1}^n (1 - A_i) \hat{\omega}_{0i} - \frac{\sum_{i=1}^n (1 - A_i) \hat{\omega}_{0i}^2}{\sum_{i=1}^n (1 - A_i)\hat{\omega}_{0i}}}\right)}.
\end{align*}
The relative difference is then $R_j = D_j / S_j$ and the range of relative differences across all covariates is $[ \min_j R_j, \max_j R_j ]$.

\vspace{.3in}
\centerline{\bf\Large References}
\begin{description}\addtolength{\itemsep}{-.15in}

\item Ghosh, S. and Tan, Z. (2022) Doubly robust semiparametric inference using regularized calibrated estimation with high-dimensional data, {\it Bernoulli}, 28, 1675–1703.

\item Tan, Z. (2020) Model-assisted inference for treatment effects using regularized calibrated estimation with high-dimensional data, {\it Annals of Statistics}, 48, 811–837.

\item Tan, Z. (2023) Consistent and robust inference in hazard probability and odds models with discrete-time survival data, {\it Lifetime Data Analysis}, 29, 555–584.

\end{description}

\end{document}